\newif\ifmapx
{\catcode`/=0 \catcode`\\=12/gdef/mkillslash\#1{#1}}
\edef\jobnametmp{\expandafter\string\csname isit14_paper_apx\endcsname}
\edef\jobnameapx{\expandafter\mkillslash\jobnametmp}
\edef\jobnameexpand{\jobname}
\ifx\jobnameexpand\jobnameapx
\mapxtrue
\else
\mapxfalse
\fi

\documentclass[12pt, onecolumn]{IEEEtran}

\usepackage{epsfig,latexsym,amsfonts,amsmath,amssymb,verbatim,cite,mathrsfs,amsthm}
\usepackage{graphicx}
\usepackage{color}
\usepackage{epstopdf}
\usepackage{multirow}
\usepackage{url}
\usepackage{dsfont} % \mathds{1} gives indicator function

\newtheorem{theorem}{Theorem}
\newtheorem{proposition}[theorem]{Proposition}
\newtheorem{lemma}[theorem]{Lemma}

\theoremstyle{definition}
\newtheorem{defn}{Definition}

\newtheorem{remark}{Remark}

\newcommand{\Var}{\mathrm{Var}}
\newcommand{\E}{\mathbb{E}}
\newcommand{\Cov}{\mathrm{Cov}}
\newcommand{\tr}{\mathrm{tr}}

\newcommand{\inprod}[2]{{\langle #1, #2 \rangle}}

\newcommand{\matc}{\ensuremath{\mathcal{C}}}
\newcommand{\matx}{\ensuremath{\mathcal{X}}}

\newcommand{\maty}{\ensuremath{\mathcal{Y}}}
\newcommand{\matn}{\ensuremath{\mathcal{N}}}

\newcommand{\matp}{\ensuremath{\mathcal{P}}}

\newcommand{\mreals}{\ensuremath{\mathbb{R}}}

\ifx\eqref\undefined
    \newcommand{\eqref}[1]{~(\ref{#1})}
\fi
\ifx\mod\undefined
    \def\mod{\mathop{\rm mod}}
\fi

\newcommand{\mat}[1]{#1}
\newcommand{\vecc}[1]{#1}

\def\exp{\mathop{\rm exp}}

\def\tr{\mathop{\rm tr}}
\def\rank{\mathop{\rm rank}}
\def\cov{\mathop{\rm cov}}

\def\EE{\mathbb{E}\,}

\def\Var{\mathrm{Var}}
\def\Cov{\mathrm{Cov}}

\def\PP{\mathbb{P}}

\def\eqdef{\stackrel{\triangle}{=}}
\def\eqdist{\stackrel{d}{=}}

\def\dperp{\perp\!\!\!\perp}

\def\unifto{\mathop{{\mskip 3mu plus 2mu minus 1mu%
    \setbox0=\hbox{$\mathchar"3221$}%
    \raise.6ex\copy0\kern-\wd0%
    \lower0.5ex\hbox{$\mathchar"3221$}}\mskip 3mu plus 2mu minus 1mu}}

\ifx\lesssim\undefined
\def\simleq{{{\mskip 3mu plus 2mu minus 1mu%
    \setbox0=\hbox{$\mathchar"013C$}%
    \raise.2ex\copy0\kern-\wd0%
    \lower0.9ex\hbox{$\mathchar"0218$}}\mskip 3mu plus 2mu minus 1mu}}
\else
\def\simleq{\lesssim}
\fi

\ifx\gtrsim\undefined
\def\simgeq{{{\mskip 3mu plus 2mu minus 1mu%
    \setbox0=\hbox{$\mathchar"013E$}%
    \raise.2ex\copy0\kern-\wd0%
    \lower0.9ex\hbox{$\mathchar"0218$}}\mskip 3mu plus 2mu minus 1mu}}
\else
\def\simgeq{\gtrsim}
\fi

\begin{document}

%%%%%%%%%%%%%%%%%%%%%%%%%%%%%%%%%%%%%%%%%%%%%%%%%%%%%%%%%%%%%%%%%%
%%%%%  TITLE AND ABSTRACT %%%%%%%%%%%%%%%%%%%%%%%%%%%%%%%%%%%%%%%%
%%%%%%%%%%%%%%%%%%%%%%%%%%%%%%%%%%%%%%%%%%%%%%%%%%%%%%%%%%%%%%%%%%

\title{Coherent multiple-antenna block-fading channels at finite blocklength}
\author{Austin Collins and Yury Polyanskiy
\thanks{Authors are with the Department
    of Electrical Engineering and Computer Science, MIT, Cambridge, MA 02139 USA.
    \mbox{e-mail:~{\ttfamily \{austinc,yp\}@mit.edu}.}}%
\thanks{
This material is based upon work supported by the National Science
    Foundation CAREER award under grant agreement CCF-12-53205, by the NSF grant CCF-17-17842
and by the Center for Science of Information (CSoI),
an NSF Science and Technology Center, under grant agreement CCF-09-39370.
}%
}

\maketitle

\thispagestyle{plain}
\pagestyle{plain}

\begin{abstract}
In this paper we consider a channel model that is often used to describe the mobile wireless scenario: multiple-antenna
additive white Gaussian noise channels subject to random (fading) gain with full channel state information at the
receiver.  Dynamics of the fading process are approximated by a piecewise-constant process (frequency non-selective
isotropic block fading).  This work addresses the finite blocklength fundamental limits of this channel model.
Specifically, we give a formula for the channel dispersion -- a quantity governing the delay required to achieve
capacity. The multiplicative nature of the fading disturbance leads to a number of interesting technical difficulties that
required us to enhance traditional methods for finding the channel dispersion. Alas, one difficulty remains: 
the converse (impossibility) part of our result holds under an extra constraint on the
growth of the peak-power with blocklength. 

Our results demonstrate, for example, that while capacities of $n_t\times n_r$ and $n_r \times
n_t$ antenna configurations coincide (under fixed received power), the coding delay can be sensitive to this
switch. For example, at the received SNR of $20$ dB the $16\times 100$ system achieves capacity with codes of length (delay)
which is only $60\%$ of the length required for the $100\times 16$ system. Another
interesting implication is that for the MISO channel, the dispersion-optimal coding schemes require employing orthogonal
designs such as Alamouti's scheme -- a surprising observation considering the fact that Alamouti's scheme was designed
for reducing demodulation errors, not improving coding rate. Finding these dispersion-optimal coding schemes naturally gives 
a criteria for producing orthogonal design-like inputs in dimensions where orthogonal designs do not exist.
\end{abstract}

%%%%%%%%%%%%%%%%%%%%%%%%%%%%%%%%%%%%%%%%%%%%%%%%%%%%%%%%%%%%%%%%%%
%%%%% INTRODUCTION %%%%%%%%%%%%%%%%%%%%%%%%%%%%%%%%%%%%%%%%%%%%%%%
%%%%%%%%%%%%%%%%%%%%%%%%%%%%%%%%%%%%%%%%%%%%%%%%%%%%%%%%%%%%%%%%%%

\section{Introduction}

Given a noisy communication channel, the maximal cardinality of a codebook of
blocklength $n$ which can be decoded with block error probability  no greater than
$\epsilon$ is denoted as $M^*(n,\epsilon)$. The evaluation of this function -- the fundamental performance limit 
of block coding -- is alas computationally impossible for most channels of interest. As a resolution of this
difficulty~\cite{PPV08} proposed a closed-form normal approximation, based on the asymptotic expansion:
\begin{equation}\label{eq:dispersion}
    \log M^*(n,\epsilon) = nC -\sqrt{nV} Q^{-1}(\epsilon) + O(\log n)\,,
\end{equation}
where the capacity $C$ and dispersion $V$ are two intrinsic characteristics of the channel and $Q^{-1}(\epsilon)$ is the inverse of the
$Q$-function\footnote{As usual, 
$ Q(x) = \int_x^\infty {1\over \sqrt{2\pi}} e^{-t^2/2} \,dt\,.$}. 
One immediate consequence of the normal approximation is an estimate for
the minimal blocklength (delay) required to achieve a given fraction $\eta$ of the channel capacity:
\begin{equation}\label{eq:minblock}
     n \simgeq \left({Q^{-1}(\epsilon)\over 1-\eta}\right)^2 {V\over C^2}\,.
\end{equation}
Asymptotic expansions such as~\eqref{eq:dispersion} are
rooted in the central-limit theorem and have been known classically for discrete memoryless
channels~\cite{RD61,VS64} and later extended in a wide variety of directions; see the surveys in~\cite{PV13-isit,tan2014asymptotic}. 

The fading channel is the centerpiece of the theory and practice of wireless communication, and hence there are 
many slightly different variations of the model: differing assumptions on the dynamics and distribution of the
fading process, antenna configurations, and channel state knowledge.  
The capacity of the fading channel was found independently by Telatar~\cite{ET99} and Foschini and
Gans~\cite{foschini1998limits} for the case of Rayleigh fading and channel state information available at the receiver
only (CSIR) and at both the transmitter and receiver (CSIRT).  Motivated by the linear gains promised by capacity
results, space time codes were introduced to exploit multiple antennas, most notable amongst them is Alamouti's
ingenious orthogonal scheme~\cite{alamouti1998simple} along with a generalization of Tarokh, Jafarkhani and Calderbank~\cite{tarokh1999space}.
Motivated by a recent surge of orthogonal frequency division (OFDM) technology, this paper focuses on an isotropic channel
gain distribution, which is piecewise independent (``block-fading'') and assume full channel state information available
at the receiver (CSIR). This work describes finite blocklength effects incurred by the fading on the fundamental
communication limits. 

Some of the prior work on similar questions is as follows. Single antenna channel dispersion was computed in~\cite{polyanskiy11-08a} for a more
general stationary channel gain process with memory. In~\cite{yang12-09} finite-blocklength effects  are explored for the
non-coherent block fading setup. Quasi-static fading  channels in the general MIMO setting have been thoroughly
investigated  in~\cite{YDKP13-qsmimo}, showing that the expansion~\eqref{eq:dispersion}  changes dramatically (in
particular the channel dispersion term becomes    zero); see also~\cite{molavianJazi13-10} for evaluation of the bounds.
Coherent
quasi-static channel has been studied in the limit of infinitely  many antennas in~\cite{hoydis13} appealing to
concentration properties of  random matrices. Dispersion for lattices (infinite constellations) in      fading channels
has been investigated in a sequence of works,              see~\cite{vituri2013dispersion} and references. Note also
that there are   some very fine differences between stationary and block-fading channel     models, cf.~\cite[Section
4]{AL03}.  The minimum energy to send $k$ bits over a MIMO channel for both the coherent and non-coherent case was
studied  in~\cite{yang2015minimum}, showing the latter requires orders of magnitude larger latencies. \cite{yang2014optimum} investigates the problem of power control with an average
power constraint on the codebook in the    quasi-static fading channel with perfect CSIRT. A novel achievability bound
was found and evaluated for the fading channel with CSIR in~\cite{YAGP16-bb}. Parts of this work have previously appeared
in~\cite{collins2014miso,CP16-mimobf-isit}.

The paper is organized as follows. In
Section~\ref{sec:main_res} we describe the channel model and state all our main results formally. 
Section~\ref{sec:prelims} characterizes 
capacity achieving input/output distributions (caid/caod, resp.) and evaluates moments of the information density. Then
in Sections~\ref{sec:ach} and~\ref{sec:conv} we prove the achievability and converse parts of our (non rank-1) results,
respectively. Section~\ref{sec:MISO} focuses on the special case of when the matrix of channel gains has rank 1.
Finally, Section~\ref{sec:discussion} contains a discussion of numerical results and the behavior of channel dispersion as a function of the number of antennas.

The numerical software used to compute the achievability bounds, dispersion and normal approximation in this work can be found online under the Spectre project~\cite{Spectre-github}.

%%%%%%%%%%%%%%%%%%%%%%%%%%%%%%%%%%%%%%%%%%%%%%%%%%%%%%%%%%%%%%%%%%
%%%%% CHANNEL MODEL %%%%%%%%%%%%%%%%%%%%%%%%%%%%%%%%%%%%%%%%%%%%%%
%%%%%%%%%%%%%%%%%%%%%%%%%%%%%%%%%%%%%%%%%%%%%%%%%%%%%%%%%%%%%%%%%%

\section{Main Results}\label{sec:main_res}

\subsection{Channel Model}

The channel model considered in this paper is the frequency-nonselective coherent real block fading (BF) discrete-time
channel with multiple transmit and receive antennas (MIMO) (See~\cite[Section II]{BPS98} for background on this model).
We will simply refer to it as the \textit{MIMO-BF channel}, which we formally define here.
 Given parameters $n_t,n_r,P,T$ as follows: let $n_t\ge 1$ be the number of transmit antennas, $n_r\ge 1$ be the number of receive antennas, and $T\ge 1$ be the coherence time of the channel. The input-output relation at block $j$ (spanning time instants $(j-1)T+1$ to $jT$) with $j=1,\ldots,n$ is given by
\begin{equation}\label{eq:channel}
        \mat{Y}_j = \mat{H}_j \mat{X}_j + \mat{Z}_j\,,
\end{equation}
where $\{\mat{H}_j, j=1,\ldots\}$ is a $n_r \times n_t$ matrix-valued random fading process, $\mat{X}_j$ is a $n_t
\times T$ matrix channel input, $\mat{Z}_j$ is a $n_r \times T$ Gaussian random real-valued matrix with independent
entries of zero mean and unit variance, and $\mat{Y}_j$ is the $n_r \times T$ matrix-valued channel output. The process
$\mat{H}_j$ is assumed to be i.i.d. with isotropic distribution $P_{\mat{H}}$, i.e. for any orthogonal matrices $U \in
\mathbb{R}^{n_r\times n_r}$ and $V \in \mathbb{R}^{n_t\times n_t}$, both $U\mat{H}$ and $\mat{H}V$ are equal in
distribution to $\mat{H}$.  We also assume 
\begin{equation}\label{eq:hnorm}
        \PP[H \neq 0] > 0
\end{equation}
to avoid trivialities.  Note that due to merging channel inputs at time instants $1,\ldots, T$ into one matrix-input, the block-fading channel becomes memoryless. We assume
coherent demodulation so that the channel state information (CSI) $\mat{H}_j$ is fully known to the receiver (CSIR).

An $(nT, M, \epsilon, P)_{CSIR}$ code of blocklength $nT$, probability of error $\epsilon$ and power-constraint $P$ is a pair of maps: the
encoder $f:[M]\to (\mreals^{n_t\times T})^n$ and the decoder $g:(\mreals^{n_r \times T})^n \times (\mreals^{n_r \times n_t})^n \to [M]$ satisfying the probability of error constraint
\begin{equation}\label{eq:perr}
        \PP[W \neq \hat W] \le \epsilon\,. 
\end{equation}  
on the probability space 
$$ W \to \mat{X}^n \to (\mat{Y}^n, \mat{H}^n) \to \hat W\,,$$
where the message $W$ is uniformly distributed on $[M]$, $\mat{X}^n = f(W)$, $\mat{X}^n\to(\mat{Y}^n, \mat{H}^n)$ is as
described in~\eqref{eq:channel}, and $\hat W=g(\mat{Y}^n,\mat{H}^n)$.  In addition the input sequences are required to
satisfy the power constraint:
$$ \sum_{j=1}^{n} \|\mat{X}_j\|_F^2 \le n TP \qquad \PP\mbox{-a.s.}\,,$$
where $\|M\|_F^2 \eqdef \sum_{i,j} M_{i,j}^2$ is the Frobenius norm of the matrix $M$.

Under the isotropy assumption on $P_H$, the capacity $C$ appearing in~\eqref{eq:dispersion} of this channel is given
by~\cite{ET99}
\begin{align}
C(P) &= \frac{1}{2}\E\left[ \log\det\left(I_{n_r} +                        \frac{P}{n_t}\mat{H}\mat{H}^T\right)\right] \label{eq:mimo_capacity}\\
	&= \sum_{i=1}^{n_{\min}} \E\left[C_{AWGN}\left(\frac{P}{n_t}\Lambda_i^2\right) \right]\,,\label{eq:capacity}
\end{align}
where $C_{AWGN}(P) = {1\over 2} \log (1 + P)$ is the capacity of the additive white Gaussian noise (AWGN) channel with SNR
$P$, $n_{\min} = \min(n_r,n_t)$ is the minimum of the transmit and receive antennas, and $\{\Lambda_i^2, i=1,\ldots,
n_{\min}\}$ are eigenvalues of $\mat H \mat H^T$. Note that it is common to think that as $P\to \infty$ the
expression~\eqref{eq:capacity} scales as $n_{\min} \log P$, but this is only true if $\PP[\rank \mat H = n_{\min}]=1$.

The goal of this line of work is to characterize the dispersion of the present channel. Since the channel is
memoryless it is natural to expect, given the results in~\cite{PPV08,polyanskiy11-08a}, that dispersion (for $\epsilon < 1/2$) is given by
\begin{equation}\label{eq:vmindef}
V(P) \eqdef \inf_{P_{\mat X}: I(\mat{X};\mat{Y}|\mat{H})=C} {1\over T}   \E\left[\Var(i(\mat{X};
\mat{Y},\mat{H}) | \mat{X})\right]
\end{equation}
where we denoted (single $T$-block) information density by
\begin{equation}\label{eq:infodens}
	i(x; y,h) \eqdef \log {dP_{\mat{Y},\mat{H}|\mat{X}=x} \over dP^*_{\mat Y,\mat H}}(y,h)  
\end{equation}and $P_{\mat Y,\mat H}^*$
is the capacity achieving output distribution (caod). Justification of~\eqref{eq:vmindef} as the actual (operational)
dispersion, appearing in the expansion of $\log M^*(n, \epsilon)$ is by no means trivial and is the subject of this
work.

%%%%%%%%%%%%%%%%%%%%%%%%%%%%%%%%%%%%%%%%%%%%%%%%%%%%%%%%%%%%%%%%%%
%%%%% MAIN RESULTS %%%%%%%%%%%%%%%%%%%%%%%%%%%%%%%%%%%%%%%%%%%%%%%
%%%%%%%%%%%%%%%%%%%%%%%%%%%%%%%%%%%%%%%%%%%%%%%%%%%%%%%%%%%%%%%%%%

\subsection{Statement of Main Theorems}

Here we formally state the main results, then go into more detail in the following sections.  Our first result is an achievability and partial converse bound for the MIMO-BF fading channel for fixed parameters $n_t,n_r,T,P$.

\begin{theorem}\label{thm:mimo_dispersion}
For the MIMO-BF channel, there exists an $(nT,M,\epsilon,P)_{CSIR}$        maximal probability of error code with $0 < \epsilon < 1/2$ satisfying
\begin{align}
\log M \geq nTC(P) - \sqrt{nTV(P)}Q^{-1}(\epsilon) +         o(\sqrt{n})\ .
\end{align}
Furthermore, for any $\delta_n\to0$ there exists $\delta'_n\to 0$ so that  every $(nT, M, \epsilon, P)_{CSIR}$ code with
extra constraint  that $\max_j\|x^j\|_F \leq \delta_n n^{1/4}$, must satisfy
\begin{align}\label{eq:thm1e1}
\log M \leq nTC(P) - \sqrt{nTV(P)}Q^{-1}(\epsilon) + \delta'_n \sqrt{n}
\end{align}
where the capacity $C(P)$ is given by~\eqref{eq:mimo_capacity} and dispersion $V(P)$ by~\eqref{eq:vmindef}.\footnote{For the
explicit expression for $i(x; y, h)$ see~\eqref{eq:mimo_info_density} below.}
\end{theorem}
\begin{proof} This follows from Theorem~\ref{thm:mainach} and Theorem~\ref{thm:mainconv} below. \end{proof}

\begin{remark}
Note that the converse has an extra constraint $\max_j \|x^j\|_F \leq \delta_n n^{1/4}$. Mathematically,  this constraint is needed so that the $n$-fold information information density $i(x^n;Y^n,H^n)$ behaves Gaussian-like, via the Berry-Esseen theorem.  For example, if $x^n$ had $x_{11} = \sqrt{nTP}$ and zeroes in all other coordinates, then one term in the information density would be $O(n)$ while the rest would be $O(1)$, and hence no asymptotic structure would emerge.  All known bounds to obtain the channel dispersion rely on approximating the information density by a Gaussian, and hence a fundamentally different method of analysis is needed to handle the situation where $\max_j \|x^j\|_F \geq \delta_n n^{1/4}$.

Note that to violate this constraint, a significant portion of the power budget must be poured into a single coherent
block, which 1) creates a very large peak-to-average power ratio (PAPR) -- an illegal (for regulating bodies) or
impractical (for power amplifiers) situation, and 2) does a poor job of exploiting the
diversity gain from coding over multiple independent coherent blocks.  Therefore, our converse results are sufficient
from the point of view of any practical system.

In addition, the random codebook used for the achievability (uniform on the power sphere) can be expurgated with
a rate loss of $-\delta_n^2 n^{-\tfrac{1}{2}}$ so that it entirely consists of codewords satisfying $\max_j \|x_j\|_F
\leq \delta_n n^{1/4}$. This is easiest to see by noticing that a standard Gaussian vector $Z^n$ satisfies
$\PP[\|Z^n\|_\infty > \delta_n n^{1/4}] \le e^{-O(\delta_n^2 \sqrt{n})}$. This observation shows that our analysis of the
random coding bound (with spherical codebook) is tight in terms of the dispersion term.
\end{remark}

\begin{remark}
The remainder term $o(\sqrt{n})$ in~\eqref{eq:thm1e1} depends on the
system parameters $(n_t,n_r, T, P_H)$ in a complicated way, which we do not attempt to study here.
\end{remark}

The behavior of dispersion found in Theorem \ref{thm:mimo_dispersion} turns out to depend crucially on whether $\rank(H) \leq 1$ a.s. or not.  When $\rank(H) > 1$, all capacity achieving input distributions (caids) yield the same conditional variance \eqref{eq:vmindef}, yet when $\rank(H) \leq 1$, the conditional variance varies over the set of caids.  The following theorem discusses the case where $\PP[\rank \mat H > 1]>0$. In this case, the dispersion~\eqref{eq:vmindef} can be calculated for the simplest Telatar caid (i.i.d. Gaussian matrix $\mat X$). The following theorem gives full details.

\begin{theorem}\label{thm:mimo_disp_expression}
Assume that $\PP[\rank \mat H > 1]>0$, then $V(P)=V_{iid}(P)$, where
\begin{align}
V_{iid}(P) =\ &T\Var\left(\sum_{i=1}^{n_{\min}}                                  C_{AWGN}\left(\frac{P}{n_t} \Lambda_i^2\right)\right) \nonumber\\
&+ \sum_{i=1}^{n_{\min}}\E\left[V_{AWGN}\left(\frac{P}{n_t}\Lambda_i^2\right)\right]\nonumber\\
&+ \left(\frac{P}{n_t}\right)^2 \left(\eta_1 - \frac{\eta_2}{n_t}\right)
\label{eq:mimo_disp_expression}
\end{align}
where $\{\Lambda_i^2, i=1,\ldots,n_{\min}\}$ are eigenvalues of $\mat{H}\mat{H}^T$, $V_{AWGN}(P) = \frac{\log^2
e}{2}\left(1 - \frac{1}{\left( 1 + P \right)^2}\right)$, and 
\begin{align}
c(\sigma) &  \triangleq {\sigma\over 1 + \frac{P}{n_t}\sigma}\label{eq:cfdef}\\
\eta_1 &\eqdef \frac{\log^2e}{2} \sum_{i=1}^{n_{\min}} 
		\E\left[c^2(\Lambda_i^2)\right]\label{eq:eta1def}\\
\eta_2 &\eqdef \frac{\log^2 e}{2} \left( \sum_{i=1}^{n_{\min}} \E\left[c(\Lambda_i^2)\right] \right)^2\label{eq:eta2def}
\end{align}
\end{theorem}
\begin{proof} This is proved in Proposition~\ref{prop:vminrank2} below.
\end{proof}

\begin{remark}
Each of the three terms in~\eqref{eq:mimo_disp_expression} is non-negative, see Remark~\ref{rmk:v1positive} below for more details.
\end{remark}

In the case where the fading process has rank 1 (e.g. for MISO systems), there are a multitude of caids, and the minimization problem
in~\eqref{eq:vmindef} is non-trivial.  Quite surprisingly, for some values of $n_t,T$, we show that the (essentially
unique) minimizer is a full-rate orthogonal design. The latter were introduced into the field of communications by
Alamouti~\cite{alamouti1998simple} and Tarokh et al~\cite{tarokh1999space}. This shows a somewhat unexpected connection
between schemes optimal from modulation-theoretic and information-theoretic points of view. The precise results are as
follows.

\begin{theorem}\label{thm:rank1_disp}
When $\PP[\text{rank}(H) \le 1] = 1$, we have
\begin{align}
V(P) &= T\Var\left( C_{AWGN}\left(\frac{P}{n_t}\Lambda^2\right)\right)
+ \E\left[ V_{AWGN}\left(\frac{P}{n_t}\Lambda^2\right)\right]\\
&+ \left(\frac{P}{n_t}\right)^2 \left( \eta_1 -
\frac{\eta_2}{n_t^2T} v^*(n_t, T)\right)
\end{align}
where $\Lambda^2$ is the non-zero eigenvalues of $HH^T$, and
\begin{align}\label{eq:v_star_minimization}
v^*(n_t,T) = \frac{n_t^2}{2P^2}\max_{P_X : I(\mat{X};\mat{Y},\vecc{H}) = C}
\Var(\|\mat{X}\|_F^2)
\end{align}
\end{theorem}
\begin{proof} This is the content of Proposition~\ref{prop:disp_calc} below.\end{proof}

%\begin{theorem}\label{thm:MISO_disp}
%For the MISO-BF channel, (\textbf{TODO:} This is rank-1 case, not MISO, correct? What is $V_1$? What is $V$?)
%\begin{align}
%V = V_1(n_t,1,T,P) - \frac{2\chi_2}{n_t^2T}\left(\frac{P}{n_t}\right)^2v^*(n_t,T)\,,
%\end{align}
%where
%\begin{align}\label{eq:v_star_minimization}
%v^*(n_t,T) = \frac{1}{2\left(P/n_t\right)^2}\min_{P_X : I(\mat{X};\mat{Y},\vec{H}) = C} \Var(\|\mat{X}\|_F^2)
%\end{align}
%\end{theorem}

The quantity $v^*(n_t,T)$ is defined separately in Theorem \ref{thm:rank1_disp} because it isolates how the dispersion depends on the input distribution.  Unfortunately, $v^*(n_t,T)$ is generally unknown, since the maximization in~\eqref{eq:v_star_minimization}
is over a manifold of matrix-valued random variables.  However, for many dimensions, the maximum can be found by
invoking the Hurwitz-Radon theorem~\cite{JR22}. We state this below to introduce the notation, and expand on it in
Section~\ref{sec:MISO}.

\begin{theorem}[Hurwitz-Radon]\label{thm:radon-hurwitz}
There exists a family of $n\times n$ real matrices $V_1,\hdots, V_k$ satisfying
\begin{align}\label{eq:hr1}
&V_i^T V_i = I_n  \qquad\qquad i = 1,\hdots, k\\
&V_i^TV_j + V_j^T V_i = 0 \quad i\not= j \label{eq:hr2}
\end{align}
if and only if $k \leq \rho(n)$, where
\begin{align}\label{eq:rho_funct}
\rho(2^ab) = 8 \left\lfloor {a\over 4} \right\rfloor + 2^{a\, \mathrm{mod}\, 4},\qquad  a,b\in\mathbb{Z}, \mbox{b--odd}\,.
\end{align}
In particular, $\rho(n)\le n$ and $\rho(n)=n$ only for $n=1,2,4,8$.
\end{theorem}

For a concrete example, note that Alamouti's scheme is created from a Hurwitz-Radon family for $n = k = 2$.  Indeed, take the matrices
\begin{align*}
V_1 = I_2, \quad
V_2 = 
\left[ \begin{array}{cc}
0 & 1\\
-1 & 0\\
\end{array} \right],
\end{align*}
then Alamouti's orthogonal design can be formed by taking $a V_1 + b V_2$. It turns out that ``maximal'' Hurwitz-Radon families give capacity achieving input distributions for the MIMO-BF channel, see Proposition \ref{prop:od_to_caid} for the details.

The following theorem summarizes our current knowledge of $v^*(n_t,T)$.
\begin{theorem}\label{th:vstar} For any pair of positive integers $n_t, T$ we have
\begin{equation}\label{eq:collins_sum}
    v^*(T, n_t) = v^*(n_t,T) \leq n_tT\min(n_t,T)\,.
\end{equation}
If $n_t \leq \rho(T)$ or $T\le \rho(n_t)$ then a full-rate orthogonal design is dispersion-optimal and
\begin{equation}\label{eq:collins_tight}
		v^*(n_t,T) = n_tT\min(n_t,T)\,.
\end{equation}	
If instead $n_t > \rho(T)$ and $T > \rho(n_t)$ then for a jointly-Gaussian capacity-achieving input $X$ we
have\footnote{So that in these cases the bound~\eqref{eq:collins_sum} is either
non-tight, or is achieved by a non-jointly-Gaussian caid.}
	\begin{equation}\label{eq:collins_nontight}
		{n_t^2\over 2P^2} \Var(\|X\|_F^2) < n_tT\min(n_t,T)\,. 
\end{equation}	
Finally, if $n_t \le T$ and~\eqref{eq:collins_tight} holds, then $v^*(n_t',T)=n_t'^2 T$ for any $n_t' \le n_t$ (and
similarly with the roles of $n_t$ and $T$ switched).  
\end{theorem}

Note that the $\rho(n)$ function is monotonic in even values of $n$ (and
is $1$ for $n$ odd), and $\rho(n) \to \infty$ along even $n$. Therefore, for any number of transmit antennas $n_t$, there is a large enough $T$ such that $n_t \leq
\rho(T)$, in which case an $n_t \times T$ full rate orthogonal design achieves the optimal
$v^*(n_t,T)$.\\

\section{Preliminary results}\label{sec:prelims}

The section gives some results that will be useful for the achievability and converse proofs (Theorem \ref{thm:mainach} and Theorem \ref{thm:mainconv}, respectively), along with generally aiding our understanding of the MIMO-BF channel at finite blocklength.  The results in this section and where they are used is summarized as follows:
\begin{itemize}
\item Theorem \ref{thm:caid_conds} gives a characterization of the caids for MIMO-BF channel. While all caids give the
same capacity (by definition), when the channel matrix is rank 1, they do not all yield the same dispersion.  This
characterization is needed to reason about the minimizers in \eqref{eq:vmindef}, especially in the rank 1 case.
\item Proposition \ref{prop:cond_var_x} computes variance $V_n(x^n)$ of information density conditioned on the channel input $x^n$.  A key characteristic of the fading channel is that $V_n(x^n)$ varies as $x^n$ moves around the input space, which does not happen in DMC's or the AWGN channel.  This variation in $V_n(x^n)$ poses additional challenges in the converse proof, where we partition the codebook based on thresholding $V_n(x^n)$ (see the proof of Theorem \ref{thm:mainconv} for details).  Knowledge of $V_n(x^n)$ will also allow us to understand when the information density can be well approximated by a Gaussian (see Lemma \ref{lem:thirdmom}).
\item Propositions \ref{prop:vminrank2} and \ref{prop:disp_calc} explicitly give the expression for the dispersion found
from the achievability and converse proofs for the $\rank(H) > 1$ and $\rank(H) \le 1$ case, respectively.  These expressions show how the dispersion depends on $n_t, n_r, T, P$, and are the contents of Theorems \ref{thm:mimo_disp_expression} and \ref{thm:rank1_disp} above.
%\item Lemma \ref{lem:thirdmom} computes the Berry Esseen ratio, which is needed to apply the Berry Esseen theorem in both the achievablility and converse proofs.  This lemma relies on the assumption $\max_j\|x_j\|_F \le \delta n^{1\over 4}$, and also uses the structure of $V_n(x^n)$ found in Proposition \ref{prop:cond_var_x}.
\end{itemize}

\subsection{Known results: capacity and capacity achieving output distribution}
First we review a few known results on the MIMO-BF channel.  Since the channel is memoryless, the capacity is given by
\begin{align}\label{eq:mimo_capacity_opt}
C &= \frac{1}{T}\max_{P_{\mat X}: \E[\|\mat{X}\|_F^2] \leq TP} I(\mat{X};\mat{Y},\mat{H})\,.
\end{align}
It was shown by Telatar~\cite{ET99} that whenever distribution of $\mat{H}$ is isotropic, the input $\mat{X} \in \mathbb{R}^{n_t \times T}$ with entry $i,j$ given by
\begin{equation}\label{eq:telatar}
	X_{i,j} \stackrel{iid}{\sim} \matn\left(0, {P\over n_t}\right)\,,
\end{equation}
is a maximizer, resulting in the capacity
formula~\eqref{eq:mimo_capacity}. The distribution induced by a caid at the channel output $(\mat{Y}, \mat{H})$ is
called the capacity achieving output distribution (caod). A classical fact is that, while
there may be many caids, the caod is unique, e.g.~\cite[Section 4.4]{PW2016notes}. Thus, from \eqref{eq:telatar} we infer that the caod is given by 
\begin{align} P_{\mat{Y},\mat{H}}^* &\eqdef P_{\mat{H}} P^*_{\mat{Y}|\mat{H}}\,,\label{eq:mimo_caod}\\
   P^*_{\mat{Y}|\mat{H}} &\eqdef \prod_{j=1}^T P_{Y^{(j)}|\mat{H}}^*\,,\\
   P_{Y^{(j)}|\mat{H}=h}^* &\eqdef \mathcal{N}\left( 0,I_{n_r} + \frac{P}{n_t}hh^T\right)\,,
\end{align}

$Y=[Y^{(1)}, \ldots, Y^{(T)}]$, where $Y^{(j)}$ is $j$-th column of $Y$, which, as we specified in (3), is a $n_r \times T$ matrix.

%where $\mat{Y}=[Y_1 \cdots Y_T]$ is decomposed into $T$ columns $Y_i\in\mreals^{n_r\times 1}$.

\subsection{Capacity achieving input distributions}\label{sec:caids}

A key feature of the MIMO-BF channel is that it has many caids, whereas many commonly studied channels (e.g. BSC, BEC, AWGN) have a unique caid.  Understanding the set of distributions that achieve capacity is essential for reasoning about the minimizer of the condition variance in~\eqref{eq:vmindef}.  The following theorem characterizes the set of caids for the MIMO-BF channel. Somewhat surprisingly, for the case of rank-1 $\mat{H}$ (e.g. for MISO) there are multiple non-trivial jointly Gaussian caids with different correlation structures. For example, space-time block codes can achieve the capacity in the rank 1 case, but do not achieve capacity when the rank is 2 or greater e.g.~\cite{SSPA00}.
%Later, this characterization will be used to show that dispersion-minimizing caid in Theorem~\ref{thm:mimo_dispersion} is given by orthogonal designs (such as Alamouti's coding), for dimensions when those exist.

\begin{theorem} \label{thm:caid_conds}\
\begin{enumerate}
\item Every caid $\mat X$ satisfies 
	\begin{equation}\label{eq:pcc_1}
		\forall a\in \mreals^{n_t}, b\in \mreals^T: \qquad \sum_{i=1}^{n_t} \sum_{j=1}^T a_i b_j X_{i,j} \sim
		\matn\left(0, {P\over n_t}\|a\|_2^2 \|b\|_2^2\right)\,.
	\end{equation}	
	If $\PP[\rank \mat{H}\le 1]=1$ then condition~\eqref{eq:pcc_1} is also sufficient for $\mat{X}$ to be caid.
\item Let $\mat{X} = \begin{pmatrix}R_1\cr \cdots\cr R_{n_t}\end{pmatrix}$ be decomposed into rows $R_i$. If $\mat{X}$ is a caid, then each $R_i\sim
\matn(0,{P\over n_t} I_T)$ (i.i.d. Gaussian) and 
\begin{align} \E[R_i^T R_i] &= \frac{P}{n_t}I_T, &i=1,\hdots,n_t\label{eq:row1}\\
	    \E[R_i^T R_j] &= - \E[R_j^T R_i], & i \not= j\,.  \label{eq:row2}
\end{align}
If $\mat{X}$ is jointly zero-mean Gaussian and $\PP[\rank \mat{H}\le 1]=1$, then~\eqref{eq:row1}-\eqref{eq:row2} are sufficient for $\mat{X}$ to be caid.
\item Let $\mat{X} = (C_1 \ldots C_T)$ be decomposed into columns $C_j$. If $\mat{X}$ is a caid, then each $C_j\sim
\matn(0,{P\over n_t} I_{n_t})$ (i.i.d. Gaussian) and 
\begin{align}
    \E[C_iC_i^T] &= \frac{P}{n_t}I_{n_t}, &i=1,\hdots,T\label{eq:col1}\\
    \E[C_i C_j^T] &= - \E[C_j C_i^T], & i \not= j\,. \label{eq:col2}
\end{align}
If $\mat{X}$ is jointly zero-mean Gaussian and $\PP[\rank \mat{H}\le 1]=1$, then~\eqref{eq:col1}-\eqref{eq:col2} are sufficient for $\mat{X}$ to be caid.

\item When $\PP[\rank \mat{H}>1]>0$, any caid has pairwise independent rows:
\begin{equation}\label{eq:rk2_caidrows}
	R_i \dperp R_j \sim \matn\left(0, {P\over n_t} I_{T}\right) \qquad \forall i\neq j
\end{equation}
and in particular
	\begin{equation}\label{eq:rk2_caid}
		X_{i,j} \dperp X_{k,l} \qquad \forall (i,j) \neq (k,l)\,. 
	\end{equation}	
	Therefore, among jointly Gaussian $\mat X$ the i.i.d. $X_{i,j}$ is the unique caid. 
\item There exist non-Gaussian caids if and only if $\PP[\rank \mat{H} \ge \min(n_t,T)] = 0$. 
\end{enumerate}
\end{theorem}

\begin{remark}(Special case of rank-1 $H$) In the MISO case when $n_t>1$ and $n_r=1$ (or more generally, $\rank \mat
H\le 1$ a.s.), there is not only a multitude of caids, but in fact they can have non-trivial correlations between 
entries of $\mat X$ (and this is ruled out by~\eqref{eq:rk2_caid} for all other cases). As an example, for the
$n_t=T=2$ case, any of the following random matrix-inputs $X$ (parameterized by $\rho\in[-1,1]$) is a Gaussian caid:
\begin{align}\label{eq:2x2_gsn_caids}
X = 
\sqrt{\frac{P}{2}} \left[ \begin{array}{cc}
\xi_1 & -\rho \xi_2 + \sqrt{1-\rho^2}\xi_3\\
\xi_2 & \rho \xi_1 + \sqrt{1-\rho^2}\xi_4
\end{array} \right] \,,
\end{align}
where $\xi_1,\xi_2,\xi_3,\xi_4 \sim \matn(0,1)$ i.i.d..  In particular, there are caids for which not all entries of $\mat
X$ are pairwise independent. 
\end{remark}
\begin{remark}Another way to state conditions~\eqref{eq:row1}-\eqref{eq:row2} is: all elements in a row (resp.
column) are
pairwise independent $\sim \mathcal{N}(0,\frac{P}{n_t})$ and each $2\times 2$ minor has antipodal correlation for the two
diagonals. In particular, if $\mat{X}$ is a caid, then $\mat{X}^T$ and any submatrix of $\mat{X}$ are caids too (for different $n_t$ and
$T$).
\end{remark}

\begin{proof} 
We will rely repeatedly on the following observations:
\begin{enumerate}
\item if $A,B$ are two random vectors in $\mreals^n$ then for any $v\in\mreals^n$ we have
	\begin{equation}\label{eq:pcc_0a}
		\forall v\in\mreals^n: v^T A \eqdist  v^T B \quad \iff \quad A \eqdist B \,.
\end{equation}	
	This is easy to show by computing characteristic functions.
\item If $A,B$ are two random vectors in $\mreals^n$ independent of $Z \sim \matn(0, I_n)$, then
	\begin{equation}\label{eq:pcc_0b}
		A+Z \eqdist B+Z \quad \iff \quad A\eqdist B\,.
\end{equation}	
	This follows from the fact that the characteristic function of $Z$ is nowhere zero.
\item For two matrices $Q_1,Q_2 \in \mreals^{n\times n}$ we have
	\begin{equation}\label{eq:pcc_0c}
		\forall v\in\mreals^n: v^T Q_1 v = v^T Q_2 v \quad \iff \quad Q_1+Q_1^T=Q_2+Q_2^T\,. 
\end{equation}	
	This follows from the fact that a quadratic form that is zero everywhere on $\mreals^n$ must have all coefficients
	equal to zero.
\end{enumerate}

Part 1 (necessity). Recall that the caod is unique and given by~\eqref{eq:mimo_caod}. Thus an input $\mat{X}$ is a caid iff for
$P_{\mat{H}}$-almost every $h_0 \in \mreals^{n_r\times n_t}$ we have
\begin{equation}\label{eq:pcc_2}
	h_0 \mat{X} + \mat{Z} \eqdist h_0 \mat{G} + \mat{Z}\,,
\end{equation}
where $\mat{G}$ is an $n_t \times T$ matrix with i.i.d. $\matn(0,P/n_t)$ entries (for sufficiency, just write $I(\mat X; \mat Y,
\mat H) = h(\mat Y|\mat H) - h(\mat Z)$ with $h(\cdot)$ denoting differential entropy). We will argue next that \eqref{eq:pcc_2} implies (under
isotropy assumption on $P_{\mat H}$) that 
\begin{equation}\label{eq:pcc_3}
	\forall a\in\mreals^{n_t}: \quad a^T \mat{X} \eqdist a^T \mat{G}\,. 
\end{equation}
From~\eqref{eq:pcc_0a},~\eqref{eq:pcc_3} is equivalent to $\sum_{i,j} a_i b_j X_{i,j} \eqdist \sum_{i,j} a_i b_j G_{i,j}$ for all $b \in \mathbb{R}^{n_t}$.

Let $E_0$ be a $P_H$-almost sure subset of $\mathbb{R}^{n_t\times n_r}$ for which~\eqref{eq:pcc_2} holds. Let $O(n)=\{U\in \mreals^{n\times n}:
U^T U = U U^T = I_n\}$ denote the group of orthogonal matrices, with the topology inherited from $\mreals^{n\times n}$. 
Let $\{U_k\}$ and $\{V_k\}$ for
$k\in\{1,2,\hdots\}$ be countable dense subsets of $O(n_t)$ and $O(n_r)$, respectively. (These exist since
$\mreals^{n^2}$ is a second-countable topological space).  
By isotropy of $P_{\mat{H}}$ we have $P_{\mat{H}}[ U_k(E_0)V_l] = 1$ and therefore
\begin{align}
E \eqdef E_0 \cap \bigcap_{k=1,l=1}^\infty U_k(E_0)V_l
\end{align}
is also almost sure: $P_{\mat H}[E]=1$, since $E$ is the intersection of countably many almost sure sets.  Here, $U_k(E_0)$ denotes the image of $E_0$ under $U_k$. By assumption (4), $E$ must contain a non-zero element $h_0$, for otherwise we would
have $P_H[0] = 1$, contradicting (4). Consequently, $h_0 \in U_k (E_0) V_l$ for all $k,l$, and so $U_k^{-1} h_0 V_l^{-1} \in E_0$ for all $k,l$.  Since for $U \in O(n)$, the map $U \mapsto U^{-1}$ is a bijective continuous transformation of $O(n)$, we have that $\{U_k^{-1}\}$ and $\{V_l^{-1}\}$ are also countable dense subsets of $O(n_t)$ and $O(n_r)$, respectively.  From~\eqref{eq:pcc_0b} and \eqref{eq:pcc_2} along with the definition of $E_0$, we conclude that
$$ U_k^{-1} h_0 V_l^{-1} \mat X \eqdist U_k^{-1} h_0 V_l^{-1} \mat G \qquad \forall k,l\,.$$
Arguing by continuity and using the density of $\{U_k^{-1}\}$ and $\{V_l^{-1}\}$, this implies also 
\begin{equation}\label{eq:pcc_4}
	U h_0 V \mat X \eqdist U  h_0 V \mat G \qquad \forall U \in O(n_t),  V\in O(n_r)\,.
\end{equation}
In particular, for any $a\in \mreals^{n_t}$ there must exist a choice of $U,V$ such that $Uh_0V$ has the top row equal
to $c_0a^T$ for some constant $c_0 > 0$. Choosing these $U,V$ in~\eqref{eq:pcc_4} and comparing distributions of top rows, we
conclude~\eqref{eq:pcc_3} after scaling by $1/c_0$.

Part 1 (sufficiency). Suppose $\PP[\rank \mat{H}\le 1]=1$. Then our goal is to show that~\eqref{eq:pcc_3} implies that
$\mat X$ is a caid. To that end, it is sufficient to show $h_0 \mat X \eqdist h_0 \mat G$ for all rank-1 $h_0$.
In the special case 
$$ h_0 = \begin{pmatrix} a^T\\0\\\vdots \\0 \end{pmatrix}\,,$$
the claim follows directly from~\eqref{eq:pcc_3}. Every other
rank-1 $h_0'$ can be decomposed as $h_0' = U h_0$ for some matrix $U$, and thus again we get $U h_0 \mat{X} \eqdist U
h_0 \mat {G}$, concluding the proof.

Parts 2 and 3 (necessity). From part 1 we have that for every $a,b$ we must have $a^T \mat X b \sim \matn(0, \|a\|_2^2
\|b\|_2^2 {P\over n_t})$. 
Computing expected square we get
\begin{equation}\label{eq:pcc_5}
	\EE[(a^T \mat X b)^2] = {P\over n_t} \left(\sum_i a_i^2\right)\left(\sum_j b_j^2\right)\,.
\end{equation}
Thus, expressing the left-hand side in terms of rows $R_i$ as $a^T \mat X = \sum_i a_i R_i$ we get 
$$ b^T \left\{ \EE\left[\left(\sum_i a_{i} R_{i}\right)^T  \left(\sum_i a_{i} R_{i}\right) \right] \right\} b = b^T \left(\sum_i a_i^2 I_T\right) b\,,$$
and thus by~\eqref{eq:pcc_0c} we conclude that for all $a$:
$$ \EE\left[\left(\sum_i a_{i} R_{i}\right)^T  \left(\sum_i a_{i} R_{i}\right)\right] = \left(\sum_i a_i^2 \right)I_T\,.$$
Each entry of the $T\times T$ matrices above is a quadratic form in $a$ and thus again by~\eqref{eq:pcc_0c} we
conclude~\eqref{eq:row1}-\eqref{eq:row2}. Part 3 is argued similarly with roles of $a$ and $b$ interchanged.

Parts 2 and 3 (sufficiency). When $\mat H$ is (at most) rank-1, we have from part 1
that it is sufficient to show that $a^T \mat X b \sim \matn(0, \|a\|_2^2
\|b\|_2^2 {P\over n_t})$. When $\mat X$ is jointly zero-mean Gaussian, we have $a^T X \mat b$ is zero-mean Gaussian and so
we only need to check its second moment satisfies~\eqref{eq:pcc_5}. But as we just argued,~\eqref{eq:pcc_5} is equivalent
to either~\eqref{eq:row1}-\eqref{eq:row2} or~\eqref{eq:col1}-\eqref{eq:col2}.

Part 4. As in Part 1, there must exist $h_0 \in \mreals^{n_r \times n_t}$ such that~\eqref{eq:pcc_4} holds and $\rank
h_0>1$. Thus, by choosing $U,V$ we can diagonalize $h_0$ and thus we conclude any pair of rows $R_i,R_j$ must be
independent. 

Part 5. This part is never used in subsequent parts of the paper, so we only sketch the argument and move the most
technical part of the proof to Appendix~\ref{apx:nongauss}. Let $\ell = \max \{r: \PP[\rank \mat H \ge r] > 0\}$. Then
arguing as for~\eqref{eq:pcc_4} we conclude that $\mat X$ is a caid if and only if for any $h$ with $\rank h\le \ell$ we have 
$$ h \mat{X} \eqdist h \mat{G}\,.$$
In other words, we have
\begin{equation}\label{eq:pcc_6}
	\sum_{i,j} a_{i,j} X_{i,j} \eqdist \sum_{i,j} G_{i,j} \qquad \forall a\in\mreals^{n_t\times T}: \rank a \le
\ell\,.
\end{equation}
If $\ell = \min(n_t, T)$, then rank condition on $a$ is not active and hence, we conclude by~\eqref{eq:pcc_0a} that
$\mat X\eqdist \mat G$. So assume $\ell < \min(n_t,T)$. Note that~\eqref{eq:pcc_6} is equivalent to the condition on
characteristic function of $\mat X$ as follows: 
\begin{equation}\label{eq:pcc_7}
	\EE\left[ e^{i \sum_{i,j} a_{i,j} X_{i,j}} \right] = e^{-{P\over 2 n_t} \sum_{i,j} a_{i,j}^2} \qquad \forall a: \rank
a \le \ell\,.
\end{equation}
It is easy to find polynomial (in $a_{i,j}$) that vanishes on all matrices of rank $\le \ell$ (e.g. take the product of all $\ell\times \ell$ minors). Then Proposition~\ref{prop:nongauss} in
Appendix~\ref{apx:nongauss} constructs non-Gaussian $\mat X$ satisfying~\eqref{eq:pcc_7} and hence~\eqref{eq:pcc_6}.
\end{proof}

\subsection{Information density and its moments}

In finite blocklength analysis, a key object of study is the information density, along with its first and second moments.  In this section we'll find expressions for these moments, along with showing when the information density is asymptotically normal.

It will be convenient to assume that the matrix $\mat{H}$ is represented as
\begin{equation}\label{eq:hdecomp}
	\mat{H} = U\Lambda V^T\,,
\end{equation}
where $U,V$ are uniformly distributed on $O(n_r)$ and $O(n_t)$ (which follows from the isotropic assumption on $H$), respectively,\footnote{Recall that $O(m)=\{A \in
\mreals^{m\times m}: A A^T = A^T A = I_m\}$ is the space of all orthogonal matrices. This space is compact in a natural
topology and admits a Haar probability measure.} and $\Lambda$ is the $n_r \times n_t$ diagonal matrix with diagonal entries
$\{\Lambda_i, i=1,\ldots,n_{\min}\}$. Joint distribution of $\{\Lambda_i\}$ depends on the
fading model. It does not matter for our analysis whether $\Lambda_i$'s are sorted in some way, or
permutation-invariant.

For the MIMO-BF channel, let $P_{\mat{Y}\mat{H}}^*$ denote the caod~\eqref{eq:mimo_caod}.  To compute the information density with
respect to $P_{\mat{Y}\mat{H}}^*$ (for a single $T$-block of symbols) as defined in~\eqref{eq:infodens}, denote $y = hx
+ z$ and write an SVD decomposition for matrix $h$ as 
$$ h = u \lambda v^T\,,$$
where $u \in O(n_r)$, $v\in O(n_t)$ and $\lambda$ is an $n_r \times n_t$ matrix which is zero except for the diagonal
entries, which are equal to $\lambda_1,\ldots,\lambda_{n_{\min}}$. Note that this representation is unique up to
permutation of $\{\lambda_j\}$, but the choice of this permutation will not affect any of the expressions below. With
this decomposition we have:
\begin{align}\label{eq:mimo_info_density}
i(x;y,h) &\eqdef \frac{T}{2} \log\det\left(I_{n_r} + \frac{P}{n_t}hh^T\right)\nonumber\\
&+ \frac{\log e}{2}\sum_{j=1}^{n_{\min}}\frac{\lambda_j^2\|v_j^T x\|^2 + 2\lambda_j\inprod{v_j^T x}{\tilde{z}_j} - \frac{P}{n_t}\lambda_j^2\|\tilde{z}_j\|^2}{1+\frac{P}{n_t} \lambda_j^2}
\end{align}
where we denoted by $v_j$ the $j$-th column of
$V$, and have set  $\tilde z = u^T z$, with
$\tilde{z}_j$ representing the $j$-th row of $\tilde{z}$.
The definition naturally extends to blocks of length $nT$ additively:
\begin{equation}\label{eq:rbg1}
	i(x^n; y^n, h^n) \eqdef \sum_{j=1}^n i(x_j; y_j, h_j)\,.
\end{equation}
We compute the (conditional) mean of information density to get
\begin{align} D_n(x^n) &\eqdef {1\over n T} \EE[i(X^n; Y^n,H^n)|X^n=x^n] \label{eq:dndef}\\
	&= C(P) + {\sqrt{\eta_2\over 2} \over n_t n T} \sum_{j=1}^n (\|x_j\|_F^2 - TP)\,,\label{eq:dndef2}
\end{align}	
where we used the following simple fact:
\begin{lemma}\label{lem:var_w}
Let $U \in \mathbb{R}^{1\times n_t}$ be uniformly distributed on the unit sphere, and $x \in \mathbb{R}^{n_t\times T}$ be a fixed matrix, then
\begin{align}
\E[\|Ux\|^2] = \frac{\|x\|_F^2}{n_t}
\end{align}
\end{lemma}
\begin{proof} Note that by additivity of $\|Ux\|^2$ across columns, it is sufficient to consider the case $T=1$, for
which the statement is clear from symmetry.
\end{proof}

\begin{remark}\label{rem:exp_hx} A simple consequence of Lemma~\ref{lem:var_w} is $\E[\|Hx\|_F^2] = \E[\|H\|_F^2]\frac{\|x\|_F^2}{n_t}$, which follows from considering the SVD of $H$.
\end{remark}

\begin{proposition}\label{prop:cond_var_x}
Let $V_n(x^n) \eqdef \frac{1}{nT}\Var(i(X^n;\mat{Y}^n,\mat{H}^n) | \mat X^n = x^n)$, then we have
\begin{equation}\label{eq:cvx_nto1}
	V_n(x^n) = {1\over n} \sum_{j=1}^n V_1(x_j)\,,
\end{equation}
where the function $V_1 : \mathbb{R}^{n_t\times T} \mapsto \mathbb{R}$ defined as
$V_1(x) \triangleq \frac{1}{T}\Var(i(X;Y,H)|X=x)$ is given by
\begin{align}
V_1(x) &= T\Var\left(C_{r}(H,P)\right) \label{eq:cvx1}\\
& + \sum_{i=1}^{n_{\min}}\E\left[V_{AWGN}\left(\frac{P}{n_t}\Lambda_i^2\right)\right]\label{eq:cvx2}\\
&+ \eta_5 \left(\frac{\|x\|_F^2}{n_t} - \frac{TP}{n_t}\right) \label{eq:cvx5}\\
&+ \eta_3 \left(\frac{\|x\|_F^2}{n_t} - \frac{TP}{n_t}\right)^2\label{eq:cvx3}\\
&+ \eta_4 \left( \|xx^T\|_F^2 - \frac{1}{n_t}\|x\|_F^4\right) \label{eq:cvx4}
\end{align}
where $c(\cdot)$ was defined in~\eqref{eq:cfdef} and
\begin{align} 
C_r(H,P) &\eqdef \frac{1}{2}\log\det\left(I_{n_r} + \frac{P}{n_t}HH^T\right) = \sum_{i=1}^{n_{\min}}
		C_{AWGN}\left(\frac{P}{n_t} \Lambda_i^2\right)\label{eq:crdef}\\
\eta_3    & \triangleq \frac{\log^2 e}{4}\Var\left(\sum_{k=1}^{n_{\min}} c(\Lambda_k^2)\right) \label{eq:eta3def}\\
   \eta_4    & \triangleq \frac{\log^2 e}{2n_t(n_t+2)}\left( \E\left[ \sum_{i=1}^{n_{\min}}
   c^2(\Lambda_i^2)\right] - 
   \frac{1}{(n_t - 1)} \sum_{i\neq j} \E\left[
   c(\Lambda_i^2)c(\Lambda_j^2)\right]\right) \\
\eta_5 & \triangleq {\log e\over 2}\Cov\left( C_r(H,P), \sum_{k=1}^{n_{\min}} c(\Lambda_k^2)\right)
+ \frac{\log^2 e}{T} \sum_{k=1}^{n_{\min}} \E\left[
\frac{\Lambda_k^2}{\left(1 + \frac{P}{n_t}\Lambda_k^2\right)^2}\right]\label{eq:eta5def}\ .
\end{align}   
\end{proposition}
\begin{remark}\label{rmk:v1positive} Every term in the definition of $V_1(x)$ (except the one with $\eta_5$) is non-negative (for
$\eta_4$-term, see~\eqref{eq:pct98}). The
$\eta_5$-term will not be important because for inputs satisfying power-constraint with equality it vanishes. Note
also that the first term in~\eqref{eq:eta5def} can alternatively be given
as
$$ \Cov\left(C_r(H,P), \sum_{k=1}^{n_{\min}} c(\Lambda_k^2)\right) = n_t {d\over dP} \Var\left[C_r(H,P)\right]\,. $$
\end{remark}

\begin{proof}
From~\eqref{eq:mimo_info_density}, we have the form of the information density.  First note that the information density over $n$ channel uses decomposes into a sum of $n$ independent terms,
\begin{align}
i(x^n; \mat{Y}^n, \mat{H}^n) = \sum_{j=1}^n i(x_j, \mat{Y}_j, \mat{H}_j)\,.
\end{align}
As such, the variance conditioned on $x^n$ also decomposes as
\begin{align}
\Var(i(x^n;\mat{Y}^n, \mat{H}^n)) = \sum_{j=1}^n \Var(i(x_j;\mat{Y}_j,\mat{H}_j))\,,
\label{eq:vs1}
\end{align}
from which~\eqref{eq:cvx_nto1} follows.  Because the variance decomposes as a sum
in~\eqref{eq:vs1}, we focus on
only computing $\Var(i(x;Y,H))$ for a single coherent block. Define
\begin{align}
f(h) &\eqdef T C_r(h,P)\\
g(x,h,z) &\eqdef \frac{\log
e}{2}\sum_{k=1}^{n_{\min}}\frac{\Lambda_k^2\|v_k^Tx\|^2 + 2\Lambda_k\inprod{v_k^T x}{\tilde{z}_k} - \frac{P}{n_t}\Lambda_k^2\|\tilde{z}_k\|^2}{1+\frac{P}{n_t} \Lambda_k^2}
\end{align}
so that $i(x;y,h) = f(h) + g(x,h,z)$ in notation from~\eqref{eq:mimo_info_density}.  With this, the quantity of interest is
\begin{align}\label{eq:var_terms_sum}
\Var(i(x,\mat{Y},\mat{H})) &= 
\Var(f(\mat{H})) + \Var(g(x,\mat{H}, \mat{Z})) + \Cov(f(\mat{H}), g(x,\mat{H}, \mat{Z})) \\
	&= \underbrace{\Cov(f(\mat{H}), g(x,\mat{H}, \mat{Z}))}_{\eqdef T_1} 
		+ \underbrace{\Var(f(\mat{H}))}_{\eqdef T_2} 
    + \underbrace{\Var\left(\E[g(x,\mat{H}, \mat{Z})|\mat{H}]\right)}_{\eqdef T_3}
		+ \underbrace{\E\left[\Var(g(x,\mat{H}, \mat{Z})|\mat{H})\right]}_{\eqdef T_4} \label{eq:vcv75}
\end{align}
where~\eqref{eq:vcv75} follows from the identity
\begin{align}\label{eq:var_identity_expandd}
\Var(g(x, \mat{H}, \mat{Z})) &= \E\left[\Var(g(x,\mat{H}, \mat{Z})|\mat{H})\right] +
\Var\left(\E[g(x,\mat{H},\mat{Z})|\mat{H}]\right)\ .
\end{align}

Below we show that $T_1$ and $T_3$ corresponds to~\eqref{eq:cvx5}, $T_2$ corresponds to~\eqref{eq:cvx1}, $T_4$ corresponds
to~\eqref{eq:cvx2}, and $T_3$ corresponds to~\eqref{eq:cvx3}
and~\eqref{eq:cvx4}. We evaluate each term separately. 
\begin{align}
T_1 &= \Cov(f(\mat{H}), g(x,\mat{H}, \mat{Z})) \\
&= \E\left[ (f(\mat{H}) -
\E[f(\mat{H})] )( g(x,H,Z) - \E[g(x,H,Z)])\right]\label{eq:agb1}\\
&= \frac{\log e}{2} \left(\frac{\|x\|_F^2}{n_t} - \frac{TP}{n_t}\right)
\sum_{k=1}^{n_{\min}}\E\left[(f(H) - \E[f(H)])(c(\Lambda_k^2) - \E[c(\Lambda_k^2)])  \right] \label{eq:agb2}\\
&= \frac{\log e}{2} \left(\frac{\|x\|_F^2}{n_t} - \frac{TP}{n_t}\right) \sum_{k=1}^{n_{\min}}
\Cov\left(f(H), c(\Lambda_k^2) \right) \label{eq:agb3}
\end{align}
where \eqref{eq:agb2} follows from noting that
\begin{align}
\E\left[ g(x, \mat{H}, \mat{Z}) | \mat{H} \right]
= \sum_{k=1}^{n_{\min}} \left(\|V_k^Tx\|^2 -                \frac{TP}{n_t}\right)
c(\Lambda_k^2) {\log e\over 2}\,.
\end{align}
Now, since $V_k$ is independent from $\Lambda_k$ by the rotational invariance assumption, we have that $f(\mat{H})$ is
independent from $V_k$, since $f(\mat{H})$
only depends on $\mat{H}$ through its eigenvalues, see~\eqref{eq:crdef}. We are only concerned with the expectation over $g(x,H,Z)$ in~\eqref{eq:agb1}, which reduces to
\begin{align}
\E\left[g(x,H,Z) - \E[g(x,H,Z)]  \middle| \Lambda_1,\hdots,\Lambda_{n_{\min}} \right]
&= \left(\frac{\|x\|_F^2}{n_t} - \frac{TP}{n_t}\right) \sum_{k=1}^{n_{\min}} c(\Lambda_k^2) - \E[c(\Lambda_k)^2] {\log
e\over2}\ ,
\end{align}
giving~\eqref{eq:agb2}.

Next, $T_2$ in~\eqref{eq:vcv75} becomes
\begin{align}
T_2 &= \Var(f(\mat{H})) \\
%&= n\Var\left(\frac{T}{2}\log\det\left(I_{n_r} +
%\frac{P}{n_t}\mat{H}\mat{H}^T\right)\right)\,.
&= T^2\Var\left(\sum_{k=1}^{n_{\min}}
C_{AWGN}\left(\frac{P}{n_t}\Lambda_k^2\right)\right)\ .
\end{align}

For $T_3$ in~\eqref{eq:vcv75}, we obtain
\begin{align}
T_3 &= \E\left[\Var(g(x,\mat{H}, \mat{Z})|\mat{H})\right]\\ 
&= \frac{\log^2 e}{4} \E\left[
\sum_{k=1}^{n_{\min}}\frac{4\Lambda_k^2\|V_k^Tx\|^2 +
2T\left(\frac{P}{n_t}\right)^2\Lambda_k^4}{\left(1 +
\frac{P}{n_t}\Lambda_k\right)^2}\right]\label{eq:bb75}\\
&= \frac{\log^2 e}{2}\sum_{k=1}^{n_{\min}} T \E\left[\frac{2\frac{TP}{n_t}\Lambda_k^2 +
T\left(\frac{P}{n_t}\right)^2\Lambda_k^4}{\left(1 +\frac{P}{n_t}\Lambda_k\right)^2}\right] 
+ 2\E\left[\frac{\frac{\|x\|_F^2}{n_t}\Lambda_k^2 -
\frac{TP}{n_t}\Lambda_k^2}{\left(1 +
\frac{P}{n_t}\Lambda_k^2\right)^2}\right]
\label{eq:bb76}\\
&= T \sum_{k=1}^{n_{\min}} V_{AWGN}\left({P \over n_t} \Lambda_k^2\right) +
\log^2(e) \left(\frac{\|x\|_F^2}{n_t} -
\frac{TP}{n_t}\right)\E\left[\frac{\Lambda_k^2}{\left(1 +
\frac{P}{n_t}\Lambda_k^2\right)^2}\right]
\end{align}
where 
\begin{itemize}
\item \eqref{eq:bb75} follows from taking the variance over $\tilde{Z}$ (recall $\tilde{Z} = U^TZ$ in~\eqref{eq:mimo_info_density}).
\item \eqref{eq:bb76} follows from Lemma~\ref{lem:var_w} applied to $\E[\|V_k^Tx\|^2]$, and
adding and subtracting the term
\begin{align}
\log^2(e) \E\left[ \frac{\frac{TP}{n_t}\Lambda_k^2}{\left(1 +
\frac{P}{n_t}\Lambda_k^2\right)^2}\right]\ .
\end{align}
\end{itemize}

Continuing with $T_3$ from~\eqref{eq:vcv75}, 
\begin{align}
T_3 &= \Var \E[g(x,\mat{H}, \mat{Z})|\mat{H}]\\ 
&= \Var\left( \frac{\log e}{2} \sum_{k=1}^{n_{\min}} c(\Lambda_k^2) \left(\|V_k^T x\|^2 -
\frac{TP}{n_t}\right)\right)\label{eq:eqe93}\\
&= 
\eta_3 \left(\frac{\|x\|_F^2}{n_t} - \frac{TP}{n_t}\right)^2 \nonumber\\
&+ \frac{\log^2 e}{4} \E\left[ \Var\left( \sum_{k=1}^{n_{\min}} c(\Lambda_k^2) \|V_k^T x\|^2 \middle|
\Lambda_1,\hdots,\Lambda_{n_{\min}}\right) \right]\label{eq:eqe95}
\end{align}
where
\begin{itemize}
\item \eqref{eq:eqe93} follows from taking the expectation over $\tilde Z$,
\item \eqref{eq:eqe95} follows from applying the variance
identity~\eqref{eq:var_identity_expandd} with respect to $V$ and
$\Lambda_1,\hdots,\Lambda_{n_{\min}}$, as well as recalling~\eqref{eq:eta3def}.
\end{itemize}

We are left to show that the term~\eqref{eq:eqe95} equals~\eqref{eq:cvx4}. 
To that end, define
\begin{align}
\phi(x) &\triangleq \E\left[\Var\left( \sum_{k=1}^{n_{\min}} c(\Lambda_k^2) \|V_k^T x\|^2 \middle|
 \Lambda_1,\hdots,\Lambda_{n_{\min}}\right)\right]\\
&= \sum_{k=1}^{n_{\min}}
\E[c^2(\Lambda_k^2)]\Var\left( \|V_k^Tx\|^2\right) + \sum_{k\not= l}^{n_{\min}}
\E[c(\Lambda_k^2)c(\Lambda_l^2)] \Cov(\|V_k^Tx\|^2,
\|V_l^Tx\|^2)\label{eq:pct98}\ .
\end{align}
We will finish the proof by showing
$$ \phi(x) = \frac{4}{\log^2 e} \eta_4 \left( \|x x^T\|_F^2 - \frac{1}{n_t}\|x\|_F^4\right)\,.$$
To that end, we first compute moments of $V$ drawn from the Haar measure on the orthogonal group.

\begin{lemma}\label{lem:moments_of_v}
Let $V$ be drawn from the Haar measure on $O(n)$, then for $i,j,k,l = 1,\hdots,n$ all unique,
\begin{align}
\E[V_{ij}^2] &= \frac{1}{n}\label{eq:mov1}\\
\E[V_{ij}V_{ik}] &= 0 \label{eq:mov2}\\
\E[V_{ij}^2 V_{ik}^2] &= \frac{1}{n(n+2)} \label{eq:mov3}\\
\E[V_{ij}^2 V_{kl}^2] &= \frac{n+1}{n(n-1)(n+2)} \label{eq:mov4}\\
\E[V_{ij}^4] &= \frac{3}{n(n + 2)} \label{eq:mov5}\\
\E[V_{ij}V_{ik}V_{lj}V_{lk}] &= \frac{-1}{n(n-1)(n+2)} \label{eq:mov5}\ .
\end{align}
\end{lemma}
Proof of this Lemma is given below.\\

First, note that the variance $\Var(\|V_k^Tx\|^2)$ does not depend on $k$, since the marginal
distribution of each $V_k$ is uniform on the unit sphere. Hence below we only consider $V_1$.  We
obtain
\begin{align}
\Var(\|V_1^T x\|^2) &=  \E[\|V_1^Tx\|^4] - \E^2[\|V_1^Tx\|^2]\label{eq:bgt106}\\
&= \E\left[\left(\sum_{i=1}^T
\sum_{j=1}^{n_t}\sum_{k=1}^{n_t}V_{j1}V_{k1}x_{ji}x_{ki}\right)^2\right] - \frac{\|x\|_F^4}{n_t^2}\\
&=
\E\left[\sum_{j=1}^{n_t}\sum_{k=1}^{n_t}\sum_{l=1}^{n_t}\sum_{m=1}^{n_t} V_{j1}V_{k1}V_{l1}V_{m1}
\inprod{r_j}{r_k}\inprod{r_l}{r_m}\right]\label{eq:bgt108}
\end{align}
where $r_j$ denotes the $j$-th row of $x$.  Now it is a matter counting similar terms.
\begin{align}
\E[\|V_1^Tx\|^4] &= \sum_{j=1}^{n_t} \E[V_{j1}^4]\|r_j\|^4 + 2\sum_{j\not= k}^{n_t}
\E[V_{j1}^2V_{k1}^2]
\inprod{r_j}{r_k}^2 + \sum_{j\not= k}^{n_t} \E[V_{j1}^2V_{k1}^2] \|r_j\|^2
\|r_k\|^2\label{eq:vnx109}\\
&= \frac{3}{n_t(n_t+2)} \sum_{j=1}^{n_t} \|r_j\|^4 + \frac{2}{n_t(n_t+2)} \sum_{j\not= k}^{n_t}
\inprod{r_j}{r_k}^2 + \frac{1}{n_t(n_t+2)}\sum_{j\not=
k}\|r_j\|^2\|r_k\|^2\label{eq:vnx110}\\
&= \frac{1}{n_t(n_t+2)}\left(\|x\|_F^4 + 2\|xx^T\|_F^2\right)\label{eq:vnx111}
\end{align}
where
\begin{itemize}
\item \eqref{eq:vnx109} follows from collecting like terms from the summation
in~\eqref{eq:bgt108}.
\item \eqref{eq:vnx110} uses Lemma~\ref{lem:moments_of_v} to compute each expectation.
\item \eqref{eq:vnx111} follows from realizing that
\begin{align}
\|x\|_F^4 &= \left( \sum_{j=1}^{n_t} \|r_j\|^2\right)^2 = \sum_{j=1}^{n_t} \|r_j\|^4 + \sum_{j\not= k}^{n_t} \|r_j\|^2\|r_k\|^2\\
\|xx^T\|_F^2 &= \sum_{j=1}^{n_t}\sum_{k=1}^{n_t} \inprod{r_j}{r_k}^2 = \sum_{j=1}^{n_t}
\|r_j\|^4 + \sum_{j\not= k}^{n_t} \inprod{r_j}{r_k}^2
\end{align}
\end{itemize}
Plugging this back into~\eqref{eq:bgt106} yields the variance term,
\begin{align}
\Var(\|V_1^Tx\|^2) &= \frac{1}{n_t(n_t+2)}\left(\|x\|_F^4 + 2\|xx^T\|_F^2\right) -
\frac{\|x\|_F^4}{n_t^2}
&= \frac{2}{n_t(n_t+2)}\left( \|xx^T\|_F^2 -
\frac{\|x\|_F^4}{n_t}\right)\label{eq:vtm114}\ .
\end{align}
Now we compute the covariance term from~\eqref{eq:pct98} in a similar way.  By symmetry of the columns of $V$, we
can consider only the covariance between $\|V_1^Tx\|^2$ and $\|V_2^Tx\|^2$, i.e.
\begin{align}
\Cov(\|V_1^Tx\|^2, \|V_2^Tx\|^2) = \E[\|V_1^2x\|^2\|V_2^Tx\|^2] -
\frac{\|x\|_F^4}{n_t^2}\label{eq:cvt115}\ .
\end{align}
Expanding the expectation, we get
\begin{align}
\lefteqn{\E[\|V_1^Tx\|^2\|V_2^Tx\|^2]}\\
&= \sum_{j,k,l,m}
\E[V_{1j}V_{1k}V_{2l}V_{2m}] \inprod{r_j}{r_k}\inprod{r_l}{r_m}\label{eq:ctr116}\\
&= \sum_{j=1}^{n_t} \E[V_{1j}^4] \|r_j\|^4 + \sum_{j \not= k} \E[V_{1j}^2V_{2k}^2]
\|r_j\|^2\|r_k\|^2 + 2\sum_{j\not= k}
\E[V_{1j}V_{1k}V_{2j}V_{2k}]\inprod{r_j}{r_k}^2 \label{eq:ctr117} \\
&= \frac{1}{n_t(n_t + 2)} \sum_{j=1}^{n_t} \|r_j\|^4 + \frac{n_t + 1}{(n_t-1)n_t(n_t+2)}
\sum_{j\not= k}\|r_j\|^2\|r_k\|^2 - \frac{2}{(n_t-1)n_t(n_t+2)}\sum_{j\not= k}
\inprod{r_j}{r_k}^2 \label{eq:ctr118}\\
&= \frac{1}{(n_t - 1)n_t(n_t + 2)}\left( (n_t + 1)\|x\|_F^4 -
2\|xx^T\|_F^2\right)\label{eq:119}\ .
\end{align}
With this, we obtain from~\eqref{eq:cvt115},
\begin{align}
\Cov(\|V_1^Tx\|^2, \|V_2^Tx\|^2) = \frac{2}{(n_t-1)n_t(n_t+2)} \left( \frac{\|x\|_F^4}{n_t} -
\|xx^T\|_F^2\right) \label{eq:ctm120}
\end{align}
where the steps follow just as in the variance
computation~\eqref{eq:vnx109}-\eqref{eq:vnx111}.

Finally, returning to~\eqref{eq:pct98}, using the variance~\eqref{eq:vtm114} and
covariance~\eqref{eq:ctm120}, we obtain
\begin{align}
\phi(x) = \frac{2}{n_t(n_t+2)}\left(\|xx^T\|_F^2 -
\frac{\|x\|_F^4}{n_t}\right)\left( \sum_{k=1}^{n_t}\E[c^2(\Lambda_k^2)] -
\frac{1}{n_t-1}\sum_{k\not= l} \E[c(\Lambda_k^2)c(\Lambda_l^2)]\right)\ .
\end{align}
Plugging this into~\eqref{eq:eqe95} finishes the proof.
\end{proof}

\begin{proof}[Proof of Lemma~\ref{lem:moments_of_v}]
We first note that all entries of $V$ have identical distribution, since permutations of rows and
columns leave the distribution invariant.  Because of this, we can WLOG only consider $V_{11},
V_{12}, V_{21}, V_{22}$ to prove the lemma.
\begin{itemize}
\item \eqref{eq:mov1} follows immediately from $\sum_{i=1}^n V_{ij}^2 = 1$ a.s.
\item Let $V_i, V_j$ be any two distinct columns of $V$, then \eqref{eq:mov2} follows from
\begin{align}
0 = \E[\inprod{V_i}{V_j}] = n\E[V_{11}V_{21}]
\end{align}
\item For~\eqref{eq:mov3} and~\eqref{eq:mov5}, let $E_1 = \E[V_{11}^4]$ and $E_2 =
\E[V_{11}^2V_{21}^2]$.  The following relations between $E_1, E_2$ hold,
\begin{align}
1 &= \E\left[ \left( \sum_{j=1}^n V_{1j}^2\right)^2\right]\\
&= n E_1 + n(n-1)E_2 \label{eq:vrel1}
\end{align}
and, noticing that multiplication of $V$ by the matrix
\begin{align}
\left[ \begin{array}{ccc}
1/\sqrt{2} & -1/\sqrt{2} & 0 \\
1/\sqrt{2} & 1/\sqrt{2} & 0 \\
0 & 0 & I_{n-2}
\end{array} \right]
\end{align}
where $I_n$ is the $n\times n$ identity matrix.  This is an orthogonal matrix, so we obtain the
relation
\begin{align}
E_1 &= \E\left[\left(\frac{V_{11}}{\sqrt{2}} +
\frac{V_{12}}{\sqrt{2}}\right)^4\right]\\
&= \frac{1}{2}E_1 + \frac{3}{2}E_2
\end{align}
from which we obtain $E_1 = 3E_2$.  With this and~\eqref{eq:vrel1}, we obtain
\begin{align}
E_1 &= \frac{3}{n(n+2)}\\
E_2 &= \frac{1}{n(n+2)}\label{eq:vee2}
\end{align}
\item For \eqref{eq:mov4}, take
\begin{align}
E_3 &= \E[V_{11}^2V_{22}^2]\\
&= \E\left[V_{11}^2 \left(1 - \sum_{j\not= 2}^n V_{2j}^2\right)\right]\\
&= \frac{1}{n} - \frac{1}{n(n+2)} - (n-2)E_3\ .
\end{align}
Solving for $E_3$ yields~\eqref{eq:mov4}.
\item For~\eqref{eq:mov5}, let $V_1, V_2$ denote the first and second column of $V$ respectively, and let $E_4 =
\E[V_{11}V_{12}V_{21}V_{22}]$, then \eqref{eq:mov5} follows from
\begin{align}
0 &= \E[\inprod{V_1}{V_2}^2]\\
&= nE_2 + n(n-1)E_4\ .
\end{align}
Using $E_2$ from~\eqref{eq:vee2} and solving for $E_4$ gives \eqref{eq:mov5}.
\end{itemize}
\end{proof}

The following propsition gives the value of the conditional variance of the information density when input distribution has i.i.d. $\matn(0,P/n_t)$ entries.  This will turn out to be the operational dispersion in the case where $\rank H > 1$.
\begin{proposition}\label{prop:telatar_disp}
Let $X^n =(\mat X_1, \ldots, \mat X_n)$ be i.i.d. with Telatar distribution~\eqref{eq:telatar} for each entry. Then 
\begin{equation}\label{eq:telatar_disp}
		\E\left[\Var(i(\mat X^n; \mat Y^n, \mat H^n)|\mat X^n)\right] = nT V_{iid}(P)\,,
\end{equation}	
	where $V_{iid}(P)$ is the right-hand side of~\eqref{eq:mimo_disp_expression}. 
\end{proposition}
\begin{proof}
To show this, we take the expectation of the expression given in
Proposition~\ref{prop:cond_var_x} when $X^n$ has i.i.d. $\mathcal{N}(0,P/n_t)$ entries.  The
terms~\eqref{eq:cvx1} and~\eqref{eq:cvx2} do not depend on $X^n$, and these give us the first
two terms in~\eqref{eq:mimo_disp_expression}.~\eqref{eq:cvx5} vanishes immediately, since
$\E[\|X\|_F^2] = TP$ by the power constraint.  It is left to compute the expectation
over~\eqref{eq:cvx3} and~\eqref{eq:cvx4} from the expression in
Proposition~\ref{prop:cond_var_x}.  Using identities for $\chi^2$ distributed random variables (namely,
$\EE[\chi^2_k]=k$, $\Var(\chi^2_k) = 2k$), we get:
\begin{align}
%\frac{\eta_3}{n}\sum_{j=1}^n \E\left(\frac{\|X_k\|_F^2}{n_t} - \frac{TP}{n_t}\right)^2
\frac{\eta_3}{n_t^2}\Var(\|X_1\|_F^2) &= \frac{\eta_3}{n_t}\left(\frac{P}{n_t}\right)^2 2 T\\
\E[\|X_1\|_F^4] &= %\Var(\|X\|_F^2) + \E[\|X\|_F^2]^2\\
%&= 2n_tT\left(\frac{P}{n_t}\right)^2 + T^2P^2\\
%&= 
TP^2\left( T + \frac{2}{n_t}\right)\\
\E[\|X_1X_1^T\|_F^2] &= %n_t\E[\|R_1\|_F^4] + n_t(n_t-1)\E[\inprod{R_1}{R_2}^2]\\
%&= n_t T\left(\frac{P}{n_t}\right)^2 (2 + T) + n_t(n_t-1)T\left(\frac{P}{n_t}\right)^2\\
%&= 
n_t T \left(\frac{P}{n_t}\right)^2(1 + T + n_t)\\
\E\left[\|X_1X_1^T\|_F^2 - \frac{\|X_1\|_F^4}{n_t}\right] 
&= T\left(\frac{P}{n_t}\right)^2(n_t-1)(n_t+2)\,.
\end{align}
Hence, the sum of terms in~\eqref{eq:cvx3} + \eqref{eq:cvx4} after taking expectation over $X^n$ yields
$$ T \left(P\over n_t\right)^2 \left[2 {\eta_3\over n_t} + (n_t-1)(n_t+2) \eta_4\right]\,.$$
Introducing random variables $U_i = c(\Lambda_i^2)$ the expression in the square brackets equals 
\begin{equation}\label{eq:grb1}
	{\log^2 e\over 2} {1\over n_t} \left[ \Var\left(\sum_i U_i\right) + (n_t-1) \sum_i \EE[U_i^2] - \sum_{i\neq j} \EE[U_i U_j] \right]\,.
\end{equation}
At the same time, the third term in expression~\eqref{eq:mimo_disp_expression} is
\begin{equation}\label{eq:grb2}
	{\log^2 e\over 2} {1\over n_t} \left[ n_t \sum_i \EE[U_i^2] - \left(\sum_i \EE[U_i]\right)^2 \right]\,.
\end{equation}
One easily checks that~\eqref{eq:grb1} and~\eqref{eq:grb2} are equal.
\end{proof}

The next proposition shows that, when the rank of $H$ is larger than $1$, the conditional variance in \eqref{eq:vmindef} is constant over the set of caids.  Thus we can compute the conditional variance for the i.i.d. $\matn(0,P/n_t)$ caid, and conclude that this expression is the minimizer in \eqref{eq:vmindef}.
\begin{proposition}\label{prop:vminrank2} If $\PP[\rank H>1]>0$, then for \underline{any} caid $X\sim P_X$ 
	we have 
	$$ \Var(X; Y,H)) = T\EE[V_1(X)] = T V_{iid}(P)\,.$$
	In particular, the $V(P)$ 
	defined as infimum over all caids~\eqref{eq:vmindef} satisfies $V(P)=V_{iid}(P)$.
\end{proposition}
\begin{proof} For any caid the term~\eqref{eq:cvx5} vanishes. Let $X^*$ be Telatar distributed. 
To analyze~\eqref{eq:cvx3} we recall that from~\eqref{eq:rk2_caid} we have
	$$ \EE[\|X\|_F^4] = \sum_{i,j,i',j'} \EE[X_{i,j}^2 X_{i',j'}^2] = \EE[\|X^*\|_F^4]\,.$$
For the term~\eqref{eq:cvx4} we notice that
	$$ \|XX^T\|_F^2 = \sum_{i,j} \langle R_i, R_j\rangle^2\,,$$
	where $R_i$ is the $i$-th row of $X$. By~\eqref{eq:rk2_caidrows} from Theorem \ref{thm:caid_conds} we then also have
	$$ \EE[\|XX^T\|_F^2] = \EE[\|X^* X^{*T}\|_F^2]\,.$$
To conclude, $\EE[V_1(X)]=\EE[V_1(X^*)]=V_{iid}(P)$.
\end{proof}

In the case where $\rank H \leq 1$, it turns out that the conditional variance \emph{does} vary over the set of caids.  The following proposition gives the expression for the conditional variance in this case, as a function of the caid.
\begin{proposition} 
\label{prop:disp_calc}
If $\PP[\text{rank}(H) \le 1] = 1$, then for any capacity achieving input $X$ we have
\begin{align} \label{cond_var_expression}
    %{1\over T} \Var[i(X;Y,H) | X] = V_1(n_t,1,T,P) - \frac{\chi_2}{n_t^2 T}\Var(\|X\|_F^2)
{1\over T} \E\left[\Var(i(X;Y,H) | X)\right] &= T\Var\left(
C_{AWGN}\left(\frac{P}{n_t}\Lambda_1^2\right) \right)
+ \E V_{AWGN}\left( \frac{P}{n_t}\Lambda_1^2\right)\\
&+ \eta_1\left(\frac{P}{n_t}\right)^2 - \frac{\eta_2}{2n_t^2 T}\Var(\|X\|_F^2)
\end{align}
where $\eta_1,\eta_2$ are defined in~\eqref{eq:eta1def}-\eqref{eq:eta2def}.
\end{proposition}

\begin{proof}
As in Prop.~\ref{prop:telatar_disp} we need to evaluate the expectation of terms in~\eqref{eq:cvx5}-\eqref{eq:cvx4}. Any
caid $X$ should satisfy $\EE[\|X\|_F^2] = TP$ and thus the term~\eqref{eq:cvx5} is zero. The term ~\eqref{eq:cvx3} can
be expressed in terms of $\Var(\|X\|_F^2)$, but the \eqref{eq:cvx4} presents a non-trivial complication due to the
presence of $\|XX^T\|_F^2$, whose expectation is possible (but rather tedious) to compute by invoking properties of
caids established in Theorem~\ref{thm:caid_conds}. Instead, we recall that the sum ~\eqref{eq:cvx3}+\eqref{eq:cvx4}
equals~\eqref{eq:eqe95}. Evaluation of the latter can be simplified in this case due to constraint on the rank of $H$.
Overall, we get 
\begin{align}
\E\left[\Var(i(X;Y,H)|X)\right] &=
T^2\Var\left(C_{AWGN}\left(\frac{P}{n_t}\Lambda_1^2\right)\right)
+ T \E\left[V_{AWGN}\left(\frac{P}{n_t}\Lambda_1^2\right)\right]\\
&+ \frac{\log^2 e}{4}
\E\left[\Var\left( c(\Lambda_1^2) \left(\|V_1^TX\|^2 - \frac{TP}{n_t}\right)\middle| X \right)\right]\,, \label{eq:mvt149}
\end{align}
where $c(\cdot)$ is from~\eqref{eq:cfdef}. The last term in~\eqref{eq:mvt149} can be written as
\begin{align}
\E\left[c(\Lambda_1^2)^2\right] \E\left[\left(\|V_1^TX\|^2 -
\frac{TP}{n_t}\right)^2\right] - \E^2[c(\Lambda_1^2)]\E\left[\left( \E[\|V_1^TX\|_F^2 | X] -
\frac{TP}{n_t}\right)^2\right]\label{eq:mtt152}
\end{align}
which follows from the identity $\Var(AB) = \E[A^2]\E[B^2] - \E^2[A]\E^2[B]$ for independent $A,B$. The second term
in~\eqref{eq:mtt152} is easily handled since from Lemma~\ref{lem:var_w}  we have $\E[\|V_1^TX\|_F^2|X] = \|X\|_F^2/n_t$.
To compute the first term in~\eqref{eq:mtt152} recall from Theorem~\ref{thm:caid_conds} that for any fixed unit-norm $v$
and caid $X$ we must have $v^T X \sim \matn(0, P/n_t I_T)$. Therefore, we have 
$$ \E\left[\left(\|V_1^TX\|^2 - \frac{TP}{n_t}\right)^2 \bigg| V_1 \right] = {2TP^2\over n_t^2}\,.$$
Putting everything together we get that~\eqref{eq:mtt152} equals 
\begin{align}
\E[c(\Lambda_1^2)^2] 2T\left(\frac{P}{n_t}\right)^2 - \E[c(\Lambda_1^2)]^2 \frac{1}{n_t^2}\Var(\|X\|_F^2)
\end{align}
concluding the proof.
\end{proof}

The question at hand is: which input distribution $\mat{X}$ that achieves capacity minimizes~\eqref{cond_var_expression}?  Proposition~\ref{prop:disp_calc} reduces this problem to maximizing $\Var(\|\mat{X}\|_F^2)$ over the set of capacity achieving input distributions.  This will be analyzed in Section~\ref{sec:MISO}.

Finally, the following lemma computes the Berry Esseen constant.  This is a technical result that will be needed for both the achievability and converse proofs.

\begin{lemma}\label{lem:thirdmom} Fix $x_1,\ldots,x_n \in \mreals^{n_t \times T}$  and let $W_j = i(x_j;
Y_j,H_j)$, where $Y_j, H_j$ are distributed as the output of channel~\eqref{eq:channel} with input $x_j$. Define the
Berry-Esseen ratio
\begin{equation}\label{eq:bndef}
	B_n(x^n) \eqdef \sqrt{n}{\sum_{j=1}^n \EE[|W_j - \EE[W_j]|^3] \over \left(\sum_{j=1}^n \Var(W_j)\right)^{3/2}}  \ .
\end{equation}
Then whenever $\sum_{j=1}^n \|x_j\|_F^2 = nTP$ and  $\max_j\|x_j\|_F \le \delta n^{1\over 4}$ we have
$$ B_n(x^n) \le K_1 \delta^2 \sqrt{n} + K_2 n^{1/4} + \frac{K_3}{n^{1/2}} $$
where $K_1,K_2, K_3>0$ are constants which only depend on channel parameters but not $x^n$ or $n$.
\end{lemma}

The proof of Lemma~\ref{lem:thirdmom} can be found in Appendix~\ref{apx:be}.

\subsection{Hypothesis testing}

Many finite blocklength results are derived by considering an optimal hypothesis between appropriate distributions. We define $\beta_{\alpha}(P,Q)$ to be the minimum error probability of all statistical tests $P_{Z|W}$ between distributions $P$ and $Q$, given that the test chooses $P$ when $P$ is correct with at least probability $\alpha$.  Formally:
\begin{align}
\beta_{\alpha}(P,Q) = \inf_{P_{Z|W}} \left\{\int_{\mathcal{W}} P_{Z|W}(1|w)dQ(w): \int_{\mathcal{W}} P_{Z|W}(1|w)dP(w)
\geq \alpha \right\}\ .
\end{align}
The classical Neyman-Pearson lemma shows that the optimal test achieves
\begin{align}
\beta_{\alpha}(P,Q) = Q\left[\frac{dP}{dQ} > \gamma\right]
\end{align}
where $dP\over dQ$ denotes the Radon-Nikodym derivative of $P$ with respect to $Q$, and $\gamma$ is chosen to satisfy
\begin{align}
\alpha = P\left[\frac{dP}{dQ} > \gamma \right]\,.
\end{align}
We recall a simple bound on $\beta_\alpha$ following from the data-processing inequality (see~\cite[(154)-(156)]{PPV08}
or, in different notation,~\cite[(10.21)]{CKv2}):
\begin{equation}\label{eq:ba_dpq}
	\beta_{\alpha}(P,Q) \geq \exp\left( -\frac{D(P||Q) + h_B(\alpha)}{\alpha} \right)\,.
\end{equation}
A more precise bound~\cite[(102)]{PPV08} is
\begin{align}\label{eq:beta_std_bound}
\beta_{\alpha}(P,Q) \geq \sup_{\gamma > 0} \frac{1}{\gamma}\left( \alpha - \PP\left[ \log \frac{dP}{dQ} \geq \log \gamma\right]\right)\,.
\end{align}

We will also need to define the performance of composite hypothesis tests. To this end, let $F \subset \matx$ and 
$P_{Y|X}:\matx\to\maty$ be a random transformation. We define
\begin{align}\label{eq:ktdef}
\kappa_{\tau}(F,Q_Y) = \inf_{P_{Z|Y}} \left\{\int_{\mathcal{Y}} P_{Z|Y}(1|y)dQ_Y : \inf_{x\in F} \int_{\mathcal{Y}}
P_{Z|Y}(1|y)dP_{Y|X=x} \geq \tau \right\} \ .
\end{align}
We can lower bound the error in a composite hypothesis test $\kappa_{\tau}$ by the error in an appropriately chosen binary hypothesis test as follows:
% Kt >= Ba lemma
\begin{lemma}\label{lemma:kt_ba}
For any $P_{\tilde X}$ on $\matx$ we have 
\begin{align}
\kappa_{\tau}(F,Q_Y) \geq \beta_{\tau P_{\tilde{X}}[F]}(P_{Y|X} \circ P_{\tilde{X}}, Q_Y)
\end{align}
\end{lemma}
\begin{proof}
Let $P_{Z|Y}$ be any test satisfying conditions in the definition~\eqref{eq:ktdef}. We have the chain

\begin{align}\label{eq:nhy1}
	\int_{\maty} P_{Z|Y}(1|y) d(P_{Y|X} \circ P_{\tilde X}) &= 
	\int_{\matx} dP_{\tilde X} \int_{\maty} P_{Z|Y}(1|y) dP_{Y|X=x} \\
	&\ge \tau P_{\tilde X}[F]\,,\label{eq:nhy2}
\end{align}	
where~\eqref{eq:nhy1} is from Fubini and~\eqref{eq:nhy2} from constraints on the test.
Thus $P_{Z|Y}$ is also a test satisfying conditions in the definition of $\beta_{\tau P_{\tilde X}[F]}$. Optimizing over
the tests completes the proof.
\end{proof}

%%%%%%%%%%%%%%%%%%%%%%%%%%%%%%%%%%%%%%%%%%%%%%%%%%%%%%%%%%%%%%%%%%
%%%%% ACHIEVABILITY %%%%%%%%%%%%%%%%%%%%%%%%%%%%%%%%%%%%%%%%%%%%%%
%%%%%%%%%%%%%%%%%%%%%%%%%%%%%%%%%%%%%%%%%%%%%%%%%%%%%%%%%%%%%%%%%%

\section{Achievability}\label{sec:ach}

% Intro and kb statement
In this section, we prove the achievability side of the coding theorem for the MIMO-BF channel. We will rely on the $\kappa\beta$ bound~\cite[Theorem 25]{PPV08}, quoted here:
\begin{theorem}[$\kappa\beta$ bound]
Given a channel $P_{Y|X}$ with input alphabet $\mathcal{A}$ and output alphabet $\mathcal{B}$, for any distribution $Q_Y$ on $\mathcal{B}$, any non-empty set $F \subset \mathcal{A}$, and $\epsilon, \tau$ such that $0 < \tau < \epsilon < 1/2$, there exists and $(M,\epsilon)$-max code satisfying
\begin{align}\label{eq:kappabeta}
M \geq \frac{\kappa_{\tau}(F,Q_Y)}{\sup_{x\in F} \beta_{1-\epsilon+\tau}(P_{Y|X=x},Q_Y)}\ .
\end{align}
\end{theorem}

The art of applying this theorem is in choosing $F$ and $Q_Y$ appropriately. The intuition in choosing these is as
follows: although we know the distributions in the collection $\{P_{Y|X=x}\}_{x\in F}$, we do not know which $x$ is
actually true in the composite, so if $Q_Y$ is in the ``center'' of the collection, then the two hypotheses can be difficult to
distinguish, making the numerator large. However, for a given $x$, $P_{Y|X=x}$ vs $Q_Y$ may still be easily to
distinguish, making the denominator small.  The main principle for applying the $\kappa \beta$-bound is thus:
Choose $F$ and $Q_Y$ such that $P_{Y|X=x}$ vs $Q_Y$ is easy to distinguish for any given $x$, yet the composite
hypothesis $Y \sim \{P_{Y|X=x}\}_{x\in F}$ is hard to distinguish from a simple one $Y \sim Q_Y$.

The main theorem of this section gives achievable rates for the MIMO-BF channel, as follows:
\begin{theorem}\label{thm:mainach}
Fix an arbitrary caid $P_X$ on $\mreals^{n_t \times T}$ and let 
\begin{align}
V' &\eqdef \frac{1}{T}\E\left[\Var(i(\mat{X};\mat{Y},\mat{H})|\mat{X})\right] = \EE[V_1(X)]\,,
\label{eq:vprime}
\end{align}
where $V_1(x)$ is introduced in Proposition~\ref{prop:cond_var_x}. Then we have
\begin{align}
\log M^*(nT,\epsilon,P) \geq nTC(P) - \sqrt{nT V'}Q^{-1}(\epsilon) + o(\sqrt{n})
\end{align}
with $C(P)$ given by~\eqref{eq:mimo_capacity}.
\end{theorem}

%We first give the main structure of the proof, then the more technical steps follow as lemmas.

\begin{proof}
Let $\tau>0$ be a small constant (it will be taken to zero at the end).
We apply the $\kappa\beta$ bound (\ref{eq:kappabeta}) with auxiliary distribution $Q_Y=(P_{Y,H}^*)^n$, where $P_{Y,H}^*$ is the
caod~\eqref{eq:mimo_caod}, and the set $F_n$ is to be specified shortly. Recall notation $D_n(x^n)$, $V_n(x^n)$ and
$B_n(x^n)$ from~\eqref{eq:dndef},~\eqref{eq:cvx_nto1} and~\eqref{eq:bndef}. For any $x^n$ 
such that $B_n(x^n) \le \tau \sqrt{n}$, we have from \cite[Lemma 14]{YP10},
\begin{equation}\label{eq:beta_normal}
	-\log \beta_{1-\epsilon+\tau}(P_{\mat{Y}^n\mat{H}^n|\mat{X}^n=x^n}, P_{YH}^{* n}) \ge 
	nT D_n(x^n) + \sqrt{nT V_n(x^n)} Q^{-1}(1-\epsilon-2\tau) - \log {1\over \tau} - K' 
\end{equation}	
where $K'$ is a constant that only depends on channel parameters.
We mention that obtaining~\eqref{eq:beta_normal} from~\cite[Lemma 14]{YP10} also requires that $V_n(x^n)$ be bounded
away from zero by a constant, which holds since in the expression for $V_n(x^n)$ in Proposition~\ref{prop:cond_var_x}, the term~\eqref{eq:cvx2}
is strictly positive, term~\eqref{eq:cvx5} will vanish, and terms~\eqref{eq:cvx3} and~\eqref{eq:cvx4} are both
non-negative.

Considering~\eqref{eq:beta_normal}, our choice of the set $F_n$ should not be surprising:
\begin{align}\label{eq:F_set}
F_n \eqdef \left\{x^n : \|x^n\|_F^2 = nTP, V_n(x^n) \leq V' + \tau, \max_j\|x_j\|_F \le \delta n^{1\over 4}
\right\}\,,
\end{align}
where $\delta = \delta(\tau)>0$ is chosen so that Lemma~\ref{lem:thirdmom} implies $B_n(x^n) \le \tau \sqrt{n}$ for any $x^n\in
F_n$. 
Under this choice from~\eqref{eq:beta_normal},~\eqref{eq:dndef2} and Lemma~\ref{lem:thirdmom} we conclude
\begin{equation}\label{eq:abet1}
	\sup_{x^n \in F_n} \log \beta_{1-\epsilon+\tau}(P_{\mat{Y}^n\mat{H}^n|\mat{X}^n=x^n}, P_{YH}^{* n}) \le
- nT C(P) + \sqrt{nT(V'+\tau)} Q^{-1}(\epsilon-2\tau) + K''\,,
\end{equation}
where $K'' = K' + \log {1\over \tau}$.

To lower bound the numerator $\kappa_{\tau}(F_n,P_{Y,H}^{*n})$ we first state two auxiliary lemmas, whose proofs follow.  The first, Lemma \ref{lemma:divergence_convergence}, shows that the output distribution induced by an input distribution that is uniform on the sphere is ``similar'' (in the sense of divergence) to the $n$-fold product of the caod.
\begin{lemma}\label{lemma:divergence_convergence} 
Fix an arbitrary caid $P_X$ and let $X^n$ have i.i.d. components $\sim P_X$. Let 
\begin{align} \label{eq:norm_caid}
\tilde X^n \eqdef \frac{\mat{X}^n}{\|\mat{X}^n\|_F}\sqrt{nTP}
\end{align}
where $\|X^n\|_F = \sqrt{\sum_{t=1}^n \|X_j\|_F^2}$. Then
\begin{align}\label{eq:dcc}
D(P_{\mat{Y}^n\mat{H}^n|\mat{X}^n} \circ P_{\tilde X^n} || P_{Y,H}^{*n}) \leq {TP \log e
\over n_t} \EE[\|H\|_F^2] \,,
\end{align}
where $P_{Y,H}^{*n}$ is the $n$-fold product of the caod~\eqref{eq:mimo_caod}.
\end{lemma}

The second, Lemma \ref{lem:fnprob}, shows that a uniform distribution on the sphere has nearly all of its mass in $F_n$ as $n\to\infty$.

\begin{lemma}\label{lem:fnprob} With $\tilde X^n$ as in Lemma 
\ref{lemma:divergence_convergence} and set $F_n$ defined as in~\eqref{eq:F_set} (with arbitrary $\tau>0$ and $\delta>0$) we have as $n\to\infty$,
	$$ \PP[\tilde X^n \in F_n] \to 1 $$
\end{lemma}

Denote the right-hand side of~\eqref{eq:dcc} by $K_1$ and consider the following chain:
\begin{align} 
\kappa_{\tau}(F_n,Q_{\mat{Y}^n}) &\geq 
\exp\left( -\frac{D(P_{\mat{Y}^n\mat{H}^n|\mat{X}^n} \circ P_{\tilde{\mat{X}}^n} || Q_{\mat{Y}^n} ) + \log 2}{\tau
P_{\tilde{\mat{X}}^n}[F_n]} \right)\label{eq:dcp1}\\
	&\geq \exp\left(-{K_1 + \log 2 \over \tau P_{\tilde X^n}[F_n]}\right)\label{eq:dcp2}\\
	&= \exp\left(-{K_1 + \log 2 \over \tau +o(1)}\right)\label{eq:dcp3}\\
	&\ge K_2(\tau) \label{eq:dcp4}\,,
\end{align}
where~\eqref{eq:dcp1} follows from Lemmas~\ref{lemma:kt_ba} and~\eqref{eq:ba_dpq} with $P_{\tilde X^n}$ as in
Lemma~\ref{lemma:divergence_convergence},~\eqref{eq:dcp2} is from Lemma \ref{lemma:divergence_convergence},~\eqref{eq:dcp3} is from
Lemma~\ref{lem:fnprob}, and in~\eqref{eq:dcp4} we introduced a $\tau$-dependent constant $K_2$.

Putting~\eqref{eq:abet1} and~\eqref{eq:dcp4} into the $\kappa\beta$-bound we obtain
$$ \log M^*(nT,\epsilon,P) \ge nT C(P) - \sqrt{nT(V'+\tau)} Q^{-1}(\epsilon-2\tau) - K'' - K_2(\tau)\,.$$
Taking $n\to\infty$ and then $\tau\to 0$ completes the proof.
\end{proof}

Now we prove the two lemmas used in the Theorem.

\begin{proof}[Proof of Lemma~\ref{lemma:divergence_convergence}]
In the case of no-fading ($H_j=1$) and SISO, this Lemma follows from~\cite[Proposition 2]{molavianjazi2015second}. Here
we prove the general case. Let us introduce an auxiliary channel acting on $X_j$ as follows:
\begin{align}\label{eq:tildech}
\tilde Y_j = \mat{H}_j \frac{\mat{X}_j}{\|\mat{X}^n\|_F}\sqrt{nTP} + \mat{Z}_j, \qquad j=1,\hdots,n
\end{align}
With this notation, consider the following chain:
\begin{align} 
	D(P_{\mat{Y}^n\mat{H}^n|\mat{X}^n} \circ P_{\tilde{\mat{X}}^n}  || P_{Y,H}^{*n} ) &=
	D(P_{\tilde Y^n\mat{H}^n|\mat{X}^n} \circ P_{{\mat{X}}^n}  || P_{Y,H}^{*n} ) \label{eq:pl1}\\
	&=
	D(P_{\tilde Y^n\mat{H}^n|\mat{X}^n} \circ P_{{\mat{X}}^n}  || P_{\mat{Y}^n\mat{H}^n|\mat{X}^n} \circ
	P_{{\mat{X}}^n} ) \label{eq:pl2}\\
	&= D(P_{\tilde Y^n\mat{H}^n|\mat{X}^n} ||  P_{\mat{Y}^n\mat{H}^n|\mat{X}^n} | P_{\mat{X}^n}) \label{eq:pl3}\\
	&= D(P_{\tilde Y^n|\mat{H}^n,\mat{X}^n} ||  P_{\mat{Y}^n|\mat{H}^n|\mat{X}^n} | P_{\mat{X}^n} P_{H^n}) 
		\label{eq:pl4}\\
	&= \frac{\log e}{2} \E\left[ \left(1 -
{\sqrt{nTP}\over \|\mat{X}^n\|_F}\right)^2 \sum_{t=1}^n \|H_j X_j\|_F^2 \right]\label{eq:pl5}\\
	&=
	\frac{\log e}{2n_t}\E[\|H\|_F^2] \E\left[ \left(\|\mat{X}^n\|_F - \sqrt{nTP}\right)^2 \right]\label{eq:pl6}\\
	&= 
	\frac{\log e}{n_t}\E[\|H\|_F^2] (nTP - \sqrt{nTP} \EE[\|\mat{X}^n\|_F]) \label{eq:pl7}
\end{align}
where~\eqref{eq:pl1} is by clear from~\eqref{eq:tildech},~\eqref{eq:pl2} follows since $P_X$ is a
caid,~\eqref{eq:pl3}-\eqref{eq:pl4} are standard identities for divergence,~\eqref{eq:pl5} follows since both $\tilde
Y_j$ and $Y_j$ are unit-variance Gaussians and $D(\matn(0,1)\|\matn(a,1))={a^2\log e\over 2}$,~\eqref{eq:pl6} is
from Lemma~\ref{lem:var_w} (see Remark \ref{rem:exp_hx}) and~\eqref{eq:pl7} is just algebra along with the assumption that $\EE[\|X^n\|_F^2] = nTP$.

It remains to lower bound the expectation $\EE[\|X^n\|_F]$. Notice that for any uncorrelated random variables $B_t \ge 0$ with
mean 1 and variance 2 we have
\begin{equation}\label{eq:plxx}
	\EE\left[\sqrt{{1\over n} \sum_{t=1}^n B_t}\right] \ge 1 - {1\over n}\,,
\end{equation}
which follows from $\sqrt{x} \ge {3x-x^2\over 2}$ for all $x\ge 0$ and simple computations.
Next consider the chain:
\begin{align} \E[\|\mat{X}^n\|_F] &= 
	\E\left[ \sqrt{\sum_{i,j} \sum_{t=1}^n (X_t)_{i,j}^2 } \right]\\
	&\geq \sqrt{\frac{n}{n_tT}} \sum_{i,j} \E\left[  \sqrt{{1\over n}\sum_{t=1}^n (X_t)_{i,j}^2}\right]\\
	&=\sqrt{n T P} \left(1-{1\over n}\right)\label{eq:plxy}
\end{align}
where in~\eqref{eq:plxy} we used the fact that for any caid, $\{(X_t)_{i,j}, t=1,\ldots n\} \sim \matn(0, P/n_t)$ i.i.d. (from
Theorem~\ref{thm:caid_conds}) and applied~\eqref{eq:plxx} with $B_t = {(X_t)_{i,j}^2 n_t\over P}$.
Putting together~\eqref{eq:pl7} and~\eqref{eq:plxy} completes the proof.
\end{proof}

\begin{proof}[Proof of Lemma~\ref{lem:fnprob}]
Note that since $\|X^n\|_F^2$ is a sum of i.i.d. random variables, we have ${\|X^n\|_F\over \sqrt{nTP}}\to 1$ almost surely. In addition we
have
$$ \EE[\|X_1\|_F^8] \le (n_t T)^3 \sum_{i,j} \EE[(X_1)_{i,j}^8] \eqdef K\,,$$
where we used the fact (Theorem~\ref{thm:caid_conds}) that $X_1$'s entries are Gaussian. Then we have from independence
of $X_j$'s and Chebyshev's inequality,
$$ \PP[\max_j \|X_j\|_F \le \delta' n^{1\over 4}] = \PP[\|X_1\|_F \le \delta' n^{1\over 4}]^{n} \ge \left(1-{K\over
\delta'^8 n^2}\right)^n \to 1 $$
as $n\to \infty$.
Consequently, 
$$ \PP[\max_j \|\tilde X_j\|_F \le \delta n^{1\over 4}] \ge \PP\left[\max_j \|X_j\|_F \le {\delta\over 2} n^{1\over
4}\right] - \PP\left[{\|X^n\|_F\over \sqrt{nTP}} < {1\over 2}\right] \to 1 $$
as $n\to \infty$.

Next we analyze the behavior of $V_n(\tilde X^n)$. From Proposition~\ref{prop:cond_var_x} we see that, due to $\|\tilde
X^n\|_F^2 = nTP$, the term~\eqref{eq:cvx5} vanishes, while~\eqref{eq:cvx3} simplifies. Overall, we have
\begin{align}
V_n(\tilde X^n) &= K
+ \left( \frac{nTP}{\|X^n\|_F^2}\right)^2 
	{1\over n} \sum_{j=1}^n \left({\eta_3-\eta_4\over n_t} \|X_j\|_F^4 + \eta_4  \|X_jX_j^T\|_F^2 \right)\,, \label{eq:cle1}
\end{align}
where we replaced the terms that do not depend on $x^n$ with $K$. Note that the first term in parentheses (premultiplying the sum) converges almost-surely to 1, by the strong law of large numbers. Similarly, the normalized sum converges to the expectation (also by the strong law of large
numbers). Overall, applying the SLLN in the limit as $n\to\infty$, we obtain:
\begin{align} \lim_{n\to\infty}V_n(\tilde{\mat{X}}^n) &= \lim_{n\to\infty} \frac{1}{n} \sum_{j=1}^n
V_1(\tilde{\mat{X}_j}) \\&= \E[V_1(\mat{X})] \eqdef V'\,.
\end{align}
In particular, $\PP[V_n(\tilde X^n) \le V' + \tau] \to 1$. This concludes the proof of $\PP[\tilde X^n \in F_n]\to 1$.
\end{proof}

%%%%%%%%%%%%%%%%%%%%%%%%%%%%%%%%%%%%%%%%%%%%%%%%%%%%%%%%%%%%%%%%%%
%%%%% CONVERSE %%%%%%%%%%%%%%%%%%%%%%%%%%%%%%%%%%%%%%%%%%%%%%%%%%%
%%%%%%%%%%%%%%%%%%%%%%%%%%%%%%%%%%%%%%%%%%%%%%%%%%%%%%%%%%%%%%%%%%

\section{Converse}\label{sec:conv}

Here we state and prove the converse part of Theorem~\ref{thm:mimo_dispersion}.  There are two challenges in proving the converse relative to
other finite blocklength proofs. First, behavior of the information density~\eqref{eq:mimo_info_density}
varies widely as $x^n$ varies over the power-sphere 
\begin{equation}\label{eq:powersp}
	S_n = \{x^n \in (\mathbb{R}^{n_t\times T })^n : \|x^n\|_F^2
= nTP\}.
\end{equation}
Indeed, when $\max_j \|x_j\|_F \ge c n^{1\over4}$ the distribution of information density ceases to be Gaussian.  
In contrast, the information density for the AWGN channel is constant over $S_n$.

Second, assuming asymptotic normality, we have for any $x^n \in S_n$:
$$ -\log \beta_{1-\epsilon}(P_{\mat{Y}^n\mat{H}^n|\mat{X}^n=x^n}, P_{Y,H}^{*n}) \approx nC(P) - \sqrt{nV_n(x^n)}
Q^{-1}(\epsilon) + o(\sqrt{n})\,.$$
However, the problem is that $V_n(x^n)$ is also non-constant. In fact there exists regions of $S_n$ where
$V_n(x^n)$ is abnormally small. Thus we need to also show that no capacity-achieving codebook can live on those
abnormal sets.

The main theorem of the section is the following:

\begin{theorem}\label{thm:mainconv} For any $\delta_n\to0$ there exists $\delta'_n\to 0$ such that any $(n,M,\epsilon)$-max code 
with $\epsilon < 1/2$ and codewords satisfying $\max_{1\le j\le n}\|x_j\|_{F} \le \delta_n n^{1\over 4}$ has size
bounded by
\begin{align}\label{eq:mainconv}
\log M \leq nTC(P) - \sqrt{nTV(P)}Q^{-1}(\epsilon) + \delta'_n \sqrt{n}\,,
\end{align}
where $C(P)$ and $V(P)$ are defined in~\eqref{eq:capacity} and~\eqref{eq:vmindef}, respectively.
\end{theorem}
\begin{proof}
As usual, without loss of generality we may assume that all codewords belong to~$S_n$ as defined in~\eqref{eq:powersp},
see~\cite[Lemma 39]{PPV08}.
The maximal probability of error code size is bounded by a meta-converse theorem~\cite[Theorem~31]{PPV08}, which states that for any $(n,M,\epsilon)$ code and distribution $Q_{Y^nH^n}$ on the output space of the channel,
\begin{align}\label{eq:conv_bound}
\frac{1}{M} \geq \inf_{x^n} \beta_{1-\epsilon}(P_{Y^nH^n|X=x^n},Q_{Y^nH^n})\,,
\end{align}
where infimum is taken over all codewords.
The main problem is to select $Q_{Y^n H^n}$ appropriately. We do this separately for the two subcodes defined as
follows. Fix arbitrary $\delta>0$ (it will be taken to 0 at the end) and introduce:
\begin{align}
&\mathcal{C}_l \triangleq \mathcal{C} \cap \{x^n : V_n(x^n) \leq n(V(P) - \delta) \}\\
&\mathcal{C}_u \triangleq \mathcal{C} \cap \{x^n : V_n(x^n) > n(V(P) - \delta) \}\ .
\end{align}
To bound the cardinality of $\matc_u$, we select $Q_{Y^n H^n} = (P_{Y,H}^*)^n$ to be the $n$-product of the
caod~\eqref{eq:mimo_caod}, then apply the following estimate from~\cite[Lemma 14]{YP10}, quoted here: for any $\Delta>0$ we have
\begin{equation}
	\log\beta_{1-\epsilon}(P_{Y^nH^n|X=x^n}, P_{Y,H}^{* n}) \ge 
		 -n D_n(x^n) - \sqrt{nV_n(x^n)} Q^{-1}\left( 1 - \epsilon - {B_n(x^n)+\Delta\over
		\sqrt{n}}\right)  - {1\over2}\log {n\over \Delta^2}\,,\label{eq:beta_as_l}
\end{equation}		
where $D_n$, $V_n$ and $B_n$ are given by~\eqref{eq:dndef2},~\eqref{eq:cvx_nto1} and~\eqref{eq:bndef}, respectively.
We choose $\Delta = n^{1\over 4}$ and then from Lemma~\ref{lem:thirdmom} (which relies on the assumption that $\|x_j\|_F \leq \delta n^{\frac{1}{4}}$) we get that for some constants $K_1,K_2$ we
have for all $x^n \in \matc_u$:
$$ B_n(x^n) + \Delta \le K_1 \delta_n^2 \sqrt{n} + K_2 n^{1\over 4} + \frac{K_3}{n^{1/2}}\,.$$
From~\eqref{eq:conv_bound} and~\eqref{eq:beta_as_l} we therefore obtain
\begin{equation}\label{eq:gdf1}
	\log |\matc_u| \le nTC(P) - \sqrt{nT(V(P)-\delta)} Q^{-1}(\epsilon - \delta''_n) + {1\over 4} \log n\,,
\end{equation}
where $\delta''_n = {K_1 \delta_n^2 + K_2 n^{-{1\over 4}}} \to 0$ as $n\to \infty$.

Next we proceed to bounding $|\matc_l|$. To that end, we first state two lemmas.  Lemma \ref{lem:tildec} shows that, if in addition to the power constraint $\E[\|X\|_F^2] \leq TP$, we also required $\E[V_1(X)] \leq V(P) - \delta$, then the capacity of this variance-constrained channel is strictly less than without the latter constraint.

\begin{lemma}\label{lem:tildec} Consider the following constrained capacity:
\begin{equation}\label{eq:tildecp}
\tilde C(P,\delta) \eqdef \frac{1}{T} \sup_X \left\{ I(\mat{X};\mat{Y}|\mat{H}): \E[\|\mat{X}\|_F^2] \le TP,
\E[V_1(\mat{X})] \leq V(P)-\delta\right\}\,,
\end{equation}
where $V(P)$ is from~\eqref{eq:vmindef} and $V_1(x)$ is from~\eqref{eq:cvx1}. 
For any $\delta > 0$ there exists $\tau = \tau(P,\delta)>0$ such that $\tilde C(P,\delta)<C(P)-\tau$.
\end{lemma}
\begin{remark} Curiously, if we used constraint $\EE[V_1(X)] > V(P)+\delta$ instead of $\E[V_1(\mat{X})] \leq V(P)-\delta$ in \eqref{eq:tildecp}, then the resulting capacity equals $C(P)$ regardless of $\delta$.
\end{remark}

The following Lemma shows that, with the appropriate choice of an auxiliary distribution $Q_{Y^n,H^n}$, the expected size of the normalized log likelihood ratio is strictly smaller than capacity, while the variance of that same ratio is upper bounded by a constant (i.e. does not scale with $n$).

\begin{lemma}\label{lem:tildevar}
Define the auxiliary distribution
\begin{align}
Q_{Y|H}(y|h) =
\begin{cases}
P^*_{Y|H}(y | h) & \|h\|_F^2 > A\\
\tilde{P}^*_{Y|H}(y|h) & \|h\|_F^2 \leq A
\end{cases}
\end{align}
where $A > 1$ is a constant, $P^*_{Y|H}(y | h)$ is the caod for the MIMO-BF channel, and $\tilde{P}^*_{Y|H}(y|h)$ is the caod for the variance-constrained channel in \eqref{eq:tildecp}.  Let $Q_{Y,H} = P_HQ_{Y|H}$, and $Q_{Y^n,H^n} = \prod_{i=1}^n Q_{Y,H}$.  Then there exists constants $\tau, K > 0$ such that for all $x^n \in \mathcal{C}_l$,
\begin{align}
%C_n &\triangleq \frac{1}{n}\EE\left[\log\frac{P_{Y^n,H^n|X^n}}{Q_{Y^n,H^n}}(Y^n,H^n|X^n) \middle| X^n = x^n\right] \leq C(P)-\tau\label{eq:cnvar} \\
%V_n &\triangleq \frac{1}{n}\Var\left( \log\frac{P_{Y^n,H^n|X^n}}{Q_{Y^n,H^n}}(Y^n,H^n|X^n) \middle| X^n = x^n\right)\label{eq:vnvar} \leq K
C_n &\triangleq \frac{1}{nT}\EE\left[\log\frac{P_{Y^n,H^n|X^n}}{Q_{Y^n,H^n}}(Y^n,H^n|x^n) \right] \leq C(P)-\tau\label{eq:cnvar} \\
V_n &\triangleq \frac{1}{nT}\Var\left( \log\frac{P_{Y^n,H^n|X^n}}{Q_{Y^n,H^n}}(Y^n,H^n|x^n) \right)\label{eq:vnvar} \leq K
\end{align}
where $Y_i = H_ix_i + Z_i$, $i=1,\hdots,n$ is the joint distribution.
\end{lemma}

\begin{remark}
The reason we let $Q_{Y|H}$ take on two distributions depending on the value of $H$ is because we do not know the form of $\tilde{P}^*_{Y|H}$, hence we do not explicitly know how it depends on $H$.  This choice of $Q_{Y|H}$ ensures that expectations involving $\tilde{P}^*_{Y|H}$ are finite.
\end{remark}

Choose $Q_{Y,H}$ as in Lemma \ref{lem:tildevar}, so that the bounds on $C_n$, $V_n$ from \eqref{eq:cnvar}, \eqref{eq:vnvar} respectively, hold. Applying~\cite[Lemma 15]{YP10} with $\alpha = 1-\epsilon$ (the statement of this lemma is the contents of \eqref{eq:lemapp1}), we obtain
\begin{align}
\log \beta_{1-\epsilon}(P_{Y^n, H^n |X^n=x^n}, \tilde P_{Y,H}^{*n}) &\ge -nTC_n -\sqrt{\frac{2nTV_n}{1-\epsilon}} - \log\frac{1-\epsilon}{2}\label{eq:lemapp1}\\
&\ge -nT(C(P)-\tau) - \sqrt{2nTK\over 1-\epsilon}+\log {1-\epsilon\over 2}\,.
\end{align}
Therefore, from~\eqref{eq:conv_bound} we conclude that for all $n\ge n_0(\delta)$ we have
\begin{equation}\label{eq:gdf2}
	\log |\matc_l| \le nT\left(C(P)-{\tau\over 2}\right)\,.
\end{equation}
Overall, from~\eqref{eq:gdf1} and~\eqref{eq:gdf2} we get (due to arbitrariness of $\delta$) the
statement~\eqref{eq:mainconv}.
\end{proof}

\begin{proof}[Proof of Lemma~\ref{lem:tildec}]

Introduce the following set of distributions:
\begin{align}
\mathcal{P}' \triangleq \left\{ P_X : \E[\|X\|_F^2] \leq TP,\  \E[V_1(X)] \leq V-\delta \right\}\ .
\end{align}
By Prokhorov's criterion (e.g. \cite[Theorem 5.1]{billingsley2013convergence}, tightness implies relative compactness), the norm constraint implies that this set is relatively compact in the topology of weak
convergence. So there must exist a sequence of distributions $\tilde P_n \in \matp'$ s.t. $\tilde P_n\stackrel{w}{\to}\tilde P$ and
$I(\tilde X_n; H \tilde X_n+Z|H) \to \tilde C(P,\delta)$ where $\tilde X_n \sim \tilde P_n$. By Skorokhod representation \cite[Theorem 6.7]{billingsley2013convergence}, we may assume
$\tilde X_n \stackrel{a.s.}{\to} \tilde X \sim \tilde P$, i.e. there exists random variable $\tilde X$ that is the pointwise limit of the $\tilde{X}_n$'s. Notice that for any continuous bounded function $f(h,y)$ we
have
$$ \EE[f(H, H \tilde X_n + Z)] \to \EE[f(H, H \tilde X + Z)]\,,$$
and therefore $P_{\tilde Y_n, H} \stackrel{w}{\to} P_{\tilde Y,H}$. Assume (to arrive at a contradiction) that $\tilde
C(P,\delta)=C(P)$, then by the golden formula, cf.~\cite[Theorem 3.3]{PW2016notes}, we have
\begin{align} I(\tilde X_n; H \tilde X_n + Z|H) &= D(P_{Y H| X}  \| P_{Y,H}^* | P_{\tilde X_n}) - D(P_{\tilde Y_n, H}\|P_{Y,H}^*)\\
			  	   &= \EE[D_1(\tilde X_n)] - D(P_{\tilde Y_n, H}\|P_{Y,H}^*)\label{eq:gff1}\\
				   &\le C(P) - D(P_{\tilde Y_n, H}\|P_{Y,H}^*)\,,
\end{align}				   
where $D_1(x)$ is from~\eqref{eq:dndef2}. Therefore, we have
$$ D(P_{\tilde Y_n, H}\|P_{Y,H}^*) \to 0\,.$$
From weak lower-semicontinuity of divergence~\cite[Theorem 3.6]{PW2016notes} we have $D(P_{\tilde Y, H}\|P_{Y,H}^*)=0$.
In particular, if we denote $X^*$ to have Telatar distribution~\eqref{eq:telatar}, we must have 
\begin{equation}\label{eq:gff2}
	\EE[\|\tilde Y\|_F^2] = \EE[\|H \tilde X + Z\|_F^2] = \EE[\|H X^* + Z\|_F^2]\,.
\end{equation}
From Lemma~\ref{lem:var_w} (see Remark \ref{rem:exp_hx}) we have
\begin{equation}\label{eq:gff5}
	\EE[\|Hx\|_F^2] = {\EE[\|H\|_F^2]\over n_t} \|x\|_F^2
\end{equation}
and hence from the independence of $Z$ from $(H,X)$ we get 
$$ \EE[\|H \tilde X + Z\|_F^2] = {\EE[\|H\|_F^2]\over n_t} \EE[\|\tilde X\|_F^2] + n_r T\,,$$
and similarly for the right-hand side of~\eqref{eq:gff2}. We conclude that 
$$ \EE[\|\tilde X\|_F^2] = \EE[\|X^*\|_F^2] = TP\,.$$
Finally, plugging this fact into the expression for $D_1(x)$ in~\eqref{eq:dndef2} and~\eqref{eq:gff1} we obtain
$$ I(\tilde X; H\tilde X+Z|H) = \EE[D_1(\tilde X_n)] = C(P)\,.$$
That is, $\tilde X$ is a caid. But from Fatou's lemma we have (recall that $V_1(x)\ge 0$ since it is a variance)
$$ \EE[V_1(\tilde X)] \le \liminf_{n\to \infty} \EE[V_1(\tilde X_n)] \le V(P)-\delta\,,$$
where the last step follows from $\tilde P_n \in \matp'$. A caid achieving conditional variance strictly less than $V(P)$ contradicts
the definition of $V(P)$, cf.~\eqref{eq:vmindef}, as the infimum of $\EE[V_1(X)]$ over all caids.
\end{proof}

\begin{proof}[Proof of Lemma~\ref{lem:tildevar}] % Jensen method proof
First we analyze $C_n$ from \eqref{eq:cnvar}.  Denote
\begin{align}
i(x;y,h) &= \log\frac{P_{Y|H,X}}{P^*_{Y|H}}(y|h,x)\\
\tilde{i}(x;y,h) &= \log\frac{P_{Y|H,X}}{\tilde{P}^*_{Y|H}}(y|h,x)\,.
\end{align}
Here, $i(x;y,h)$ is the information density given by \eqref{eq:mimo_info_density}, while $\tilde{i}(x;y,h)$ instead has the caod for the variance-constrainted channel \eqref{eq:tildecp} in the denominator.  Since $Q_{Y|H}$ takes on one of two distributions based on the value of $H$, conditioning on $H$ in two ways yields
\begin{align}
C_n &= \frac{1}{nT}\E\left[\log\frac{P_{Y^n,H^n|X^n}}{Q_{Y^n,H^n}}(Y^n,H^n|x^n) \right]\label{eq:lememaster}\\
&= \frac{1}{nT}\sum_{j=1}^n \E\left[ i(x_j;Y_j,H_j) \middle| \|H_j\|_F^2 > A\right] \PP[\|H_j\|_F^2 > A]\label{eq:leme1}\\
&+ \frac{1}{nT}\sum_{j=1}^n \E\left[ \tilde{i}(x_j,Y_j,H_j) \middle| \|H_j\|_F^2 \leq A\right] \PP[\|H_j\|_F^2 \leq A]\label{eq:leme2}\,.
\end{align}
The $H_j$'s are i.i.d. according to $P_H$, so we define $p \triangleq \PP[\|H_j\|_F^2 > A]$. Using capacity saddle point, \eqref{eq:leme1} is bounded by
\begin{align}
\frac{p}{nT} \E\left[ \sum_{j=1}^n i(x_j;Y_j,H_j) \middle| \|H_j\|_F^2 > A\right]\label{eq:leme1bound}
&\leq pC(P_{H>A})
\end{align}
where $C(P_{H})$ denotes the capacity of the MIMO-BF channel with fading distribribution $P_H$, and $P_{H>A}$ denotes the distribution of $H$ conditioned on $\|H\|_F^2 > A$ (similarly, $P_{H\leq A}$ will denote $H$ conditioned on $\|H\|_F^2 \leq A$).  \eqref{eq:leme1bound} follows from the fact that the information density, i.e. $\log \frac{P_{Y|H,X}}{P^*_{Y|H}}(y|h,x)$, is not a function of $P_H$, hence changing the distribution $P_H$ does not affect the form of $i(x;y,h)$.  Similarly, using Lemma \ref{lem:tildec}, \eqref{eq:leme2} is bounded by
\begin{align}
\frac{1-p}{nT}\E\left[\sum_{j=1}^n \tilde{i}(X_j;Y_j,H_j) \middle| \|H_j\|_F^2 \leq A\right] &\leq (1-p)\tilde{C}(P_{H\leq A})\\ 
&= (1-p)(C(P_{H\leq A}) - \tau')\label{eq:leme1bound2}
\end{align}
where $\tau' > 0$ is a positive constant, and $\tilde{C}(P_H)$ denotes the solution to the optimization problem \eqref{eq:tildecp} when the fading distribution is $P_H$.  Putting together \eqref{eq:leme1bound} and \eqref{eq:leme1bound2}, we obtain an upper bound on $C_n$,
\begin{align}\label{eq:lemccomb}
C_n \leq pC(P_{H>A}) + (1-p)(C(P_{H\leq A}) - \tau')\,.
\end{align}
Note that $C(P_H) = \E_{P_H}\left[\log\det(I_{n_r} + P/n_tHH^T)\right]$, so the capacity only depends on $P_H$ through the expectation -- the expression inside is not a function of $P_H$ because the i.i.d. Gaussian caid achieves capacity for all isotropic $P_H$'s. Hence, by the law of total expectation, \eqref{eq:lemccomb} simplifies to
\begin{align}\label{eq:lemecapb}
C_n \leq C(P_H) - (1-p)\tau'\,.
\end{align}
Finally, we can upper bound $p$ using Markov's inequality as
\begin{align}
p = \PP[\|H_1\|_F^2 > A] \leq \frac{1}{A}
\end{align}
since $A > 1$.  Applying this bound to \eqref{eq:lemecapb}, we obtain
\begin{align}
C_n &\leq C(P_H) - (1-p)\tau'\\
&\leq C(P_H) - \left(1 - \frac{1}{A}\right)\tau'\,.
\end{align}
Defining $\tau \triangleq (1 - 1/A)\tau'$ completes the proof of \eqref{eq:cnvar}.\\

Next we analyze $V_n$ from \eqref{eq:vnvar}.  The strategy will be to decompose \eqref{eq:vnvar} into two terms depending on the value of $\|H\|_F^2$, then show that each term is upper bounded by $A_1 + A_2\sum_{j=1}^n \|x_j\|_F^4$, where $A_1, A_2$ are constants not depending on $x^n$.  Finally, we will show that $\sum_{j=1}^n\|x_j\|_F^4 = O(n)$ when $x^n \in \mathcal{C}_l$.  To this end,
\begin{align}
V_n &= \frac{1}{nT}\Var\left( \log\frac{P_{Y^n,H^n|X^n}}{Q_{Y^n,H^n}}(Y^n,H^n|x^n) \right)\\
&= \frac{1}{nT}\sum_{j=1}^n \Var\left( \log\frac{P_{Y,H|X}}{Q_{Y,H}}(Y_j,H_j|x_j) \right) \label{eq:lemv210}\\
&\leq \frac{1}{nT}\sum_{j=1}^n \E\left[ \left(\log\frac{P_{Y,H|X}}{Q_{Y,H}}(Y_j,H_j|x_j) \right)^2 \right] \label{eq:lemv211}
\end{align}
where \eqref{eq:lemv210} follows from the independence of the terms, and \eqref{eq:lemv211} is from the bound $\Var(X) \leq \E[X^2]$. Again we condition on $H$ in two ways,
\begin{align}
V_n &\leq \frac{p}{nT}\sum_{j=1}^n \E\left[i(x_j;Y_j;H_j)^2 \middle| \|H_j\|_F^2 > A\right] \label{eq:lemv1}\\
&+ \frac{1-p}{nT}\sum_{j=1}^n \E\left[ \tilde{i}(x_j;Y_j,H_j)^2 \middle| \|H_j\|_F^2 \leq A\right] \label{eq:lemv2}\,.
\end{align}
For the first term, \eqref{eq:lemv1}, we know the expression for $i(x;y,h)$ from \eqref{eq:mimo_info_density}, so we simply upper bound $i(x;y,h)^2$.  To this end,
\begin{align}
i(x;y,h)^2 &\leq 2\left(\frac{T}{2}\log\det\left(I_{n_r} + \frac{P}{n_t}hh^T\right)\right)^2 +  2 \left(\frac{\log e}{2}\sum_{j=1}^{n_{\min}}\frac{\lambda_j^2\|v_j^T x\|^2 + 2\lambda_j\inprod{v_j^T x}{\tilde{z}_j} - \frac{P}{n_t}\lambda_j^2\|\tilde{z}_j\|^2}{1+\frac{P}{n_t} \lambda_j^2}\right)^2\\
&\leq C_1 \|h\|_F^2 + C_2\|x\|_F^4 + C_3(\tilde{z}_j)\|x\|_F^2 + C_4(\tilde{z_j})\label{eq:lemvi2b}
\end{align}
where $C_1, C_2$ are non-negative constants, and $C_3(\tilde{z}_j), C_4(\tilde{z_j})$ are functions of only $\tilde{z}_j$ that have bounded moments.  This follows from:
\begin{itemize}
\item Bounding the first term via
\begin{align}
\left(\frac{T}{2}\log\det\left(I_{n_r} + \frac{P}{n_t}hh^T\right)\right)^2 \leq \log^2(e) \frac{PT^2}{4n_t}n_{\min} \|h\|_F^2\,,
\end{align}
which can be derived from the basic inequality $\log(1+x) \leq \log(e)\sqrt{x}$.
\item Noting that the second term is bounded in $h$, since for all $\lambda \in \mathbb{R}$,
\begin{align}
\frac{|\lambda|}{1 + \frac{P}{n_t}\lambda^2} &\leq \frac{1}{2\sqrt{\frac{P}{n_t}}}\\
\frac{\lambda^2}{1 + \frac{P}{n_t}\lambda^2} & \leq \frac{n_t}{P}\,.
\end{align}
\item Noting that all moments of $\|\tilde{z}_j\|^2$ are finite because this is the norm of a standard normal vector.
\end{itemize}
Therefore, after taking the expectation of \eqref{eq:lemvi2b} and summing over all $n$, we obtain
\begin{align}
\frac{p}{nT}\sum_{j=1}^n \E\left[i(x_j;Y_j;H_j)^2 \middle| \|H_j\|_F^2 > A\right] 
\leq \frac{1}{nT}\left( C_5 + C_6\sum_{j=1}^n \|x_j\|_F^4\right)\label{eq:lemvib}
\end{align}
for some non-negative constants $C_5, C_6$.

To bound the second term, \eqref{eq:lemv2}, first we split the logarithm as
\begin{align}
\E&\left[ \tilde{i}(x_j;Y_j,H_j)^2 \middle| \|H_j\|_F^2 \leq A\right]\\
&\leq 2\E\left[\log\left( P_{Y|H,X}(Y_j|H_j,x_j) \right)^2  \middle| \|H_j\|_F^2 \leq A \right] + 2\E\left[\log\left( \tilde{P}^*_{Y|H}(Y_j|H_j) \right)^2  \middle| \|H_j\|_F^2 \leq A\right]\label{eq:lemv2terms}
\end{align}
The first term in \eqref{eq:lemv2terms} is simple to handle, since its expression is given by the definition of the channel,
\begin{align}
\E\left[\log\left( P_{Y|H,X}(Y_j|H_j,x_j) \right)^2  \middle| \|H_j\|_F^2 \leq A \right] &= \E\left[\left( -\frac{n_rT}{2}\log(2\pi) - \frac{1}{2}\|Z_j\|_F^2 \right)^2 \right]\\&\leq \frac{1}{2}n_rT\log^2(2\pi) + \frac{1}{2}n_rT(2 + n_r T)\\
&\triangleq K_1\label{eq:lemvt1b}
\end{align}
i.e. we have a constant upper bound.  For the second term in \eqref{eq:lemv2terms}, notice that $\tilde{P}^*_{Y,H}$ that is inducible through channel, i.e. there exists an input distribution $P_X$ such that $\tilde{P}^*_{Y,H}(y,h) = \E[P_{Y,H|X}(y,h|X)]$.  Using this fact, we obtain the bound
\begin{align}
-\log \tilde{P}^*_{Y|H}(y|h) &= -\log \E[P_{Y|H,X}(y|h,X)]\label{eq:vl211}\\
&\leq \E[-\log P_{Y|H,X}(y|h,X)]\label{eq:vl212}\\
&= \E\left[ \frac{n_r T}{2}\log(2\pi) + \frac{1}{2}\|y - hX\|_F^2\right]\label{eq:vl213}\\
&\leq \frac{n_r T}{2}\log(2\pi) + \|y\|_F^2 + TP\|h\|_F^2\label{eq:vl214}
\end{align}
where~\eqref{eq:vl212} follows from Jensen's inequality, \eqref{eq:vl213} is from the definition of the channel, and \eqref{eq:vl214} follows from applying the inequality $\|A + B\|_F^2 \leq 2\|A\|_F^2 + 2\|B\|_F^2$ along with $\|hX\|_F^2 \leq \|h\|_F^2\|X\|_F^2$, then noting that $X$ satisfies $\E[\|X\|_F^2] = TP$.    Using this, we can bound the second term in \eqref{eq:lemv2terms} via
\begin{align}
\E&\left[ \log\left( \tilde{P}^*_{Y|H}(Y_j|H_j) \right)^2 \middle| \|H_j\|_F^2 \leq A\right]\\
&\leq \E\left[ \left(\frac{n_r T}{2}\log(2\pi) + \|Y_j\|_F^2| + TP\|H_j\|_F^2\right)^2 \middle| \|H_j\|_F^2 \leq A \right]\label{eq:lemv312}\\
&\leq \E\left[ 3\frac{n_r^2 T^2}{4}\log^2(2\pi) + 3\|Y_j\|_F^4 + 3T^2P^2\|H_j\|_F^4 \middle| \|H_j\|_F^2 \leq A \right]\label{eq:lemv313}\\
&\leq K_2 + K_3\|x\|_F^4\label{eq:lemv314}
\end{align}
where $K_2,K_3$ are non-negative constants which do not depend on $x$, \eqref{eq:lemv312} is from the above bound \eqref{eq:vl214}, and \eqref{eq:lemv314} follows from applying the bound
\begin{align}
\E\left[\|Y_j\|_F^4 \middle| \|H_j\|_F^2 \leq A\right] &= \E\left[\|H_jx_j + Z_j\|_F^4 \middle| \|H_j\|_F^2 \leq A \right]\\
&\leq 8\E\left[\|H_j\|_F^4 \middle| \|H_j\|_F^2 \leq A\right] \|x_j\|_F^4 + 16n_r^2T^2\\
&\leq 8A\|x_j\|_F^4 + 16n_r^2T^2\,.
\end{align}
Putting together \eqref{eq:lemv314} and \eqref{eq:lemvt1b}, we obtain an upper bound on \eqref{eq:lemv2},
\begin{align}
\frac{1-p}{nT}\sum_{j=1}^n \E&\left[ \tilde{i}(x_j;Y_j,H_j)^2 \middle| \|H_j\|_F^2 \leq A\right] 
\leq \frac{2(1-p)}{nT}\left(K_3 + K_4 + K_5\sum_{j=1}^n \|x_j\|_F^4\right)\label{eq:lemvtildeb}\,.
\end{align}
Now, since $x^n \in \mathcal{C}_l$ by assumption, we can control the quantity $\sum_{i=1}^n \|x_i\|_F^4$ via
\begin{align} 
\sum_{i=1}^n \|x_i\|_F^4 &\leq \sum_{i=1}^n V_1(x_i)\\ 
&\leq n(V(P)-\delta)\,,
\end{align}
where the first inequality follows from the non-negativity of the terms in $V_1(x)$ given in Proposition \ref{prop:cond_var_x}, and the second inequality is from the definition of $\mathcal{C}_l$.  Hence the sum of fourth powers of the $\|x_i\|_F$'s is $O(n)$ on $\mathcal{C}_l$.  All together, combining \eqref{eq:lemvtildeb} and \eqref{eq:lemvib} yields the following bound on $V_n$,
\begin{align}
V_n &\leq \frac{1}{n}\left(K' + K''\sum_{j=1}^n\|x_i\|_F^4\right)\\
&\leq K
\end{align}
which completes the proof of \eqref{eq:vnvar}.
\end{proof}

%%%%%%%%%%%%%%%%%%%%%%%%%%%%%%%%%%%%%%%%%%%%%%%%%%%%%%%%%%%%%%%%%%
%%%%% MISO %%%%%%%%%%%%%%%%%%%%%%%%%%%%%%%%%%%%%%%%%%%%%%%%%%%%%%%
%%%%%%%%%%%%%%%%%%%%%%%%%%%%%%%%%%%%%%%%%%%%%%%%%%%%%%%%%%%%%%%%%%

\section{The rank 1 case}\label{sec:MISO}

When $H$ is rank 1, for example in the MISO case, i.e. $n_t > n_r = 1$, the MIMO-BF channel has multiple input distributions that achieve capacity, as shown in Theorem~\ref{thm:caid_conds}.  Theorem~\ref{thm:mimo_dispersion} proved that the dispersion in the general MIMO-BF channel is given by~\eqref{eq:vmindef}, where we minimize the conditional variance of the information density over the set of caids.  In this section, we analyze those minimizers for the rank 1 case, which turns out to be non-trivial.

From Theorem~\ref{thm:rank1_disp}, when $H$ is rank 1, the conditional variance takes the form
\begin{align}
V(P) = K_1 - K_2 v^*(n_t,T)
\end{align}
where $K_1,K_2 > 0$ are constants that depend on the channel parameters but not the input distribution.  From~\eqref{eq:v_star_minimization}, computing $v^*(n_t, T)$ requires us to maximize the variance of the squared Frobenius norm of the input distribution over the set of caids.  Intuitively, this says that minimizing the dispersion is equivalent to maximizing the amount of correlation amongst the entries of $X$ when $X$ is jointly Gaussian. In a sense, this asks for the capacity achieving input distribution having the  least amount of randomness.

Here we characterize $v^*(n_t,T)$.  The manifold of caids is not easy to optimize over, since one must account for all
the independence constraints on the rows and columns, the covariance constraints on the $2\times 2$ minors, positive
definite constraints, etc. as described in Theorem~\ref{thm:caid_conds}.  Our strategy instead will be to give an upper
bound on $v^*(n_t,T)$, then show that for certain pairs $(n_t,T)$, the upper bound is tight.  Before stating the main
theorem of the section, we review orthogonal designs, which will play a large role in the solution to this problem.

\subsection{Orthogonal designs}\label{sec:od}

\begin{defn}[Orthogonal Design]
A real $n\times n$ orthogonal design of size $k$ is defined
to be an $n\times n$ matrix $A$ with entries given by linear forms in $x_1, \hdots, x_k$ and coefficients in $\mreals$ satisfying
\begin{equation}\label{eq:orthog}
A^TA = \left(\sum_{i=1}^k x_i^2\right) I_n
\end{equation}
\end{defn}

In other words, all columns of $A$ have squared Euclidean norm $\sum_{i=1}^k x_i^2$, and all columns are pairwise
orthogonal.  A common representation for an orthogonal design is the sum $A = \sum_{i=1}^k x_iV_i$ where
$\{V_1,\hdots,V_k\}$ is a collection of $n\times n$ real matrices satisfying Hurwitz-Radon
conditions~\eqref{eq:hr1}-\eqref{eq:hr2}.
Such collection is called a Hurwitz-Radon family.  Theorem~\ref{thm:radon-hurwitz} shows that the maximal cardinality of
a Hurwitz-Radon family is the \emph{Hurwitz-Radon number} $\rho(n)$, cf.~\eqref{eq:rho_funct}.

The definition of orthogonal designs can be generalized to rectangular matrices~\cite{tarokh1999space}, as follows:

\begin{defn}[Generalized Orthogonal Design]\label{def:god}
A generalized orthogonal design is a $p\times n$ matrix $A$ with $p\geq n$ with entries as linear forms of the
indeterminates $\{x_1,\hdots,x_k\}$ satisfying~\eqref{eq:orthog}.
\end{defn}
The quantity $R=k/p$ is often called the rate of the generalized orthogonal design.  This term is justify by noticing
that if $p$ represents a number channel uses and $k$ represents the number of data symbols, then $R$ represents sending
$k$ data symbols in $p$ channel uses. In this work, we are only interested in the case $R=1$ (i.e. $k=p$), called
\emph{full-rate orthogonal designs}.
Full-rate orthogonal design can be constructed from a Hurwitz-Radon family $\{V_1,\hdots,V_n\}$, each $V_i \in
\mathbb{R}^{k\times k}$ by forming the matrix $A$ 
\begin{align}\label{eq:god_a}
A = \left[ V_1x \ \cdots \ V_nx \right] 
\end{align}
where $x = [x_1,\hdots, x_k]^T$ is the vector of indeterminates. It follows immediately from this construction that~\eqref{eq:orthog} is satisfied.  Theorem~\ref{thm:radon-hurwitz} allows us to conclude that a generalized full rate $n\times k$ orthogonal design exists if and only if $n \leq \rho(k)$.

The following proposition shows that full rate orthogonal designs correspond to caids in the MIMO-BF channel.

\begin{proposition}\label{prop:od_to_caid}
Take $n_t = \rho(T)$ and a maximal  Hurwitz-Radon family $\{V_i, i=1,\ldots,n_t\}$ of $T\times T$ matrices (cf. Theorem~\ref{thm:radon-hurwitz}). Let $\xi \sim \mathcal{N}(0, P/n_t I_T)$ be an i.i.d. row-vector. Then the input distribution
\begin{equation}\label{eq:vrep}
    X = \left[ V_1^T \xi^T\,\, \cdots V_{n_t}^T \xi^T \right]^T 
\end{equation}
achieves capacity for any MIMO-BF channel provided $\PP[\rank H\le 1]=1$.
\end{proposition}

\begin{proof}
Since $\{V_1,\hdots,V_{n_t}\}$ is a Hurwitz-Radon family, they satisfy~\eqref{eq:hr1}-\eqref{eq:hr2}.
Form $X$ as in~\eqref{eq:vrep}.  Then each row and column is jointly Gaussian, and applying the caid conditions~\eqref{eq:row1} and~\eqref{eq:row2} from Theorem~\ref{thm:caid_conds} shows,
\begin{align}
&\E[R_i^TR_i] = V_i^T \E[\xi^T \xi] V_i = \frac{P}{n_t}V_i^TV_i =          \frac{P}{n_t}I_T\\
&\E[R_i^TR_j] = V_i^T  \E[\xi^T \xi] V_j = \frac{P}{n_t}V_i^TV_j = - \frac{P}{n_t}V_j^TV_i = -        \E[R_j^TR_i]
\end{align}
Therefore $X$ satisfies the caid conditions, and hence achieves capacity.  
\end{proof}

\begin{remark}
The above argument implies that if $X \in \mathbb{R}^{n_t\times T}$ is constructed above, then removing the last row of
$X$ gives an $(n_t - 1) \times T$ input distribution that also achieves capacity.
\end{remark}

\subsection{Proof of theorem~\ref{th:vstar}}

Theorem~\ref{th:vstar} states that for dimensions where orthogonal designs exist, the conditional variance \eqref{eq:vmindef} is minimized
if and only if the input is constructed from an orthogonal design as in Proposition~\ref{prop:od_to_caid}. The approach is
first to prove an upper bound on $v^*$, then show that conditions for tightness of the upper bound correspond to
conditions of the Hurwitz-Radon theorem.

We start with a simple lemma, which will be applied with $A,B$ equal to the rows of the capacity achieving input $X$.
\begin{lemma}\label{lem:proj_lemma}
Let $A = (A_1,\hdots,A_n)$ and $B = (B_1,\hdots,B_n)$ each be i.i.d. random vectors from the same distribution with finite second moment $\E[A_1^2] = \sigma^2 < \infty$.  While $A$ and $B$ are i.i.d. individually, they may have arbitrary correlation between them.  Then
\begin{align}\label{eq:cov_sum_lemma}
\sum_{i=1}^n \sum_{j=1}^n \Cov(A_i,B_j) \leq n\sigma^2
\end{align}
with equality iff $\sum_{i=1}^n A_i = \sum_{i=1}^n B_i$ almost surely.
\end{lemma}
\begin{proof}
Simply use the fact that covariance is a bilinear function, and apply the Cauchy-Schwarz inequality as follows:
\begin{align}
\sum_{i=1}^n \sum_{j=1}^n \Cov(A_i,B_j) &= \Cov\left(\sum_{i=1}^n A_i,  \sum_{j=1}^n B_j\right)\\
&\leq \sqrt{\Var\left(\sum_{i=1}^n A_i\right)\Var\left(\sum_{j=1}^n B_j\right)}\\
&= \sqrt{(n\Var(A_1))(n\Var(B_1))}\\
&= n\sigma^2
\end{align}
We have equality in Cauchy-Schwarz when $\sum_{i=1}^n A_i$ and $\sum_{i=1}^n B_i$ are proportional, and since these sums have the same distribution, the constant of proportionality must be equal to 1, so we have equality in~\eqref{eq:cov_sum_lemma} iff $\sum_{i=1}^n A_i = \sum_{i=1}^n B_i$ almost surely.
\end{proof}
%Geometrically, if we have two orthonormal bases $\{a_1,\hdots,a_n\}$ and $\{b_1,\hdots,b_n\}$ for a vector space, then

\begin{proof}[Proof of Theorem~\ref{th:vstar}] First, we rewrite $v^*(n_t,T)$ defined in~\eqref{eq:v_star_minimization}
as
\begin{align} \label{eq:corr_sum2}
v^*(n_t,T) \triangleq {n_t^2\over 2P^2} \max_{P_X : I(X;Y|H) = C} \sum_{i=1}^{n_t}\sum_{j=1}^{n_t}\sum_{k=1}^T\sum_{l=1}^T 
\Cov(X_{i,k}^2,X_{j,l}^2)
\end{align}
From here, $v^*(n_t,T) = v^*(T,n_t)$ follows from the symmetry to transposition of the caid-conditions on $X$ (see Theorem~\ref{thm:caid_conds})
and symmetry to transposition of~\eqref{eq:corr_sum2}.  From now on, without loss of generality we assume $n_t \le T$.

For the upper bound, since the rows and columns of $X$ are i.i.d., we can apply Lemma \ref{lem:proj_lemma} with
$A_i=X_{i,k}^2$ and $B_j = X_{j,l}^2$  (and hence $\sigma^2 = 2(P/n_t)^2$) to get 
\begin{equation}\label{eq:cbound}
	\sum_{i,j,k,l} \Cov(X_{i,k}^2,X_{j,l}^2) \le \sum_{i,j} 2T (P/n_t)^2 = 2n_t^2 T (P/n_t)^2\,,
\end{equation}
which together with~\eqref{eq:corr_sum2} yields the upper bound~\eqref{eq:collins_sum} (recall that $n_t \le T$).

Equation~\eqref{eq:cbound} implies that if $X$ achieves the bound~\eqref{eq:collins_sum}, then removing
the last row of $X$ achieves~\eqref{eq:collins_sum} as an $(n_t-1)\times T$ design. In other words, if~\eqref{eq:collins_sum} is
tight for $n_t\times T$ then it is tight for all $n_t' \le n_t$.

Notice that for any $X$ such that any
pair $X_{i,k}$,$X_{j,l}$ is jointly Gaussian, we have
\begin{equation}\label{eq:corr_sum1}
	{n_t^2\over 2P^2} \Var(\|X\|_F^2)  = \sum_{i,j,k,l} \rho_{ikjl}^2\,,
\end{equation}
where
\begin{align}
\rho_{ikjl} \eqdef {n_t\over P} \Cov(X_{ik}, X_{jl})\,.
\end{align}
Take $X \in \mathbb{R}^{n_t\times T}$ as constructed in~\eqref{eq:vrep}.  
By Proposition~\ref{prop:od_to_caid}, $X$ is capacity achieving and identity~\eqref{eq:corr_sum1} clearly holds.  In the representation~\eqref{eq:vrep}, the matrix $V_j^TV_i$ contains the correlation coefficients between rows $i$ and $j$ of $X$, since $\E[(\xi V_j)^T(\xi V_i)] = \frac{P}{n_t} V_j^TV_i$, so
\begin{align}
\|V_j^TV_i\|_F^2 = \sum_{k=1}^T\sum_{l=1}^T \rho_{ikjl}^2\,.
\end{align}
%Notice also that for 
Therefore we can represent the sum of squared correlation coefficients as
\begin{align}\label{eq:vrep2}
\sum_{i,j,k,l} \rho_{ijkl}^2 &= \sum_{i=1}^{n_t}\sum_{j=1}^{n_t} \|V_j^TV_i\|_F^2\\
&= \sum_{i=1}^{n_t}\sum_{j=1}^{n_t} \tr\left(V_jV_j^TV_iV_i^T\right)\\
&= \tr\left(\left(\sum_{i=1}^{n_t} V_i V_i^T\right)^2\right)\label{eq:bound_pf_trace}\\
&= n_t^2\,. T\label{eq:apply_hr_conds}
\end{align}
Line~\eqref{eq:apply_hr_conds} follows since the $V_i$'s are orthogonal by the Hurwitz-Radon condition, so each $V_i
V_i^T = I_T$ in the summation in~\eqref{eq:bound_pf_trace}.  Hence the $X$ constructed in~\eqref{eq:vrep} achieves the
upper bound in~\eqref{eq:cbound} and~\eqref{eq:collins_sum}.

Next we prove~\eqref{eq:collins_nontight}. Suppose $X$ is a jointly-Gaussian caid saturating the bound~\eqref{eq:cbound}. 
From Lemma~\ref{lem:proj_lemma}, the condition for equality
in~\eqref{eq:cov_sum_lemma} implies that for all $j \in \{1,\hdots,n_t\}$,
\begin{align}\label{eq:norm_equality}
\|R_j\|_F^2 = \|R_1\|_F^2 \ \ a.s.
\end{align}
where $R_j$ is the $j$-th row of $X$ for $j=1,\hdots,n_t$. In particular, this means that every $R_j$ is a linear
function of $R_1$. Consequently, we may represent $X$ in terms of a row-vector $\xi\sim \matn(0, P/n_t I)$ as
in~\eqref{eq:vrep}, that is $R_j = \xi V_j$ for some $T\times T$ matrices $V_j,j\in[n_t]$. We clearly have 
$$ \EE[R_i^T R_j] = {P\over n_t} V_i^T V_j\,.$$
But then the caid constraints~\eqref{eq:row1}-\eqref{eq:row2} imply that the matrix $A$ in~\eqref{eq:god_a} constructed
using indeterminates $\{x_1,\ldots,x_{n_t}\}$ and family $\{V_1,\ldots, V_{n_T}\}$ satisfies
Definition~\ref{def:god}. Therefore, from Theorem~\ref{thm:radon-hurwitz}, (see also~\cite[Proposition
4]{liang2003orthogonal}), we must have $n_T \le \rho(T)$.
\end{proof}
\begin{remark} In the case $n_t=T=2$ it is easy to show that for any non-jointly-Gaussian caid, there exists a
jointly-Gaussian caid achieving the same $\Var(\|X\|_F^2)$. Indeed, consider ~\eqref{eq:2x2_gsn_caids} with
$\rho={\cov(X_{1,1}^2, X_{2,2}^2) + \cov(X_{1,2}^2, X_{2,1}^2)\over 8 (P/n_t)^2}$. If this phenomena held in general, we would
conclude that~\eqref{eq:collins_tight} holds if and only if $n_t \le \rho(T)$ or $T \le \rho(n_t)$. As a step towards
the proof of the latter, we notice that any caid $X$ achieving equality in~\eqref{eq:cbound} satisfies
\begin{equation}\label{eq:xxt_miso}
		X X^T = {\|X\|_F^2\over n_t} I_{n_t} \qquad \mbox{(a.s.)}\,,
\end{equation}	
which is equivalent to saying $R_i R_j' = 0$ for $i\neq j$. The latter follows from applying~\eqref{eq:norm_equality} to
rows of $UX$, where $U$ is an arbitrary orthogonal matrix. Identity~\eqref{eq:xxt_miso} could be informally stated as ``any caid
saturating~\eqref{eq:cbound} is a random full-rate orthogonal design''. 
\end{remark}

In summary, the full-rate orthogonal designs (when those exist) achieve the optimal channel dispersion $V(P)$. Some
examples ($\xi_j$ are i.i.d. $\matn(0,1)$) for $n_t = T = 4$ and $n_t = 4, T = 3$, respectively, are as follows:
\begin{align} \label{eq:4x4}
&X = \sqrt{P\over 4}
\left[ \begin{array}{cccc}
\xi_1 & \xi_2 & \xi_3 & \xi_4\\
-\xi_2 & \xi_1 & -\xi_4 & \xi_3\\
-\xi_3 & \xi_4 & \xi_1  & -\xi_2\\
-\xi_4 & -\xi_3 & \xi_2 & \xi_1
\end{array} \right]\\
&X = \sqrt{P\over 4}
\left[ \begin{array}{cccc}
\xi_1 & \xi_2 & \xi_3 \\
-\xi_2 & \xi_1 & -\xi_4 \\
-\xi_3 & \xi_4 & \xi_1  \\
-\xi_4 & -\xi_3 & \xi_2 
\end{array} \right] \nonumber
\end{align}

\subsection{Beyond full-rate orthogonal designs}

For pairs $(n_t,T)$ where $n_t > \rho(T)$, full-rate orthogonal design do not exist. For example $\rho(3)=1$, so no
full-rate orthogonal design exits for $n_t = 2$, $T=3$. Which caids are minimizer for~\eqref{eq:vmindef} in this case?  
In general, we do not know the answer and do not even know whether one can restrict the search to
jointly-Gaussian caids. But one thing is certain: it is
definitely not an i.i.d. Gaussian (Telatar) caid. To show this claim, we will give a method for constructing improved
caids.

To that end, suppose that $X$ consists of entries $\pm \xi_j$, $j=1\hdots, d$, where $\xi_j
\stackrel{i.i.d.}{\sim}\matn(0, P/n_t)$. Then we have:
\begin{equation}\label{eq:sumrho2}
	{n_t^2 \over 2P} \Var(\|X\|_F^2) = \sum_{t=1}^d (\ell_t)^2\,, 
\end{equation}
where $\ell_t$ is the number of times $\pm\xi_t$ appears in the description of $X$. By this observation and the remark
after Theorem~\ref{thm:caid_conds} (any submatrix of a caid $X$ is also a caid), we can obtain lower bounds on $v^*(n_t, T)$ for $n_t > \rho(T)$ via the following
\textit{truncation construction}:
\begin{enumerate}
    \item Take $T'>T$ such that $\rho(T')\ge n_t$ and let $X'$ be a corresponding $\rho(T')\times T'$ full-rate
    orthogonal design with entries $\pm \xi_1,\ldots \pm \xi_{T'}$.
    \item Choose an $n_t \times T$ submatrix of $X'$ maximizing the sum of squares of the number of occurrences of
    each of $\xi_j$, cf.~\eqref{eq:sumrho2}.
\end{enumerate}

As an example of this method, by truncating a $4\times4$ design~\eqref{eq:4x4} we obtain the following $2\times3$ and
$3\times3$
submatrices:
\begin{align}\label{eq:v3x3}
X = \sqrt{\frac{P}{3}}
\left[ \begin{array}{ccc}
\xi_1 & \xi_2 & \xi_3\\
-\xi_2 & \xi_1 & \xi_4\\
-\xi_3 & -\xi_4 & \xi_1
\end{array} \right]
\ \ \ 
X = \sqrt{\frac{P}{2}}
\left[ \begin{array}{ccc}
\xi_1 & \xi_2 & \xi_3\\
-\xi_2 & \xi_1 & \xi_4
\end{array} \right]
\end{align}
By independent methods we were able to show that designs~\eqref{eq:v3x3} are dispersion-optimal out of all jointly
Gaussian caids. Note that in these cases~\eqref{eq:collins_tight} does not hold,
illustrating~\eqref{eq:collins_nontight}.

\begin{table}
\centering%
\footnotesize
\caption{Values for $v^*(n_t,T)$}
\label{tab:table1}
\begin{tabular}{| c || c | c | c | c | c | c | c | c |}
\hline
  $n_t \setminus T$ & 1 & 2 & 3 & 4 & 5 & 6 & 7 & 8\\
\hline\hline
 1 & 1 & 2 & 3 & 4 & 5 & 6 & 7 & 8\\
\hline
 2 &  & 8 & $10^*$ & 16 & 18 & 24 & 26 & 32\\
\hline
 3 &  &  & $21^*$ & 36 & [39,45]  & [46,54]  & [57,63]  & 72\\
\hline
 4 &  &  &  & 64 & [68,80]  & [80,96]  & [100,112]  &  128\\
\hline
 5 &  &  &  &  & [89,125] & [118,150]  & [155,175] &  200\\
\hline
 6 &  &  &  &  &  & [168,216] & [222,252]  &  288\\
\hline
 7 &  &  &  &  &  &  & [301,343] &  392\\
\hline
 8 &  &  &  &  &  &  &  &  512\\
\hline
\end{tabular}
\\[5pt]
Note: Table is symmetric about diagonal; intervals $[a,b]$ mark entries for which dispersion-optimal input is
unknown. The optimality of entries marked with $*$ is only established in the class of all jointly-Gaussian caids.
\end{table}

\iffalse
\begin{table}
    \centering
\footnotesize
\caption{Values of $V_{min}\over C^2$ (when known) at $SNR=20~dB$}
\label{tab:table2}
\begin{tabular}{| c || c | c | c | c | c | c | c | c |}
\hline
  $n_t \setminus T$ & 1 & 2 & 3 & 4 & 5 & 6 & 7 & 8\\
\hline\hline
 1 & 0.38 & 0.60 & 0.82 & 1.05 & 1.27 & 1.49 & 1.72 & 1.94 \\
\hline
 2 & 0.35 & 0.39 & 0.52 & 0.60 & 0.73  & 0.82 & 0.94 & 1.03\\
\hline
 3 & 0.38 & 0.42  & 0.45 & 0.49 &  &  &  & 0.79 \\
\hline
 4 & 0.42 & 0.43 & 0.44 & 0.45 &  &  &  & 0.69 \\
\hline
 5 & 0.46 & 0.48 &  &  &  &  &  & 0.64 \\
\hline
 6 & 0.50 & 0.51 &  &  &  &  &  & 0.62 \\
\hline
 7 & 0.54 & 0.55 &  &  &  &  &  & 0.61 \\
\hline
 8 & 0.59 & 0.59 & 0.59 & 0.60 & 0.60 & 0.60 & 0.61 & 0.61\\
\hline
\end{tabular}
\end{table}
\fi

Our current knowledge about $v^*$ is summarized in Table~\ref{tab:table1}. The lower bounds for cases
not handled by Theorem~\ref{th:vstar} were computed by truncating the 8x8 orthogonal                            design~\cite[(5)]{tarokh1999space}.
Based on the evidence from $2\times T$ and $3\times 3$ we \textit{conjecture} this construction to be optimal.

%Finally, returning to the original question of the minimal delay required to achieve capacity, see~\eqref{eq:minblock}, we
%calculate the value of $V_{min}\over C^2$ in Table~\ref{tab:table2}.

From the proof of Theorem~\ref{th:vstar} it is clear that Telatar's i.i.d. Gaussian is
never dispersion optimal, unless $n_t=1$ or $T=1$. Indeed, for Telatar's input $\rho_{ikjl}=0$ unless $(i,k)=(j,l)$.
Thus embedding even a single $2 \times 2$ Alamouti block into an otherwise i.i.d. $n_t \times T$ matrix $X$ strictly
improves the sum~\eqref{eq:corr_sum2}.
We note that the value of ${V\over C^2}$ entering~\eqref{eq:minblock} can be quite sensitive to the suboptimal  choice of
the design. For
example, for $n_t=T=8$ and $SNR=20~dB$ estimate~\eqref{eq:minblock} shows that one needs
\begin{itemize}
    \item around 600 channel inputs (that is 600/8 blocks) for the optimal $8\times 8$ orthogonal design, or
    \item around 850 channel inputs for Telatar's i.i.d. Gaussian design
\end{itemize}
in order to achieve 90\% of capacity. This translates into a 40\% longer delay or battery spent in running the
decoder.

Thus, curiously even in cases where pure multiplexing (that is maximizing transmission rate) is needed -- as is  often the
case in modern cellular networks -- transmit diversity enters the picture by enhancing the finite blocklength   fundamental
limits. Remember, however, that our discussion pertains only to cases when the transmitter (base-station)  is  equipped with
more antennas than the receiver (user equipment), or when the channel does not have more than one diversity branch.

In cases when
full-rate designs do not exist, there have been various suggestions as to what could be the best solution,
e.g.~\cite{liang2003orthogonal}. Thus for non full-rate designs the property of minimizing dispersion (such
as~\eqref{eq:v3x3}) could be used for selecting the best design for cases $n_t >\rho(T)$.

\section{Discussion}\label{sec:discussion}

\begin{figure}[ht]
\centering
\includegraphics[width=.8\textwidth]{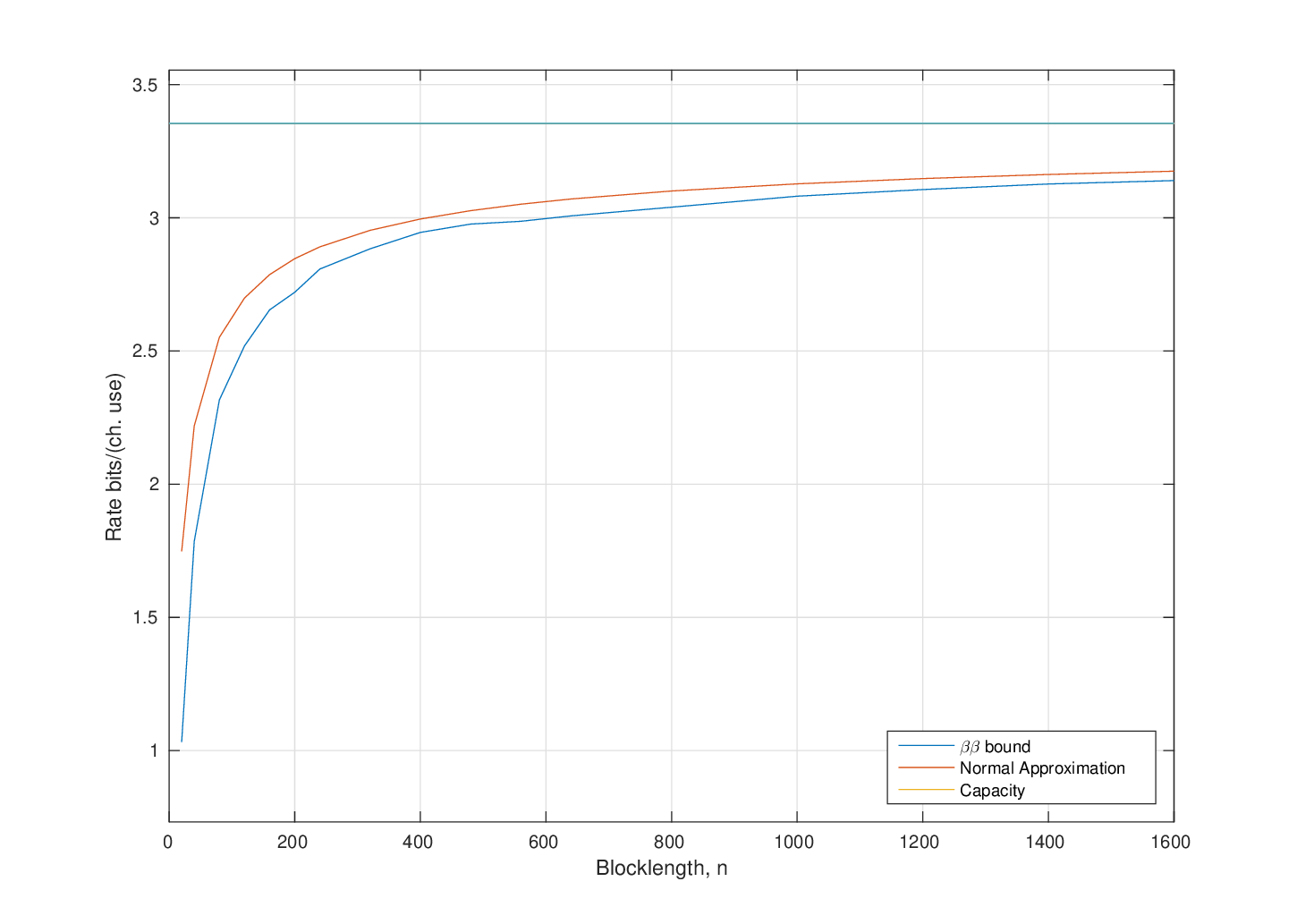}
\caption{Achievability and normal approximation for $n_t=n_r=T=4$, $P=0$dB, and $\epsilon=10^{-3}$.}\label{fig:achieve}
\end{figure}

Figure~\ref{fig:achieve} plots the capacity, normal approximation, and $\beta\beta$ achievability bound for the MIMO channel with $n_t = n_r = T = 4$ for the complex case.  The details of this computation are given in~\cite{YAGP16-bb}.  The $\beta\beta$ bound was developed by Yang et al~\cite{YAGP16-bb} and is often more computationally friendly than the $\kappa\beta$ bound.  This figure illustrates the gap between achievability and the normal approximation, as well as the gap to capacity.  For example, at blocklength 400, we can achieve about 88\% of capacity, and at blocklength 1000 we can achieve about about 92\% of capacity, given $P=0$dB and tolerating an error
probability of $10^{-3}$.

\begin{figure}[ht]
\centering
\includegraphics[width=.8\textwidth]{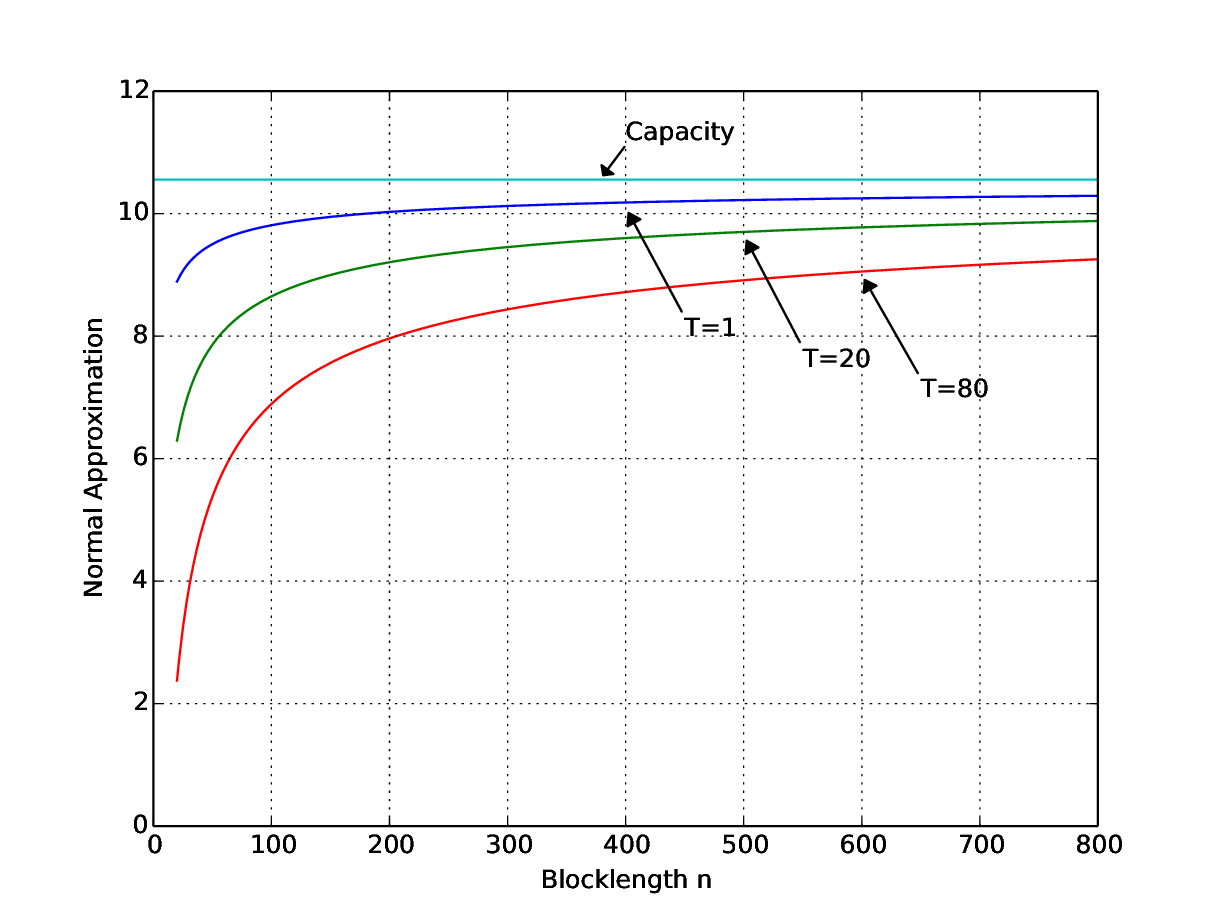}
\caption{The normal approximation for varying coherent times, with $n_t =  n_r = 4$, $P=20dB$, and $\epsilon = 10^{-3}$}\label{fig:norm_appx_coherence}
\end{figure}

Figure~\ref{fig:norm_appx_coherence} shows the dependence of the rate on the coherence time $T$ for the $4 \times 4$ MIMO channel.  The normal approximation for $T=1, 20, 80$ is plotted.  From~\eqref{eq:mimo_capacity} and~\eqref{eq:mimo_disp_expression}, we know the capacity does not depend on $T$, but the dispersion depends on $T$ in an affine relationship.  Hence, from the dispersion we see that a larger coherence time reduces the maximum transmission rate when the other channel parameters are held fixed.  Intuitively, when the coherence time is lower, we are able to average over independent realizations of the fading coefficients in less channel uses.  Note that the CSIR assumption implies that we know the channel coefficient perfectly, which may be unrealistic at short coherence times for a practical channel.

\begin{figure}[ht]
\centering
\includegraphics[width=.8\textwidth]{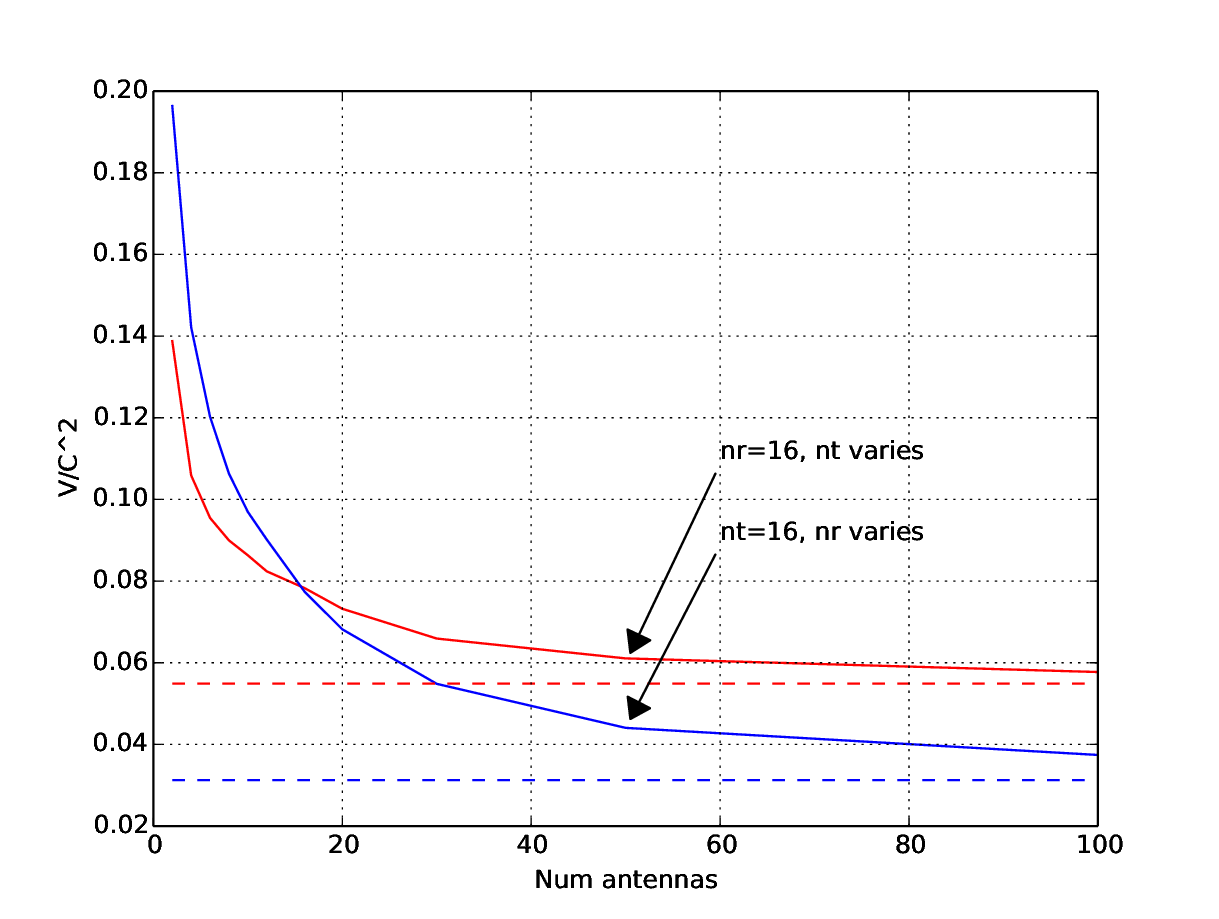}
\caption{Normalized dispersion $V \over C^2$ as a function of $n_r$ and $n_t$. The \underline{received} power is
$P_r=20dB$ and $T=16$. Dashed lines are asymptotic values from~\eqref{eq:pr_cap_nr}-\eqref{eq:pr_disp_nt}.}\label{fig:pr_norm_disp_antennas}
\end{figure}
\begin{figure}[ht]
\centering
\includegraphics[width=.8\textwidth]{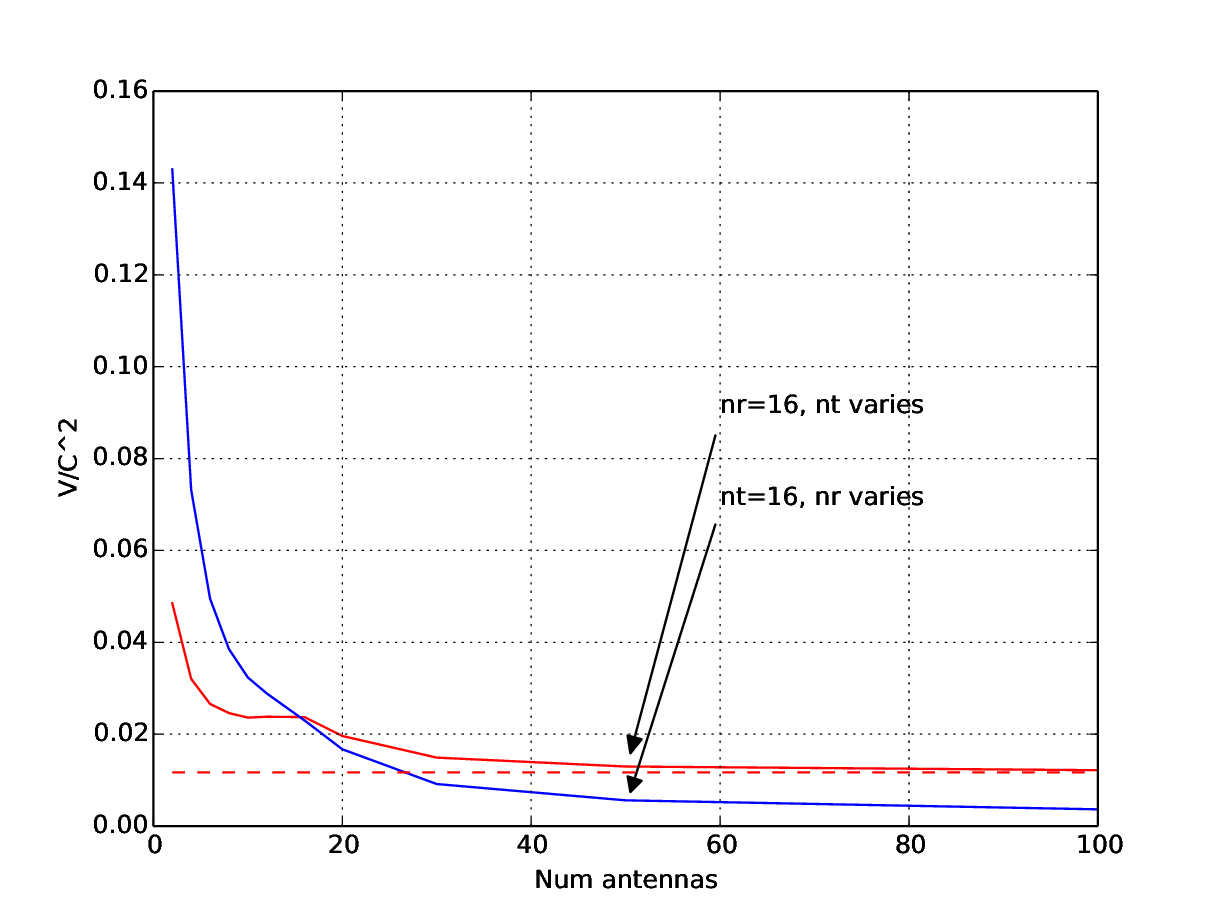}
\caption{Normalized dispersion $V \over C^2$ as a function of $n_r$ and $n_t$. The \underline{transmit} power is
$P=20dB$ and $T=16$. Dashed lines are asymptotic values from~\eqref{eq:p_cap_nr}-\eqref{eq:p_disp_nt}.}\label{fig:p_norm_disp_antennas}
\end{figure}

We now ask: how does the dispersion depend on the number of transmit and receive antennas?
Figures~\ref{fig:pr_norm_disp_antennas} and~\ref{fig:p_norm_disp_antennas} depict the normalized dispersion $V/C^2$,
cf.~\eqref{eq:minblock}, as a function of the number of antennas.  The fading process is chosen to be i.i.d. $\mathcal{N}(0,1)$.  Each plot has two
curves: one curve with $n_r$ fixed and $n_t$ growing, and the other curve with $n_t$ fixed and $n_r$ growing.  In both
plots, coherence time is $T=16$.  The difference is that on Fig.~\ref{fig:pr_norm_disp_antennas} the
\textit{received power} $P_r$ is held fixed (at 20 dB, i.e. $P$ is chosen so that $P_r=100$), whereas 
on Fig~\ref{fig:p_norm_disp_antennas} it is the transmit power $P$ that is held fixed (also at $20~dB$, i.e. $P=100$). 
The relation between $P_r$ and $P$ is as follows:
\begin{align}
P_r = \frac{P}{n_t}\E[\|H\|_F^2]\,,
\end{align}
These figures also display the asymptotic limiting values of $V\over C^2$ computed via random-matrix theory:
\begin{enumerate}
\item When $n_r$ is fixed and $n_t\to\infty$ under fixed \textit{received power $P_r$} we have
\begin{align}
C(P_r) &= \frac{n_r}{2} \log\left(1 + {P_r\over n_r}\right) + o(1) \label{eq:pr_cap_nr}\\
V(P_r) &= \log^2(e)\frac{P_r}{1 + {P_r\over n_r}} + o(1) \label{eq:pr_disp_nr}\ .
\end{align}
\item When $n_t$ is fixed and $n_r\to\infty$ under fixed \textit{received power $P_r$} we have
\begin{align}
C(P_r) &= \frac{n_t}{2} \log\left(1 + {P_r\over n_t}\right) + o(1) \label{eq:pr_cap_nt}\\
V(P_r) &= \log^2(e)\frac{P_r(2 + {P_r\over n_t})}{2(1 + {P_r\over n_t})^2} + o(1) \label{eq:pr_disp_nt}\,.
\end{align}
\item When $n_r$ is fixed and $n_t\to\infty$ under fixed \textit{transmitted power $P$} we have
\begin{align}
C(P) &= \frac{n_r}{2} \log\left(1 + P\right) + o(1) \label{eq:p_cap_nr}\\
V(P) &= \log^2(e)\frac{n_r P}{1 + P} + o(1) \label{eq:p_disp_nr}\,.
\end{align}
\item When $n_t$ is fixed and $n_r\to\infty$ under fixed \textit{transmit power $P$} we have
\begin{align}
C(P) &= \frac{n_t}{2} \log\left( 1 + {n_rP\over n_t}\right) + o(1) \label{eq:p_cap_nt}\\
V(P) &= \log^2(e)\frac{n_t}{2} + o(1)\label{eq:p_disp_nt}\,.
\end{align}
\end{enumerate}

Note that when the received power is fixed, \emph{reciprocity} holds: the capacity of the $n_t \times n_r$ channel
is the same as the capacity of the $n_r \times n_t$ one. Having information about dispersion, we may ask the more refined
question: although capacities of the channels are the same, which one has better dispersion (i.e. causes smaller coding
latency)?

From approximations~\eqref{eq:pr_disp_nr} and \eqref{eq:pr_disp_nt}, we can see that the channel dispersion is not symmetric in $n_t, n_r$. For example, in the setting of Fig.~\ref{fig:pr_norm_disp_antennas} we
see that the delay penalty in the $n_t \ll n_r$ regime is $58\%$ of the penalty in the $n_r \ll n_t$ regime.  Hence, in a two user channel, if user 1 has $n_1$ antennas and user 2 has $n_2 > n_1$ antennas, then the asymptotic analysis suggest that channel from user 1 to user 2 can support higher rates than the channel from user 2 to user 1 at finite blocklength.

Figure~\ref{fig:p_norm_disp_antennas} shows the scenario where the transmit power is fixed.  In this case, the capacity approaches a finite limit when $n_r$ is held fixed and $n_t \to \infty$, but grows logarithmically when $n_t$ is fixed and $n_r \to \infty$, as shown in equations~\eqref{eq:p_cap_nr} and~\eqref{eq:p_cap_nt}.  In this setting, the normalized dispersion approaches a finite limit when $n_r$ is fixed and $n_t \to \infty$, yet it vanishes when $n_t$ is fixed and $n_r\to\infty$.  Consequently in this regime, we can always choose the number of receive antennas $n_r$ large enough so that our system can achieve a given fraction of capacity $\eta$ using blocklength $n$.  The normalized dispersion in this case is proportional to $1/\log^2(n_r)$.

\appendices
\section{Existence of non-Gaussian caids}\label{apx:nongauss}

\begin{proposition}\label{prop:nongauss} Let $S\subset \mreals^n$ be such that a) $0\in S$ and b) there exists a non-zero 
polynomial in $n$ variables with real coefficients vanishing on $S$. Then there exists a random variable $X$ taking values in $\mreals^n$ with the property that its
characteristic function $\Psi(t)\eqdef \EE[e^{i\sum_{k=1}^n t_j X_j}], t\in\mreals^n$ satisfies
	$$ \Psi(t) = e^{-{\|t\|_2^2\over 2}}\qquad \forall t\in S $$
	but there exist a $t_0 \in \mathbb{R}^n$ such that $\Psi(t_0) \neq e^{-{\|t_0\|_2^2\over 2}}$ (i.e. $X\not\sim\matn(0,I_n)$).
\end{proposition}
% Macaulay2 script to generate q for 3x3 c/ex:
% P = QQ[a,b,c,d,e,f,g,h,z];
% P0 = QQ[a2,b2,c2,d2,e2,f2,g2,h2,z2];
% foo = map(P,P0, {a^2,b^2,c^2,d^2,e^2,f^2,g^2,h^2,z^2})
% I=ideal(a*e*z+b*f*g+c*d*h - c*e*g-b*d*z-a*f*h)
% J=preimage(foo, I);
\begin{remark} The simplest application of this proposition is the following. Suppose that three random vectors in
$\mreals^3$ have the property that projection onto any (2-dimensional) plane has the joint distribution
$\matn(0,I_2)\times\matn(0,I_2)\times \matn(0,I_2)$. Does it imply that the joint distribution of them is
$\matn(0,I_3)\times\matn(0,I_3)\times \matn(0,I_3)$? Note that it is easy to argue that joint distribution of any pair
of them is indeed $\matn(0,I_3)\times\matn(0,I_3)$ and thus the only \textit{jointly Gaussian} distribution that
satisfies the requirements is indeed the i.i.d. triplet. However, the above proposition shows that the general answer is
still negative. Here $S$ is a subset of all $\mreals^{3\times 3}$ with determinant zero.
\end{remark}
\begin{proof} We will slightly extend the argument of~\cite{hamedani1975determination}. We will assume familiarity with
basic commutatitive algebra on the level of~\cite{atiyah1969introduction}. Consider an identity expressing
the well-known computation of the Gaussian characteristic function:
$$ {1\over \sqrt{2\pi \alpha^2}} \int_{\mreals} e^{itx - {x^2\over 2\alpha^2}} = e^{-\alpha^2 {t^2\over 2}}\,.$$
Setting $\beta={1\over \alpha^2}$, changing sign of $t$ we get
$$ \int_{\mreals} e^{-itx - {\beta x^2\over 2}} \, dx = \sqrt{2\pi\over \beta} e^{-{t^2\over 2\beta }}\,.$$
Differentiating this in $\beta$ and setting $\beta={1\over 2}$ we get
$$ \int_{\mreals} x^{2k} e^{-itx - {x^2\over 4}} \, dx = p_{2k}(t) e^{-{t^2}} \,,$$
where $p_{2k}(t)$ is some polynomial of degree $2k$ with real coefficients (and involving only even powers of $t$). For
later convenience, we also interchange $t$ and $x$ to get
\begin{equation}\label{eq:pcx_1}
	\int_{\mreals} t^{2k} e^{-itx - {t^2\over 4}}\, dt = p_{2k}(x) e^{-{x^2}} \,.
\end{equation}
(Identity~\eqref{eq:pcx_1} also follows from the fact that Hermite polynomials times Gaussian density are
eigenfunctions of the Fourier transform.)

Next, suppose that there is a polynomial $q(t_1,\ldots,t_n)$ such that $q$ vanishes on $S$ and 
each monomial $t_1^{k_1}\cdots t_n^{k_n}$ in $q$ has all $k_1,\ldots,k_n$ even.
Then, define the characteristic function
\begin{equation}\label{eq:pcx_2}
	\Psi(t_1,\ldots,t_n) \eqdef e^{-{\sum_{k=1}^n t_k^2\over 2}}  + \epsilon e^{-{\sum_{k=1}^n t_k^2 \over 4}}
q(t_1,\ldots,t_n)\,.
\end{equation}
We will argue that for $\epsilon$ sufficiently small, $\Psi$ is a characteristic function of some (obviously
non-Gaussian) probability density function $f$ on $\mreals^n$. By taking the inverse Fourier transform we get that 
$$ f(x) = {1\over (2\pi)^{n\over 2}} e^{-{\|x\|_2^2\over 2}} (1 + \epsilon g(x))\,.$$
where $e^{-{\|x\|_2^2\over 2}}g(x)$ is the inverse Fourier transform of the second term in~\eqref{eq:pcx_2}. Since
$\Psi(t)$ is even in each
$t_j$, we conclude that $f(x)$ is real. Since $q(0)=0$ (recall that $0\in S$) we have $\Psi(0)=1$, and thus $\int_{\mreals^n} f = 1$. So to prove
that $f$ is a valid density function for small $\epsilon$ we only need to show that
\begin{equation}\label{eq:pcx_3}
	\sup_{x\in\mreals^n} |g(x)| < \infty\,.
\end{equation}
To that end, notice that applying~\eqref{eq:pcx_1} to each monomial $\prod t_j^{2k_j}$ we get
\begin{equation}\label{eq:pcx_1a}
	\int_{\mreals^n} \left(\prod_{j=1}^n t_j^{2k_j}\right) e^{-i \sum_j t_j x_j  - {\|t\|_2^2\over 4}}  \, dt_1
	\cdots dt_n= \left(\prod_j
	p_{2k_j}(x_j)\right) e^{-{\|x\|_2^2}} \,.
\end{equation}
Multiplying the right-hand side by $e^{{\|x\|_2^2\over 2}}$ we conclude that contribution of each monomial of $q$ to $\sup_x
|g(x)|$ is
bounded by 
	$$\sup_{x \in \mreals^n} \left|\left(\prod_j
	p_{2k_j}(x_j)\right) e^{-{\|x\|_2^2\over 2}}\right| < \infty\,. $$
Since there are finitely many monomials in $q$, the proof of~\eqref{eq:pcx_3} and of validity of $\Psi(t)$ is done.

We are left to argue that there must necessarily exist polynomial $q$ with required properties. By assumption there
exist some other polynomial $q_0$ vanishing on $S$. Consider an inclusion of rings
$$T \eqdef \mreals[x_1^2, x_2^2, \ldots, x_n^2] \hookrightarrow \mreals[x_1,\ldots,x_n]$$
where $\mathbb{R}[x_1,\hdots,x_n]$ denotes the ring of polynomials with variables $x_1,\hdots,x_n$ and coefficients in $\mathbb{R}$, and $\hookrightarrow$ denotes an inclusion map. This morphism of rings is obviously finite. Consider ideal $(q_0)$ of $\mreals[x_1,\ldots,x_n]$ and denote  as usual by
$(q_0)^c \eqdef (q_0)\cap T$ its contraction. We argue that $(q_0)^c \neq (0)$. Assume
otherwise, then we have $(q_0)^c = (0)$ and $\sqrt{(q_0)}^c = (0)$ (since $\sqrt{(0)}=(0)$ as $T$ is an integral domain). Take all minimal primes of $(q_0)$, call these $\{\mathfrak{p}_j\}$, then the radical of $(q_0)$ is the intersection of all prime ideals that contain it, i.e.
$\sqrt{(q_0)} = \cap_j \mathfrak{p}_j$. Then, denoting $\mathfrak{q}_j \eqdef \mathfrak{p}_j^c$ we get that $\cap_j
\mathfrak{q}_j = (0)$ in $T$. By ``prime-avoidance'', cf.~\cite[Prop. 1.11]{atiyah1969introduction}, we know $(0) \subset \cap_j
\mathfrak{q}_j$ implies that $\mathfrak{q}_j \subset (0)$ for some $j$, hence $\mathfrak{q}_j$ is the zero ideal for
some $j$. This contradicts the ``going-up theorem'', cf. ~\cite[Corollary 5.9]{atiyah1969introduction}, so we must have $(q_0)^c \neq (0)$, and hence we may take $q$ as an
arbitrary non-zero element of $(q_0)^c$.
\end{proof}

\section{Analysis of the Berry-Esseen constant}\label{apx:be}

\begin{proof}[Proof of Lemma~\ref{lem:thirdmom}]
We begin with upper bounding the numerator in~\eqref{eq:bndef}, i.e.
\begin{align}
\sum_{j=1}^n \EE[|W_j - \EE[W_j]|^3]\ .
\end{align}
The information density is given by
\begin{align}\label{eq:iq_x1}
i(x;y,h) = \frac{1}{2}\log\det\left(\Sigma\right) -                        \frac{1}{2}\sum_{j=1}^T \|y_j - hx_j\|^2 + \frac{1}{2}\text{tr}\left(y^T   \Sigma^{-1}y\right)
\end{align}
where
\begin{align}
\Sigma = I_{n_r} + \frac{P}{n_t}\mat{H}\mat{H}^T\,.
\end{align}
Define $W = i(x;Y,H)$ under the distribution $\mat{Y} = \mat{H}x + \mat{Z}$.  ~\eqref{eq:iq_x1}     reduces to
\begin{align}
W &=                                                    \frac{T}{2}\log\det\left(\Sigma\right) -                                   \frac{1}{2}\|\mat{Z}\|_F^2 + \frac{1}{2}\text{tr}\left(x^T\mat{H}^T        \Sigma^{-1}\mat{H}x + 2x^T\mat{H}^T\Sigma^{-1}\mat{Z} + \mat{Z}^T\Sigma^{- 1}\mat{Z}\right)\\
&= c(\mat{H},\mat{Z}) + \frac{1}{2}\text{tr}\left(x^T\mat{H}^T \Sigma^{-   1}\mat{H}x\right) + \text{tr}\left(x^T\mat{H}^T\Sigma^{-1}\mat{Z}\right)
\end{align}
where the scalar random variable 
\begin{align}
c(\mat{H},\mat{Z}) = \frac{T}{2}\log\det(\Sigma) - \frac{1}{2}\|Z\|_F^2 + \frac{1}{2}\text{tr}(Z^T \Sigma^{-1} Z)
\end{align}
is the sum of all    the terms that do not depend on $x$.  Note that
\begin{align}
&\E \text{tr}\left(x^T\mat{H}^T \Sigma^{-1}\mat{H}x\right) =               \text{tr}(x^T\E[\mat{H}^T\Sigma^{-1}\mat{H}]x)\\
&\E \text{tr}\left(x^T\mat{H}^T\Sigma^{-1}\mat{Z}\right) = 0\,.
\end{align}
Therefore, the ``centered'' information density is
\begin{align}
W - \E[W] &= c_0(\mat{H},\mat{Z}) - \E[c(\mat{H}, \mat{Z})] +
\frac{1}{2}\text{tr}\left(x^T\left(H^T\Sigma^{-1}\mat{H} -
\E[\mat{H}^T\Sigma^{-1}\mat{H}]\right)x\right) +                           \text{tr}\left(x^T\mat{H}^T\Sigma^{-1}\mat{Z}\right)\\
&= c_0(\mat{H}, \mat{Z}) + \text{tr}(x^T \mat{A}    x) + \text{tr}(x^T\mat{B})
\end{align}
where 
\begin{align}
A &= \frac{1}{2}(H^T\Sigma^{-1}\mat{H} - \E[\mat{H}^T\Sigma^{-1}\mat{H}])\\
B &= \mat{H}^T\Sigma^{-1}\mat{Z}\\
c_0(\mat{H},\mat{Z}) &= c(H,Z) - \E[c(H,Z)]\,.
\end{align}
Hence we can upper bound the centered third moment as
\begin{align}\label{eq:be100}
\E[|W - \E[W]|^3] \leq 
3 \underbrace{\E[|c_0(H,Z)|^3}_{S_1} + 3\underbrace{\E[|\text{tr}(x^T \mat{A} x)|^3]}_{S_2} + 3\underbrace{\E[|\text{tr}(x^T\mat{B})|^3]}_{S_3}\,.
\end{align}
We now proceed to upper bound each term individually.  First $S_2$,
%$S_1$
\begin{align}
S_2 &= \E[|\text{tr}(x^T \mat{A} x)|^3]\\
&= \frac{1}{8}\E[ | x^T H^T \Sigma^{-1} H x - x^T \E[H^T \Sigma^{-1} H] x|^3]\\
&\leq \frac{1}{8}\E[ | x^T H^T \Sigma^{-1} H x + x^T \E[H^T \Sigma^{-1} H]    x|^3]\label{eq:be111}\\
&\leq \frac{1}{8} \E\left[\left| \frac{2n_t}P \|x\|_F^2 \right|^3\right]\label{eq:be112}\\
&= \left(\frac{n_t}{P}\right)^3 \|x\|_F^6
\end{align}
where
\begin{itemize}
\item \eqref{eq:be111} follows since $H^T \Sigma^{-1} H$ is PSD, and $\E[H^T \Sigma^{-1} H]$ is also PSD as a non-negative combination of PSD matrices, so that both $x^T H^T \Sigma^{-1} H x$ and $x^T \E[H^T \Sigma^{-1}H]x$ are non-negative
\item \eqref{eq:be112} follows since $H^T \Sigma^{-1} H = V D V^T$ where
\begin{align}
%D = \text{diag}\left(c(\Lambda_1^2),\hdots, c(\Lambda_{n_{\min}}^2),  1,\hdots, 1\right)
D = \text{diag}\left(c(\Lambda_1^2),\hdots, c(\Lambda_{n_{\min}}^2), 0, \hdots, 0\right)
\end{align}
and $D \leq \frac{n_t}{P} I_{n_t}$ in the PSD ordering, so
\begin{align}
x^T H^T \Sigma^{-1} H x &\leq \frac{n_t}{P} x^T V V^T x = \frac{n_t}{P} \|x\|_F^2
\end{align}
and
\begin{align}
x^T \E[H^T \Sigma^{-1} H] x &\leq \frac{n_t}{P} x^T \E[V V^T] x = \frac{n_t}{P}    \|x\|_F^2\,.
\end{align}

Now we bound $S_3$ from~\eqref{eq:be100},
\end{itemize}
\begin{align}
S_3 &= \E[|\text{tr}(x^T B)|^3]\\
&= \E[\text{tr}(x^T H^T \Sigma^{-1} Z)|^3]\\
&= \E\left[ \left| \sum_{i=1}^{n_t} \sum_{j=1}^T \tilde{x}_{ij} Z_{ij} \frac{\Lambda_i}{1 + \frac{P}{n_t}\Lambda_i^2}\right|^3\right]  \label{eq:be121}\\
&\leq n_t^2 T^2 \sum_{i=1}^{n_t} \sum_{j=1}^T \E\left[ |\tilde{x}_{ij}|^3 |Z_{ij}|^3   \left|\frac{\Lambda_i}{1 + \frac{P}{n_t}\Lambda_i^2}\right|^3 \right]  \label{eq:be122}\\
&\leq \frac{n_t^2T^2}{4}\left(\frac{n_t}{P}\right)^{3/2}\|x\|_F^3\label{eq:be123}
\end{align}
where
\begin{itemize}
\item In \eqref{eq:be121}, define $\tilde{x} = V^T x$ and expand the trace.
\item \eqref{eq:be122} follows from the triangle inequality, along with $|\sum_{i=1}^n a_i|^3 \leq n^2 \sum_{i=1}^n |a_i|^3$.
\item \eqref{eq:be123} we have used $\E[|Z|^3] \leq 2$ for $Z \sim \mathcal{N}(0,1)$ along with the bound
\begin{align}
\left| \frac{x}{1 + ax^2}\right| \leq \frac{1}{2\sqrt{a}}\,.
\end{align}
Now notice that
\begin{align}
\sum_{i=1}^{n_t} \sum_{j=1}^T |\tilde{x}_{ij}|^3 \leq \left( \sum_{i=1}^{n_t} \sum_{i=1}^T \tilde{x}_{ij}^2\right)^{3/2}
\end{align}
which can be viewed as the norm inequality $\|a\|_3 \leq \|a\|_2$ for $a \in \mathbb{R}^{d}$.  Finally, we use $\|V^Tx\|_F^2 = \|x\|_F^2$ for any orthogonal matrix $V$.
\end{itemize}

For the denominator in~\eqref{eq:bndef}, the expression for $\frac{1}{T} \Var(W_j)$ is given in~\eqref{eq:cvx1}-\eqref{eq:cvx4}.  Note that the final term~\eqref{eq:cvx4} is non-negative, so we have the lower bound
\begin{align}
\sum_{j=1}^n \Var(W_j) &\geq K_1' n + K_2' \sum_{j=1}^n \left(\|x_j\|_F^2 - TP\right)^2\\
&\geq \max\left( nK_1', K_2' \sum_{j=1}^n \left(\|x_j\|_F^2 - TP\right)^2\right)\label{eq:varlbbe}
\end{align}
%where $K_1', K_2'$ are positive constants.
where
\begin{align}
K_1' &= T^2\Var\left(C_{r}(H,P)\right) + T\sum_{i=1}^{n_{\min}}\E\left[V_{AWGN}\left(\frac{P}{n_t}\Lambda_i^2\right)\right]\\
K_2' &= T\left(\frac{\|x\|_F^2}{n_t} - \frac{TP}{n_t}\right)^2\ .
\end{align}
Hence $K_1' > 0$ whenever $P > 0$.  Note that we use the assumption $\|x^n\|_F^2 = nTP$ freely here, as stated before.  The lower bound on the variance \eqref{eq:varlbbe}, we obtain the upper bound
\begin{align}\label{eq:bemaxden}
B_n(x^n) \leq \sqrt{n} \frac{\sum_{j=1}^n K_1 \|x_j\|_F^6 + K_2 \|x_j\|_F^3 + K_3}{\left(\max\left(nK_1', K_2' \sum_{j=1}^n \left(\|x_j\|_F^2 - TP\right)^2 \right) \right)^{3/2}}
\end{align}
where all constants are non-negative.  There are two cases based on which term achieves the max in the dominator.  First, suppose
\begin{align}
nK_1' \geq K_2' \sum_{j=1}^n \left(\|x_j\|_F^2 - TP\right)^2\ .
\end{align}
Expanding the square yields 
\begin{align}
K_2' \sum_{j=1}^n \|x_j\|_F^4 \leq nK_1' + nT^2P^2K_2'\ .
\end{align}
Thus the terms in the numerator are bounded by
\begin{align}
&\sum_{j=1}^n \|x_i\|_F^6 \leq \left(\max_{i=1}^n \|x_i\|_F^2\right) \sum_{j=1}^n \|x_i\|_F^4 \leq n^{3/2} \delta^2 (K_1' + T^2P^2K_2')\label{eq:bom1} \\
&\sum_{j=1}^n \|x_i\|_F^3 \leq n^{1/4} \sum_{j=1}^n \|x_i\|_F^4 \leq n^{5/4}(K_1' + T^2P^2K_2')
\end{align}
where \eqref{eq:bom1} uses the assumption $\|x_j\|_F \leq \delta n^{\frac{1}{4}}$. Applying this to $B_n$ in~\eqref{eq:bemaxden}, we see that in this case,
\begin{align}\label{eq:be_case_1}
%B_n(x^n) \leq \sqrt{n}\frac{n^{3/2} \delta^2 (K_1' + T^2P^2K_2') + n^{5/    4}(K_1' + T^2P^2K_2') + K_3}{n^{3/2}K_1'^{3/2}}
B_n(x^n) \leq \sqrt{n} \delta^2 C_1 + n^{1/4}C_2 + \frac{C_3}{n^{1/2}}
\end{align}
where the constant $C_1,C_2,C_3$ are non-negative constants. 

Now take the case when
\begin{align}
K_2' \sum_{j=1}^n \left(\|x_j\|_F^2 - TP\right)^2 \geq nK_1'\label{eq:bec201}\ .
\end{align}
Note that since $K_1' > 0$, in the case we must also have $K_2' > 0$ for the above inequality to hold. Let $a$ be defined as follows
\begin{align}
a = \frac{T^2P^2}{T^2P^2 + \frac{K_1'}{K_2'}}\ .
\end{align}
Here $a < 1$ since $K_1'/K_2' > 0$. Applying~\eqref{eq:bec201} yields
\begin{align}
a \sum_{j=1}^n\|x_j\|_F^4 &\geq a\left(n\frac{K_1'}{K_2'} + nT^2P^2\right)\\
&\geq nT^2P^2\ .\label{eq:bec202}
\end{align}
With this, from~\eqref{eq:bemaxden} we obtain the following upper bound
\begin{align}
B_n(x^n) &\leq \sqrt{n} \frac{\sum_{j=1}^n K_1 \|x_j\|_F^6 + K_2        \|x_j\|_F^3 + K_3}{K_2'^{3/2} \left( (1-a)\sum_{j=1}^n \|x_j\|_F^4 + a\sum_{j=1}^n \|x_j\|_F^4 - nT^2P^2  \right)^{3/2}}\\
&\leq \sqrt{n} \frac{\sum_{j=1}^n K_1 \|x_j\|_F^6 + K_2            \|x_j\|_F^3 + K_3}{K_2'^{3/2} \left( (1-a)\sum_{j=1}^n \|x_j\|_F^4 \right)^{3/2}}\label{eq:be200}\ .
\end{align}
where~\eqref{eq:be200} uses~\eqref{eq:bec202}.  Now, we can upper bound each term in \eqref{eq:be200} as
\begin{align}
\frac{K_1 \sum_{j=1}^n \|x_j\|_F^6}{K_2'^{3/2}\left( (1-a)\sum_{j=1}^n \|x_j\|_F^4 \right)^{3/2}} 
&\leq \frac{K_1 \max_{i=1,\hdots,n} \|x_i\|_F^2}{K_2'^{3/2}(1-a)^{3/2} \left( \sum_{j=1}^n \|x_j\|_F^4 \right)^{1/2} } \label{eq:be211}\\
&\leq \frac{K_1 \delta^{2} n^{1/2}}{n^{1/2} K_2'^{3/2} (1-a)^{3/2} (T^2P^2 + nK_1')^{1/2}}\\
\frac{K_2 \sum_{j=1}^n\|x_j\|_F^3}{K_2'^{3/2} \left( (1-a)\sum_{j=1}^n \|x_j\|_F^4\right)^{3/2}} &\leq \frac{K_2 n^{1/4}}{K_2'^{3/2} (1-a)^{3/2} \left( \sum_{j=1}^n \|x_j\|_F^4\right)^{3/4}} \label{eq:be210}\\
&\leq \frac{K_2 n^{1/4}}{n^{1/2} K_2'^{3/2}(1-a)^{3/2} (T^2P^2 +  nK_1')^{3/4}}\\
\frac{K_3}{K_2'^{3/2} \left( (1-a) \sum_{j=1}^n \|x_j\|_F^4 \right)^{3/2} } 
&\leq \frac{K_3}{n^{3/2} \left(K_2'(1-a)(T^2P^2 +               nK_1')\right)^{3/2}}\label{eq:be212}
\end{align}
where in~\eqref{eq:be210} we have used $\sum_{i=1}^n a_i^3 \leq n^{1/4} \left(\sum_{i=1}^n a_i^4\right)^{3/4}$ (easily obtained from p-norm inequalities), and both \eqref{eq:be211} and \eqref{eq:be212} use the assumption $\|x_j\|_F \leq \delta n^{\frac{1}{4}}$.  Using these bounds in~\eqref{eq:be200}, we obtain
\begin{align}\label{eq:be_case_2}
B_n(x^n) \leq \sqrt{n}\delta^2 C_1' + n^{1/4}C_2' + \frac{C_3'}{n^{1/2}}
\end{align}
where $C_1',C_2',C_3'$ are non-negative constants.

From~\eqref{eq:be_case_1} and \eqref{eq:be_case_2}, we conclude that
\begin{align}
B_n(x^n) \leq \sqrt{n}\delta^2 C_1'' + n^{1/4}C_2'' + \frac{C_3''}{n^{1/2}}\,.
\end{align}
\end{proof}

%%%%%%%%%%%%%%%%%%%%%%%%%%%%%%%%%%%%%%%%%%%%%%%%%%%%%%%%%%%%%%%%%%
%%%%% BIBLIOGRAPHY %%%%%%%%%%%%%%%%%%%%%%%%%%%%%%%%%%%%%%%%%%%%%%%
%%%%%%%%%%%%%%%%%%%%%%%%%%%%%%%%%%%%%%%%%%%%%%%%%%%%%%%%%%%%%%%%%%

\bibliographystyle{IEEEtran}
\bibliography{IEEEabrv,../../reports}

\end{document}